\documentclass[english,lineno]{article}
\usepackage[utf8x]{inputenc}
\usepackage{geometry}
\usepackage{color}
\usepackage{babel}
\usepackage{float}
\usepackage{amsthm}
\usepackage{amsmath}
\usepackage{amssymb}
\usepackage{graphicx}
\usepackage[authoryear]{natbib}
\usepackage[unicode=true,pdfusetitle,
 bookmarks=true,bookmarksnumbered=false,bookmarksopen=false,
 breaklinks=false,pdfborder={0 0 0},backref=false,colorlinks=true]
 {hyperref}
\hypersetup{
 citecolor=blue,urlcolor=blue}

\makeatletter

\floatstyle{ruled}
\newfloat{algorithm}{tbp}{loa}
\providecommand{\algorithmname}{Algorithm}
\floatname{algorithm}{\protect\algorithmname}

  \theoremstyle{definition}
  \newtheorem{condition}{\protect\conditionname}
  \theoremstyle{remark}
  \newtheorem{rem}{\protect\remarkname}
\theoremstyle{plain}
\newtheorem{thm}{\protect\theoremname}
  \theoremstyle{plain}
  \newtheorem{prop}{\protect\propositionname}
  \theoremstyle{plain}
  \newtheorem{cor}{\protect\corollaryname}
  \theoremstyle{plain}
  \newtheorem{lem}{\protect\lemmaname}
 \theoremstyle{definition}
  \newtheorem{example}{\protect\examplename}

\usepackage{times}
\usepackage{bm}

\@ifundefined{showcaptionsetup}{}{%
 \PassOptionsToPackage{caption=false}{subfig}}
\usepackage{subfig}
\makeatother

  \providecommand{\conditionname}{Condition}
  \providecommand{\examplename}{Example}
  \providecommand{\lemmaname}{Lemma}
  \providecommand{\propositionname}{Proposition}
  \providecommand{\remarkname}{Remark}
\providecommand{\corollaryname}{Corollary}
\providecommand{\theoremname}{Theorem}

\begin{document}

\title{Variance bounding and geometric ergodicity of Markov chain Monte
Carlo kernels for approximate Bayesian computation}

\author{Anthony Lee and Krzysztof Łatuszyński\\
Department of Statistics,\\
University of Warwick}
\maketitle
\begin{abstract}
Approximate Bayesian computation has emerged as a standard computational
tool when dealing with the increasingly common scenario of completely
intractable likelihood functions in Bayesian inference. We show that
many common Markov chain Monte Carlo kernels used to facilitate inference
in this setting can fail to be variance bounding, and hence geometrically
ergodic, which can have consequences for the reliability of estimates
in practice. This phenomenon is typically independent of the choice
of tolerance in the approximation. We then prove that a recently introduced
Markov kernel in this setting can inherit variance bounding and geometric
ergodicity from its intractable Metropolis--Hastings counterpart,
under reasonably weak and manageable conditions. We show that the
computational cost of this alternative kernel is bounded whenever
the prior is proper, and present indicative results on an example
where spectral gaps and asymptotic variances can be computed, as well
as an example involving inference for a partially and discretely observed,
time-homogeneous, pure jump Markov process. We also supply two general
theorems, one of which provides a simple sufficient condition for
lack of variance bounding for reversible kernels and the other provides
a positive result concerning inheritance of variance bounding and
geometric ergodicity for mixtures of reversible kernels.
\end{abstract}

\section{Introduction}

Approximate Bayesian computation refers to branch of Monte Carlo methodology
that uses the ability to simulate data according to a parametrized
likelihood function in lieu of computation of that likelihood to perform
approximate, parametric Bayesian inference. These methods have been
used in an increasingly diverse range of applications since their
inception in the context of population genetics \citep{Tavare1997,Pritchard1999},
particularly in cases where the likelihood function is either impossible
or computationally prohibitive to evaluate.

We are in a standard Bayesian setting with data $y\in\mathsf{Y}$,
a parameter space $\Theta$, a prior $p:\Theta\rightarrow\mathbb{R}_{+}$
and for each $\theta\in\Theta$ a likelihood $f_{\theta}:\mathsf{Y}\rightarrow\mathbb{R}_{+}$.
We assume $\mathsf{Y}$ is a metric space and consider the artificial
likelihood
\begin{equation}
f_{\theta}^{\epsilon}(y)=V(\epsilon)^{-1}\int_{\mathsf{Y}}I\left(y\in B_{\epsilon,x}\right)f_{\theta}(x)\mathrm{d}x=V(\epsilon)^{-1}f_{\theta}\left(B_{\epsilon,y}\right),\label{eq:abc_approx_likelihood}
\end{equation}
which is commonly employed in approximate Bayesian computation. The
value of $\epsilon$ can be interpreted as the tolerance of the approximation.
Here, $B_{r,z}$ denotes a metric ball of radius $r$ around $z$,
$V(r)=\int_{\mathsf{Y}}I\left(x\in B_{r,0}\right)\mathrm{d}x$ denotes
the volume of a ball of radius $r$ in $\mathcal{\mathsf{Y}}$ and
$I$ denotes the indicator function. We slightly abuse language by
referring to densities as distributions, and where convenient, employ
the measure-theoretic notation $\mu(A)=\int_{A}\mu(d\lambda)$. We
consider situations in which both $\epsilon$ and $y$ are fixed,
and so define functions $h:\Theta\rightarrow[0,1]$ and $w:\mathsf{Y}\rightarrow[0,1]$
by 
\begin{equation}
h(\theta)=f_{\theta}\left(B_{\epsilon,y}\right)\label{eq:h_defn}
\end{equation}
and $w(x)=I\left(y\in B_{\epsilon,x}\right)$ to simplify the presentation.
The value $h(\theta)$ can be interpreted as the probability of `hitting'
$B_{\epsilon,y}$ with a sample drawn from $f_{\theta}$.

While the artificial likelihood \eqref{eq:abc_approx_likelihood}
is also intractable in general, the approximate posterior it induces,
$\pi(\theta)=h(\theta)p(\theta)/\int_{\Theta}h(\vartheta)p(\vartheta){\rm d}\vartheta$,
can be dealt with using constrained versions of standard methods when
sampling from $f_{\theta}$ is possible for any $\theta\in\Theta$
\citep[see, e.g., ][]{Marin2012}. In particular, one typically uses
$f_{\theta}$ as a proposal in such a way that its explicit computation
is avoided. We are often interested in computing $\pi(\varphi)=\int_{\Theta}\varphi(\theta)\pi(\theta){\rm d\theta}$,
the posterior expectation of some function $\varphi$, and it is this
type of quantity that can be approximated using Monte Carlo methodology.
We focus on one such method, Markov chain Monte Carlo, whereby a Markov
chain is constructed by sampling iteratively from an irreducible Markov
kernel $P$ with unique stationary distribution $\pi$. We can use
such a chain directly to estimate $\pi(\varphi)$ using appropriately
normalized partial sums, i.e.,\ given the realization $\theta_{1},\theta_{2},\ldots$
of a chain started at $\theta_{0}$, where $\theta_{i}\sim P(\theta_{i-1},\cdot)$
for $i\in\mathbb{N}$ we compute the estimate
\begin{equation}
\frac{1}{m}\sum_{i=1}^{m}\varphi(\theta_{i}),\label{eq:mcmc_estimate}
\end{equation}
for some $m$. Alternatively, the Markov kernels can be used within
other methods such as sequential Monte Carlo \citep{DelMoral2006}.
In the former case, it is desirable that a central limit theorem holds
for \eqref{eq:mcmc_estimate} and that the asymptotic variance ${\rm var}(P,\varphi)$
of \eqref{eq:mcmc_estimate} be reasonably small, while in the latter
it is desirable that the kernel be geometrically ergodic, i.e.,\ $P^{m}(\theta_{0},\cdot)$
converges at a geometric rate in $m$ to $\pi$ in total variation
where $P^{m}$ is the $m$-fold iterate of $P$ \citep[see, e.g., ][]{roberts2004general,MT2009},
at least because this property is often assumed in analyses \citep[see, e.g.,][]{Jasra2008,Whiteley2012}.
In addition, consistent estimation of ${\rm var}(P,\varphi)$ is well
established \citep{hobert2002applicability,jones2006fixed,bednorz2007few,flegal2010batch}
for geometrically ergodic chains.

Motivated by these considerations, we study both the variance bounding
\citep{roberts2008variance} and geometric ergodicity properties of
a number of reversible kernels used for approximate Bayesian computation.
For reversible $P$, a central limit theorem holds for all $\varphi\in L^{2}(\pi)$
if and only if $P$ is variance bounding \citep[Theorem~7]{roberts2008variance},
where $L^{2}(\pi)$ is the space of square-integrable functions with
respect to $\pi$. Of course, reversible kernels that are not variance
bounding can still produce Markov chains where \eqref{eq:mcmc_estimate}
satisfies a central limit theorem for some, but not all, functions
in $L^{2}(\pi)$.

Much of the literature seeks to control the trade-off associated with
the quality of approximation \eqref{eq:abc_approx_likelihood}, controlled
by $\epsilon$ and manipulation of $y$, and counteracting computational
difficulties \citep[see, e.g.,][]{Fearnhead2012}. We address here
a separate issue, namely that many Markov kernels used in this context
are neither variance bounding nor geometrically ergodic, for any finite
$\epsilon$ in rather general situations when using `local' proposal
distributions.

As a partial remedy to the problems identified by this negative result,
we also show that under reasonably mild conditions, a kernel proposed
in \citet{Lee2012} can inherit variance bounding and geometric ergodicity
from its intractable Metropolis--Hastings \citep{Metropolis1953,Hastings1970}
counterpart. This allows for the specification of a broad class of
models for which we can be assured this particular kernel will be
geometrically ergodic. In addition, conditions ensuring inheritance
of either property can be met without knowledge of $f_{\theta}$,
e.g.\ by using a symmetric proposal and a prior that is continuous,
everywhere positive and has exponential or heavier tails.

To assist in the interpretation of results and the quantitative example
in the discussion, we provide some background on the spectral properties
of variance bounding and geometrically ergodic Markov kernels. Both
variance bounding and geometric ergodicity of a reversible Markov
kernel $P$ are related to $\sigma_{0}(P)$, the spectrum of $P$
considered as an operator on $L_{0}^{2}(\pi)$, the restriction of
$L^{2}(\pi)$ to zero-mean functions \citep[see, e.g.,][]{geyer2000non,mira2001ordering}.
Variance bounding is equivalent to $\sup\sigma_{0}(P)<1$ \citep[Theorem~14]{roberts2008variance}
and geometric ergodicity is equivalent to $\sup|\sigma_{0}(P)|<1$
\citep[Theorem~2.1]{kontoyiannis2009geometric,roberts1997geometric}.
The spectral gap ${\rm Gap}(P)=1-\sup|\sigma_{0}(P)|$ of a geometrically
ergodic kernel is closely related to its aforementioned geometric
rate of convergence to $\pi$, with faster rates associated with larger
spectral gaps. In particular, its convergence in total variation satisfies
for some $1>\rho\geq\sup|\sigma_{0}(P)|$ and some function $C_{\rho}:\Theta\rightarrow\mathbb{R}_{+}$
\citep[c.f. ][Section 6]{baxendale2005renewal} 
\begin{equation}
\Vert\pi(\cdot)-P^{m}(\theta_{0},\cdot)\Vert_{{\rm TV}}\leq C_{\rho}(\theta_{0})\rho^{m}.\label{eq:ge_tv}
\end{equation}

\section{The Markov kernels\label{sec:Algorithms}}

In this section we describe the algorithmic specification of the $\pi$-invariant
Markov kernels under study. The algorithms specify how to sample from
each kernel; in each, a candidate $\vartheta$ is proposed according
to a common proposal $q(\theta,\cdot)$ and accepted or rejected,
possibly along with other auxiliary variables, using simulations from
the likelihoods $f_{\vartheta}$ and $f_{\theta}$. We assume that
for all $\theta\in\Theta$, $q(\theta,\cdot)$ and $p$ are densities
with respect to a common dominating measure, e.g.\ the Lebesgue or
counting measures.

The first and simplest Markov kernel in this setting was proposed
in \citet{Marjoram2003}, and is a special case of a `pseudo-marginal'
kernel \citep{Beaumont2003,Andrieu2009}. Such kernels have been used
in the context of approximate Bayesian computation for the estimation
of parameters in speciation models \citep{becquet2007new,chen2009new,li2010demographic,kim2011patterns},
and as a methodological component within an SMC sampler \citep{DelMoral2011,drovandi2011estimation}.
They evolve on $\Theta\times\mathsf{Y}^{N}$ and involve sampling
auxiliary variables $z_{1:N}\sim f_{\vartheta}^{\otimes N}$ for a
fixed $N\in\mathbb{N}$. We denote kernels of this type for any $N$
by $P_{\ref{alg:pm_kernel_1},N}$, and describe their simulation in
Algorithm~\ref{alg:pm_kernel_1}. It is readily verified \citep{Beaumont2003,Andrieu2009}
that $P_{\ref{alg:pm_kernel_1},N}$ is reversible with respect to
\[
\bar{\pi}(\theta,x_{1:N})\propto p(\theta)\prod_{j=1}^{N}f_{\theta}(x_{j})\frac{1}{N}\sum_{j=1}^{N}w(x_{j}),
\]
and we have $\bar{\pi}(\theta)=\int\bar{\pi}(\theta,x_{1:N}){\rm d}x_{1:N}=\pi(\theta)$,
i.e.,\ the $\theta$-marginal of $\bar{\pi}$ is $\pi(\theta).$

\begin{algorithm}[H]

\caption{To sample from $P_{\ref{alg:pm_kernel_1},N}(\theta,x_{1:N};\cdot)$\label{alg:pm_kernel_1}}

\begin{enumerate}
\item Sample $\vartheta\sim q(\theta,\cdot)$ and $z_{1:N}\sim f_{\vartheta}^{\otimes N}$. 
\item With probability 
\[
1\wedge\frac{p(\vartheta)q(\vartheta,\theta)\sum_{j=1}^{N}w(z_{j})}{p(\theta)q(\theta,\vartheta)\sum_{j=1}^{N}w(x_{j})},
\]
output $(\vartheta,z_{1:N})$. Otherwise, output $(\theta,x_{1:N})$.
\end{enumerate}
\end{algorithm}

In \citet{Lee2012}, two alternative kernels were proposed in this
context, both of which evolve on $\Theta$. One, denoted $P_{\ref{alg:pm_kernel_2},N}$
and described in Algorithm~\ref{alg:pm_kernel_2}, is an alternative
pseudo-marginal kernel that in addition to sampling $z_{1:N}\sim f_{\vartheta}^{\otimes N}$,
also samples auxiliary variables $x_{1:N-1}\sim f_{\theta}^{\otimes N-1}$.
Detailed balance can be verified directly upon interpreting $\sum_{j=1}^{N}w(z_{j})$
and $\sum_{j=1}^{N-1}w(x_{j})$ as Binomial$\left\{ N,h(\vartheta)\right\} $
and Binomial$\left\{ N-1,h(\theta)\right\} $ random variables respectively.
The other kernel, denoted $P_{\ref{alg:1_hit_kernel}}$ and described
in Algorithm~\ref{alg:1_hit_kernel}, also involves sampling according
to $f_{\theta}$ and $f_{\vartheta}$ but does not sample a fixed
number of auxiliary variables. This kernel also satisfies detailed
balance \citep[Proposition 1]{Lee2012a}.

\begin{algorithm}[H]
\caption{To sample from $P_{\ref{alg:pm_kernel_2},N}(\theta,\cdot)$\label{alg:pm_kernel_2}}

\begin{enumerate}
\item Sample $\vartheta\sim q(\theta,\cdot)$, $x_{1:N-1}\sim f_{\theta}^{\otimes N-1}$
and $z_{1:N}\sim f_{\vartheta}^{\otimes N}$.
\item With probability 
\[
1\wedge\frac{p(\vartheta)q(\vartheta,\theta)\sum_{j=1}^{N}w(z_{j})}{p(\theta)q(\theta,\vartheta)\left\{ 1+\sum_{j=1}^{N-1}w(x_{j})\right\} },
\]
output $\vartheta$. Otherwise, output $\theta$.\end{enumerate}
\end{algorithm}

\begin{algorithm}[H]
\caption{To sample from $P_{\ref{alg:1_hit_kernel}}(\theta,\cdot)$\label{alg:1_hit_kernel}}

\begin{enumerate}
\item Sample $\vartheta\sim q(\theta,\cdot)$.
\item With probability 
\[
1-\left\{ 1\wedge\frac{p(\vartheta)q(\vartheta,\theta)}{p(\theta)q(\theta,\vartheta)}\right\} ,
\]
stop and output $\theta$.
\item For $i=1,2,\ldots$ until $\sum_{j=1}^{i}w(z_{j})+w(x_{j})\geq1$,
sample $x_{i}\sim f_{\theta}$ and $z_{i}\sim f_{\vartheta}$. Set
$N\leftarrow i$. 
\item If $w(z_{N})=1$, output $\vartheta$. Otherwise, output $\theta$.\end{enumerate}
\end{algorithm}

Our first results in Section~\ref{sec:Theoretical-properties} concern
$P_{\ref{alg:pm_kernel_1},N}$ and $P_{\ref{alg:pm_kernel_2},N}$.
One typically expects better performance from these kernels for larger
values of $N$ \citep[see, e.g, ][]{Andrieu2012}, and such behaviour
can often be demonstrated empirically. However, we establish that
both of these kernels can nevertheless fail to be variance bounding
regardless of the value of $N$ when $q$ proposes moves locally.
This suggests that increasing $N$ may only bring an improvement up
to a certain point. On the other hand, subsequent results for $P_{\ref{alg:1_hit_kernel}}$
show that by expending more computational effort in particular places
one can successfully inherit variance bounding and/or geometric ergodicity
from $P_{{\rm MH}}$, the Metropolis--Hastings kernel with proposal
$q$.

Because many of our positive results for $P_{\ref{alg:1_hit_kernel}}$
are in relation to $P_{{\rm MH}}$, we provide the algorithmic specification
for sampling from $P_{{\rm MH}}$ in Algorithm~\ref{alg:mh_kernel}.
In the approximate Bayesian computation setting, use of $P_{{\rm MH}}$
is ruled out by assumption since $h$ cannot be computed. However,
the preceding kernels are all, in some sense, exact approximations
of $P_{{\rm MH}}$.

\begin{algorithm}[H]
\caption{To sample from $P_{{\rm MH}}(\theta,\cdot)$\label{alg:mh_kernel}}

\begin{enumerate}
\item Sample $\vartheta\sim q(\theta,\cdot)$.
\item With probability 
\[
1\wedge\frac{p(\vartheta)h(\vartheta)q(\vartheta,\theta)}{p(\theta)h(\theta)q(\theta,\vartheta)},
\]
output $\vartheta$. Otherwise, output $\theta$.\end{enumerate}
\end{algorithm}

The kernels share a similar structure, and $P_{\ref{alg:pm_kernel_2},N}$,
$P_{\ref{alg:1_hit_kernel}}$ and $P_{{\rm MH}}$ can each be written
as
\begin{equation}
P(\theta,{\rm d}\vartheta)=q(\theta,{\rm d\vartheta)}\alpha(\theta,\vartheta)+\left\{ 1-\int_{\Theta}q(\theta,{\rm d}\theta')\alpha(\theta,\theta')\right\} \delta_{\theta}({\rm d}\vartheta),\label{eq:alpha_kernel}
\end{equation}
where only the acceptance probability $\alpha(\theta,\vartheta)$
differs. $P_{\ref{alg:pm_kernel_1},N}$ can be represented similarly,
with modifications to account for its evolution on the extended space
$\Theta\times\mathsf{Y}^{N}$. The representation \eqref{eq:alpha_kernel}
is used extensively in our analysis, and we have for $P_{\ref{alg:pm_kernel_2},N}$,
$P_{\ref{alg:1_hit_kernel}}$ and $P_{{\rm MH}}$, respectively 
\begin{align}
\alpha_{\ref{alg:pm_kernel_2},N}(\theta,\vartheta) & =\int_{\mathsf{Y}^{N}}\int_{\mathsf{Y}^{N-1}}\left[1\wedge\frac{c(\vartheta,\theta)\sum_{j=1}^{N}w(z_{j})}{c(\theta,\vartheta)\left\{ 1+\sum_{j=1}^{N-1}w(x_{j})\right\} }\right]f_{\theta}^{\otimes N-1}({\rm d}x_{1:N-1})f_{\vartheta}^{\otimes N}({\rm d}z_{1:N}),\label{eq:alpha_pm2}\\
\alpha_{\ref{alg:1_hit_kernel}}(\theta,\vartheta) & =\left\{ 1\wedge\frac{c(\vartheta,\theta)}{c(\theta,\vartheta)}\right\} \frac{h(\vartheta)}{h(\theta)+h(\vartheta)-h(\theta)h(\vartheta)},\label{eq:alpha_1hit}\\
\alpha_{{\rm MH}}(\theta,\vartheta) & =1\wedge\frac{c(\vartheta,\theta)h(\vartheta)}{c(\theta,\vartheta)h(\theta)},\label{eq:alpha_mh}
\end{align}
where $c(\theta,\vartheta)=p(\theta)q(\theta,\vartheta)$ and \eqref{eq:alpha_1hit}
is obtained, e.g.,\ in \citet{Lee2012a}. Finally, we reiterate that
all the kernels satisfy detailed balance and are therefore reversible.

\section{Theoretical properties\label{sec:Theoretical-properties}}

We assume that $\Theta$ is a metric space, and that 
\begin{equation}
H=\int_{\Theta}p(\theta)h(\theta)\mathrm{d}\theta\label{eq:H}
\end{equation}
satisfies $H\in(0,\infty)$ so $\pi$ is well defined. We allow $p$
to be improper, i.e.,\ for $\int_{\Theta}p(\theta){\rm d}\theta$
to be infinite but when it is proper we assume it is normalized so
$\int_{\Theta}p(\theta){\rm d}\theta=1$. We define the collection
of local proposals as 
\begin{equation}
\mathcal{Q}=\left\{ q\::\:\text{for all }\delta>0,\:\text{there exists }r\in(0,\infty)\:\text{such that for all }\theta\in\Theta,\: q\left(\theta,B_{r,\theta}^{\complement}\right)<\delta\right\} ,\label{eq:defQ}
\end{equation}
which encompasses a broad number of common choices in practice, e.g.,\ $q$
being a random walk. This corresponds to the\emph{ }tightness of centred
proposals $q$.

We denote by $\mathcal{V}$ and $\mathcal{G}$ the collections of
reversible kernels that are respectively variance bounding \citep{roberts2008variance}
and geometrically ergodic \citep[see, e.g., ][]{roberts2004general,MT2009},
noting that $\mathcal{G}\subset\mathcal{V}$. In our analysis, we
make use of the following conditions.
\begin{condition}
\label{cond:c1}The proposal $q$ is a member of $\mathcal{Q}$. In
addition, $\pi\left(B_{r,0}^{\complement}\right)>0$ for all $r>0$
but $\lim_{v\rightarrow\infty}\sup_{\theta\in B_{v,0}^{\complement}}h(\theta)=0$.
\end{condition}

\begin{condition}
a\label{cond:c2}The proposal $q$ is a member of $\mathcal{Q}$.
In addition, for all $K>0$, there exists an $M_{K}\in[1,\infty)$
such that for all $(\theta,\vartheta)$ in the set
\[
\left\{ (\theta,\vartheta)\in\Theta^{2}\::\:\vartheta\in B_{K,\theta}\text{ and }\pi(\theta)q(\theta,\vartheta)\wedge\pi(\vartheta)q(\vartheta,\theta)>0\right\} ,
\]
either $h(\vartheta)/h(\theta)\in[M_{K}^{-1},M_{K}]$ or $c(\vartheta,\theta)/c(\theta,\vartheta)\in[M_{K}^{-1},M_{K}]$.
\end{condition}
Condition~\ref{cond:c1} ensures that the posterior has mass arbitrarily
far from $0$ but that $h(\theta)$ gets arbitrarily small as we move
away from some compact set in $\Theta$, while Condition~\ref{cond:c2}
constrains the interplay between the likelihood and the prior-proposal
pair. For example, it is satisfied for symmetric $q$ when $p$ is
continuous, everywhere positive with exponential or heavier tails,
or alternatively, if the likelihood is continuous, everywhere positive
and decays at most exponentially fast. Conditions~\ref{cond:c1}
and~\ref{cond:c2} are not mutually exclusive.
\begin{rem}
\label{rem:global_remark}A global variant of Condition~\ref{cond:c2}
can be defined where $q$ need not be a member of $\mathcal{Q}$,
but there exists an $M\in[1,\infty)$ such that for all $(\theta,\vartheta)$
in the set $\left\{ (\theta,\vartheta)\in\Theta^{2}\::\:\pi(\theta)q(\theta,\vartheta)\wedge\pi(\vartheta)q(\vartheta,\theta)>0\right\} $,
either $h(\vartheta)/h(\theta)\in[M^{-1},M]$ or $c(\vartheta,\theta)/c(\theta,\vartheta)\in[M^{-1},M]$.
Theorems~\ref{thm:vb}--\ref{thm:ge}, which hold under Condition~\ref{cond:c2},
also hold under this variant, with simplified proofs that are omitted.
\end{rem}
We first provide a general theorem that supplements \citet[Theorem~5.1]{Roberts1996}
for reversible kernels, indicating that lack of geometric ergodicity
due to arbitrarily `sticky' states coincides with lack of variance
bounding. All proofs are housed in Appendix~\ref{sec:Proofs_main}.
\begin{thm}
\label{thm:gen_nvb}For any $\nu$ not concentrated at a single point
and any reversible, irreducible, $\nu$-invariant Markov kernel $P$,
such that $P(\theta,\{\theta\})$ is a measurable function, if $\nu-{\rm ess}\sup_{\theta}P(\theta,\{\theta\})=1$
then $P$ is not variance bounding.
\end{thm}
Our first result concerning the kernels under study is negative, and
indicates that performance of $P_{\ref{alg:pm_kernel_1},N}$ and $P_{\ref{alg:pm_kernel_2},N}$
under Condition~\ref{cond:c1} can be poor, irrespective of the value
of $N$. 
\begin{thm}
\label{thm:nvb}Under Condition~\ref{cond:c1}, $P_{\ref{alg:pm_kernel_1},N}\notin\mathcal{V}$
and $P_{\ref{alg:pm_kernel_2},N}\notin\mathcal{V}$ for all $N\in\mathbb{N}$. \end{thm}
\begin{rem}
Theorem~\ref{thm:nvb} immediately implies that under Condition~\ref{cond:c1},
$P_{\ref{alg:pm_kernel_1},N}\notin\mathcal{G}$ and $P_{\ref{alg:pm_kernel_2},N}\notin\mathcal{G}$
by \citet[Theorem~1]{roberts2008variance}. The former implication
is not covered by \citet[Theorem~8]{Andrieu2009} or \citet[Propositions~9 or~12]{Andrieu2012}
because what they term weights in this context, $w(x)/h(\theta)$,
are upper bounded by $h(\theta)^{-1}$ for $\pi$-almost every $\theta\in\Theta$
and $f_{\theta}$-almost every $x\in\mathsf{Y}$ but are not uniformly
bounded in $\theta$.
\end{rem}
We emphasize that the choice of $q$ is crucial to establishing Theorem
\ref{thm:nvb}. Since $H>0$, if $q(\theta,\vartheta)=g(\vartheta)$,
e.g.,\ and $\sup_{\theta}p(\theta)/g(\theta)<\infty$ then by \citet[Theorem~2.1]{mengersen1996rates},
$P_{\ref{alg:pm_kernel_1},N}$ is uniformly ergodic and hence in $\mathcal{G}$.
Uniform ergodicity, however, does little to motivate the use of an
independent proposal in challenging scenarios, particularly when $\Theta$
is high dimensional.
\begin{rem}
We observe from \eqref{eq:h_defn} that when $\lim_{v\rightarrow\infty}\sup_{\theta\in B_{v,0}^{\complement}}h(\theta)=0$
holds for a given $\epsilon=\epsilon_{0}$, this implies that it holds
for all $\epsilon\in(0,\epsilon_{0}]$. Furthermore, often this condition
holds because $\lim_{v\rightarrow\infty}\sup_{\theta\in B_{v,0}^{\complement}}f_{\theta}(C)=0$
for any compact subset $C$ of $\mathsf{Y}$. In such cases, $\lim_{v\rightarrow\infty}\sup_{\theta\in B_{v,0}^{\complement}}h(\theta)=0$
for any finite $\epsilon>0$ and Theorem~\ref{thm:nvb} will correspondingly
hold for any finite $\epsilon>0$ such that $\pi\left(B_{r,0}^{\complement}\right)>0$
for all $r>0$.
\end{rem}
Our negative result is not exclusive to the particular approximate
Bayesian computation setup considered here. In Appendix \ref{sec:more_neg_results}
we provide supplementary results to indicate that the results can
be extended to the use of autoregressive proposals not covered by
$\mathcal{Q}$, approximations of the likelihood of a more general
form than \eqref{eq:abc_approx_likelihood} and Markov kernels with
an invariant distribution in which $\epsilon$ is a non-degenerate
auxiliary variable, as such cases do arise in practice \citep[see, e.g., ][]{bortot2007inference,sisson2010likelihood}.
However, the following results do not apply to these alternative settings,
since $P_{\ref{alg:1_hit_kernel}}$ lacks an obvious analogue when
the artificial likelihood is not given by \eqref{eq:abc_approx_likelihood}.

Our next three results concern $P_{\ref{alg:1_hit_kernel}}$, and
demonstrate first that variance bounding of $P_{{\rm MH}}$ is a necessary
condition for variance bounding of $P_{\ref{alg:1_hit_kernel}}$,
and further that $P_{{\rm MH}}$ is at least as good as $P_{\ref{alg:1_hit_kernel}}$
in terms of the asymptotic variance of estimates such as \eqref{eq:mcmc_estimate}.
More importantly, and in contrast to $P_{\ref{alg:pm_kernel_1},N}$
and $P_{\ref{alg:pm_kernel_2},N}$, $P_{\ref{alg:1_hit_kernel}}$
can systematically inherit variance bounding and geometric ergodicity
from $P_{{\rm MH}}$ under Condition~\ref{cond:c2}.
\begin{prop}
\label{prop:vb_necessity}$P_{\ref{alg:1_hit_kernel}}$ and $P_{{\rm MH}}$
are ordered in the sense of \citet{Peskun1973} and \citet{Tierney1998},
so $P_{\ref{alg:1_hit_kernel}}\in\mathcal{V}\Rightarrow P_{{\rm MH}}\in\mathcal{V}$
and ${\rm var}(P_{{\rm MH}},\varphi)\leq{\rm var}(P_{\ref{alg:1_hit_kernel}},\varphi)$.\end{prop}
\begin{thm}
\label{thm:vb}\textup{Under }Condition~\ref{cond:c2}\textup{, $P_{{\rm MH}}\in\mathcal{V}\Rightarrow P_{\ref{alg:1_hit_kernel}}\in\mathcal{V}$}.
\end{thm}

\begin{thm}
\label{thm:ge}\textup{Under }Condition~\ref{cond:c2}\textup{, $P_{{\rm MH}}\in\mathcal{G}\Rightarrow P_{\ref{alg:1_hit_kernel}}\in\mathcal{G}$}.\end{thm}
\begin{rem}
\label{rmk:preciseness}Proposition~\ref{prop:vb_necessity} and
Theorems~\ref{thm:vb} and~\ref{thm:ge} are precise in the following
sense. There exist models for which $P_{\ref{alg:1_hit_kernel}}\in\mathcal{V}\setminus\mathcal{G}$
and $P_{{\rm MH}}\in\mathcal{V}\setminus\mathcal{G}$ and there exist
models for which $P_{\ref{alg:1_hit_kernel}}\in\mathcal{G}$ and $P_{{\rm MH}}\in\mathcal{V}\setminus\mathcal{G}$,
i.e.,\ under Condition~\ref{cond:c2}, $P_{{\rm MH}}\in\mathcal{V}\nRightarrow P_{\ref{alg:1_hit_kernel}}\in\mathcal{G}$
and $P_{\ref{alg:1_hit_kernel}}\in\mathcal{G}\nRightarrow P_{{\rm MH}}\in\mathcal{G}$.
Section~\ref{eg:compact} illustrates these possibilities. 
\end{rem}

\begin{rem}
\label{rmk:conditions}While Condition~\ref{cond:c2} is only a sufficient
condition, counterexamples can be constructed to show that some assumptions
are necessary for Theorems~\ref{thm:vb}--\ref{thm:ge} to hold.
Condition~\ref{cond:c2} allows us to ensure that $\alpha_{{\rm MH}}(\theta,\vartheta)$
and $\alpha_{\ref{alg:1_hit_kernel}}(\theta,\vartheta)$ differ only
in a controlled manner, for all $\theta$ and `enough' $\vartheta$,
and hence that $P_{{\rm MH}}$ and $P_{\ref{alg:1_hit_kernel}}$ are
not too different. As an example of the possible differences between
$P_{{\rm MH}}$ and $P_{\ref{alg:1_hit_kernel}}$ more generally,
consider the case where $p(\theta)=\tilde{p}(\theta)/\psi(\theta)$
and $h(\theta)=\tilde{h}(\theta)\psi(\theta)$ for some $\psi:\Theta\rightarrow(0,1]$.
Then properties of $P_{{\rm MH}}$ depend only on $\tilde{p}$ and
$\tilde{h}$ whilst those of $P_{\ref{alg:1_hit_kernel}}$ can additionally
be dramatically altered by the choice of $\psi$.
\end{rem}
Theorem~\ref{thm:ge} can be used to provide sufficient conditions
for $P_{\ref{alg:1_hit_kernel}}\in\mathcal{G}$ through $P_{{\rm MH}}\in\mathcal{G}$
and Condition~\ref{cond:c2}. The regular contour condition obtained
in \citet[Theorem~4.3]{jarner2000geometric}, e.g.,\ implies the
following corollary.
\begin{cor}
\label{cor:jarner_hansen}Assume (a) $h$ decays super-exponentially
and $p$ has exponential or heavier tails, or (b) $p$ has super-exponential
tails and $h$ decays exponentially or slower. If, moreover, $\pi$
is continuous and everywhere positive, $q$ is symmetric satisfying
$q(\theta,\vartheta)\geq\varepsilon_{q}$ whenever $|\theta-\vartheta|\leq\delta_{q}$,
for some $\varepsilon_{p},\:\delta_{q}>0$, and 
\begin{equation}
\limsup_{|\theta|\rightarrow\infty}\frac{\theta}{|\theta|}\cdot\frac{\nabla\pi(\theta)}{|\nabla\pi(\theta)|}<0,\label{eq:regular_contours}
\end{equation}
where $\cdot$ denotes the Euclidean scalar product, then $P_{3}\in\mathcal{G}$.
\end{cor}
Following Remark~\ref{rem:global_remark}, an alternative condition,
independent of the choice of $q$, that ensures inheritance of variance
bounding and geometric ergodicity of $P_{\ref{alg:1_hit_kernel}}$
from $P_{{\rm MH}}$ is that $\inf_{\theta\in\Theta}h(\theta)>0$,
i.e., that $h$ is lower bounded. This condition will usually only
hold when $\Theta$ is compact. Under this condition, both $P_{\ref{alg:pm_kernel_1},N}$
and $P_{\ref{alg:pm_kernel_2},N}$ will also successfully inherit
these properties, the former being already shown in \citet[Proposition~9]{Andrieu2012}
and for $P_{\ref{alg:pm_kernel_2},N}$ the same type of argument can
be used. This allows us to state the following corollary, which can
be verified by the arguments in \citet[Section~3.3]{roberts2004general}.
\begin{cor}
\label{cor:upper_lower_bounded}Let $\Theta$ be compact with $q$,
$p$ and $h$ all continuous, with $\inf_{\theta,\vartheta\in\Theta}q(\theta,\vartheta)>0$
and $\inf_{\theta\in\Theta}h(\theta)>0$. Then $P_{\ref{alg:pm_kernel_1},N}$,
$P_{\ref{alg:pm_kernel_2},N}$ and $P_{\ref{alg:1_hit_kernel}}$ are
all geometrically ergodic.\end{cor}
\begin{rem}
\label{rem:uniformlyergodic}In fact, under the conditions of Corollary~\ref{cor:upper_lower_bounded},
$P_{\ref{alg:pm_kernel_1},N}$, $P_{\ref{alg:pm_kernel_2},N}$ and
$P_{\ref{alg:1_hit_kernel}}$ are all uniformly ergodic since the
ratio of the acceptance probabilities $\alpha_{{\rm MH}}(\theta,\vartheta)/\alpha_{i}(\theta,\vartheta)$
is upper bounded by a constant for $i\in\{1,2,3\}$. This suggests
that in approximate Bayesian computation, a conservative choice is
to restrict inference to a compact set $\Theta$ in which $h$ is
lower bounded. 
\end{rem}
The proofs of Theorems~\ref{thm:vb} and~\ref{thm:ge} can also
be extended to cover the case where $\tilde{P}_{{\rm MH}}$ is a finite,
countable or continuous mixture of $P_{{\rm MH}}$ kernels associated
with a collection of proposals $\{q_{s}\}_{s\in S}$ and $\tilde{P}_{\ref{alg:1_hit_kernel}}$
is the corresponding mixture of $P_{\ref{alg:1_hit_kernel}}$ kernels.
With a modification of Condition~\ref{cond:c2}, the following proposition
is stated without proof, and could be used, e.g.,\ in conjunction
with \citet[Theorem~3]{fort2003geometric}.
\begin{condition}
\label{cond:c3}Each proposal $q$ is a member of $\mathcal{Q}$.
In addition, for all $K>0$, there exists an $M_{K}\in[1,\infty)$
such that for all $q_{t}\in\{q_{s}\}_{s\in S}$ and $(\theta,\vartheta)$
in the set
\[
\left\{ (\theta,\vartheta)\in\Theta^{2}\::\:\vartheta\in B_{K,\theta}\text{ and }\pi(\theta)q_{t}(\theta,\vartheta)\wedge\pi(\vartheta)q_{t}(\vartheta,\theta)>0\right\} ,
\]
either $h(\vartheta)/h(\theta)\in[M_{K}^{-1},M_{K}]$ or $c_{t}(\vartheta,\theta)/c_{t}(\theta,\vartheta)\in[M_{K}^{-1},M_{K}]$,
where $c_{t}(\theta,\vartheta)=p(\theta)q_{t}(\theta,\vartheta)$.\end{condition}
\begin{prop}
\label{prop:mixture_kernels}Let $\tilde{P}_{{\rm MH}}(\theta,{\rm d}\vartheta)=\int_{S}\mu({\rm d}s)P_{{\rm MH}}^{(s)}(\theta,{\rm d}\vartheta)$,
where $\mu$ is a mixing distribution on $S$ and each $P_{{\rm MH}}^{(s)}$
is a $\pi$-invariant Metropolis--Hastings kernel with proposal $q_{s}$.
Let\textup{ $\tilde{P}_{\ref{alg:1_hit_kernel}}(\theta,{\rm d}\vartheta)=\int_{S}\mu({\rm d}s)P_{\ref{alg:1_hit_kernel}}^{(s)}(\theta,{\rm d}\vartheta)$
be defined analogously. Then }$\tilde{P}_{\ref{alg:1_hit_kernel}}\in\mathcal{V}\Rightarrow\tilde{P}_{{\rm MH}}\in\mathcal{V}$
and ${\rm var}(\tilde{P}_{{\rm MH}},\varphi)\leq{\rm var}(\tilde{P}_{\ref{alg:1_hit_kernel}},\varphi)$,
and under Condition~\ref{cond:c3}, both \textup{$\tilde{P}_{{\rm MH}}\in\mathcal{V}\Rightarrow\tilde{P}_{\ref{alg:1_hit_kernel}}\in\mathcal{V}$
and $\tilde{P}_{{\rm MH}}\in\mathcal{G}\Rightarrow\tilde{P}_{\ref{alg:1_hit_kernel}}\in\mathcal{G}$.}
\end{prop}
We provide also a general result that can justify, e.g., using $P_{3}$
as one component of a mixture of reversible kernels, of which some
may not be variance bounding or geometrically ergodic.
\begin{thm}
\label{thm:mix_geo_nongeo}Let $\tilde{K}=\sum_{i=1}^{\infty}a_{i}K_{i}$
be a mixture of reversible Markov kernels with invariant distribution
$\pi$ where $\sum_{i=1}^{\infty}a_{i}=1$ and $a_{i}\geq0$ for $i\in\mathbb{N}$.
Let $K_{1}$ have unique invariant distribution $\pi$ and $a_{1}>0$.
Then $K_{1}\in\mathcal{V}\Rightarrow\tilde{K}\in\mathcal{V}$ and
$K_{1}\in\mathcal{G}\Rightarrow\tilde{K}\in\mathcal{G}$.
\end{thm}
While the sampling of a random number of auxiliary variables in the
implementation of $P_{\ref{alg:1_hit_kernel}}$ appears to be helpful
in inheriting qualitative properties of $P_{{\rm MH}}$, one may be
concerned that the computational effort associated with the kernel
can be unbounded. Our final result indicates that this is not the
case whenever $p$ is proper.
\begin{prop}
\label{prop:n}Let $(N_{i})$ be the sequence of random variables
associated with step 3 of Algorithm~\ref{alg:1_hit_kernel} if one
iterates $P_{\ref{alg:1_hit_kernel}}$, with $N_{j}=0$ if at iteration
$j$ the kernel outputs at step 2. Then if $\int p(\theta)\mathrm{d}\theta=1$,
$H>0$, and $P_{\ref{alg:1_hit_kernel}}$ is irreducible, 
\[
n=\lim_{m\rightarrow\infty}m^{-1}\sum_{i=1}^{m}N_{i}\leq H^{-1}<\infty.
\]

\end{prop}
When $p$ is proper, $H$ is a natural quantity; if $n_{R}$ is the
expected number of proposals to obtain a sample from $\pi$ using
the rejection sampler of \citet{Pritchard1999} we have $n_{R}=H^{-1}$,
and if we construct $P_{\ref{alg:pm_kernel_1},N}$ with proposal $q(\theta,\vartheta)=p(\vartheta)$
then $H$ lower bounds its spectral gap. In fact, $n$ can be arbitrarily
smaller than $n_{R}$, as we illustrate in Section~\ref{eg:compact},
and on a realistic example in Section~\ref{eg:lotka_volterra} the
average number of samples required per iteration was much smaller
than $H^{-1}$.

One potential issue with all three of the kernels $P_{\ref{alg:pm_kernel_1},N}$,
$P_{\ref{alg:pm_kernel_2},N}$ and $P_{\ref{alg:1_hit_kernel}}$,
when implemented using local proposals, is that their performance
for a fixed computational budget will be poor if the Markov chain
is initialized in a region of the state space with little posterior
mass. This can be circumvented by trying to identify regions of high
posterior mass and initializing the chain at a point in such a region.
Finally, Remark~\ref{rem:uniformlyergodic} suggests that a conservative
choice is to let $\Theta$ be a compact set in which $h$ is lower
bounded, and would contain most of the interesting values of $\theta$.

\section{Examples\label{sec:Examples}}

\subsection{A posterior with compact support\label{eg:compact}}

We begin with a simple example that clarifies comments in Remark~\ref{rmk:preciseness}
and some of those following Proposition~\ref{prop:n}. In particular,
$\theta\in\Theta=\mathbb{R}_{+}$, $p(\theta)=I\left(0\leq\theta\leq a\right)/a$
and $h(\theta)=bI\left(0\leq\theta\leq1\right)$ for $(a,b)\in[1,\infty)\times(0,1]$,
with $\pi$ supported on $[0,1]$.

We have $H^{-1}=a/b$ and $n\leq b^{-1}$ for any $q$ so $n_{R}/n\geq a$.
Furthermore, even if $p$ is improper, $n$ is finite. Regarding Remark~\ref{rmk:preciseness},
for any $a\geq1$, consider the proposal $q(\theta,\vartheta)=2I\left(0\leq\theta\leq1/2\right)I\left(1/2<\vartheta\leq1\right)+2I\left(1/2<\theta\leq1\right)I\left(0\leq\vartheta\leq1/2\right)$.
If $b=1$, then $P_{\ref{alg:1_hit_kernel}}\in\mathcal{V}\setminus\mathcal{G}$
and $P_{{\rm MH}}\in\mathcal{V}\setminus\mathcal{G}$. However, if
$b\in(0,1)$ then $P_{\ref{alg:1_hit_kernel}}\in\mathcal{G}$ and
$P_{{\rm MH}}\in\mathcal{V}\setminus\mathcal{G}$.

\subsection{Geometric distribution\label{eg:geometric}}

We consider the situation where $\theta\in\Theta=\mathbb{Z}_{+}$,
$p(\theta)=I\left(\theta\in\mathbb{N}\right)(1-a)a^{\theta-1}$ and
$h(\theta)=b^{\theta}$ for $(a,b)\in(0,1)^{2}$. The posterior $\pi$
is a geometric distribution with success parameter $1-ab$ and geometric
series manipulations provided in Appendix~\ref{sec:eg1_calculations}
give the expected number of proposals needed in the rejection sampler
$n_{R}=(1-ab)/\left\{ b(1-a)\right\} $. If $q(\theta,\vartheta)=\left\{ I(\vartheta=\theta-1)+I(\vartheta=\theta+1)\right\} /2$,
we have 
\begin{equation}
\frac{(1-ab)}{2}\left\{ \frac{(a+b)}{b(1-a)(1+b)}-1\right\} \leq n\leq\frac{(1-ab)}{2}\left\{ \frac{a+b}{b(1-a)}-1\right\} ,\label{eq:n_bounds_eg1}
\end{equation}
where $n$ is as in Proposition~\ref{prop:n}, and so $n_{R}/n\geq2/\left\{ a(1+b)\right\} $,
which grows without bound as $a\rightarrow0$. Regarding the propriety
condition on $p$, we observe that $n_{R}\rightarrow\infty$ and $n\rightarrow\infty$
as $a\rightarrow1$ with $b$ fixed.

To supplement the qualitative results regarding variance bounding
and geometric ergodicity of the kernels, we investigated a modification
of this example with a finite number of states. More specifically,
we considered the case where the prior is truncated to the set $\{1,\ldots,D\}$
for some $D\in\mathbb{N}$. In this context, we can calculate explicit
transition probabilities and hence spectral gaps $1-|\sigma_{0}(P)|$
and asymptotic variances ${\rm var}(P,\varphi)$ of \eqref{eq:mcmc_estimate}
for $P_{\ref{alg:pm_kernel_2},N}$, $P_{\ref{alg:1_hit_kernel}}$
and $P_{{\rm MH}}$. Figure~\ref{fig:log_spectral_gaps} shows the
$\log$ spectral gaps for a range of values of $D$ for each kernel
and $b\in\{0.1,0.5,0.9\}$. We can see the spectral gaps of $P_{\ref{alg:1_hit_kernel}}$
and $P_{{\rm MH}}$ stabilize, whilst those of $P_{\ref{alg:pm_kernel_2},N}$
decrease exponentially fast in $D$, albeit with some improvement
for larger $N$. The spectral gaps obtained, with \eqref{eq:ge_tv},
suggest that the convergence of $P_{\ref{alg:pm_kernel_2},N}$ to
$\pi$ can be extremely slow for some $\theta_{0}$ even when $D$
is relatively small. Indeed, in this finite, discrete setting with
reversible $P$, the bounds 
\[
\frac{1}{2}\left\{ \max|\sigma_{0}(P)|\right\} ^{m}\leq\max_{\theta_{0}}\Vert\pi(\cdot)-P^{m}(\theta_{0},\cdot)\Vert_{{\rm TV}}\leq\frac{1}{2}\left\{ \max|\sigma_{0}(P)|\right\} ^{m}\left\{ \frac{1-\min_{\theta}\pi(\theta)}{\min_{\theta}\pi(\theta)}\right\} ^{1/2}
\]
hold \citep[Section~2 and Theorem~5.9]{montenegro2006mathematical},
which clearly indicate that $P_{\ref{alg:pm_kernel_2},N}$ can converge
exceedingly slowly when $P_{\ref{alg:1_hit_kernel}}$ and $P_{{\rm MH}}$
converge reasonably quickly. The value of $n$ in these cases stabilized
at $4.77$, $0.847$ and $0.502$ for $b\in\{0.1,0.5,0.9\}$ respectively,
within the bounds of \eqref{eq:n_bounds_eg1}, and considerably smaller
than $100$.

\begin{figure}
\begin{centering}
\subfloat[$b=0.1$]{\includegraphics[scale=0.25]{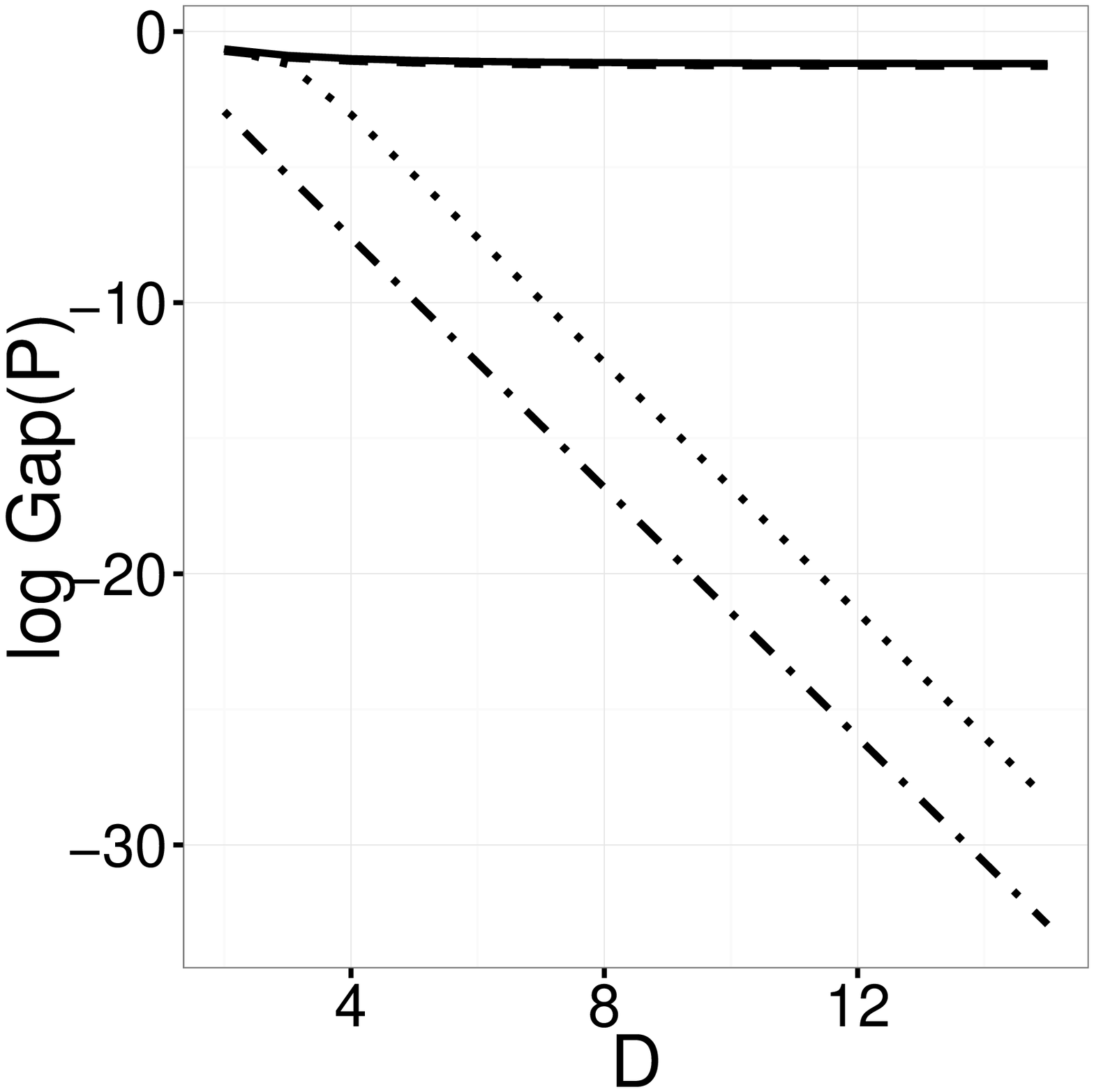}}\subfloat[$b=0.5$]{\includegraphics[scale=0.25]{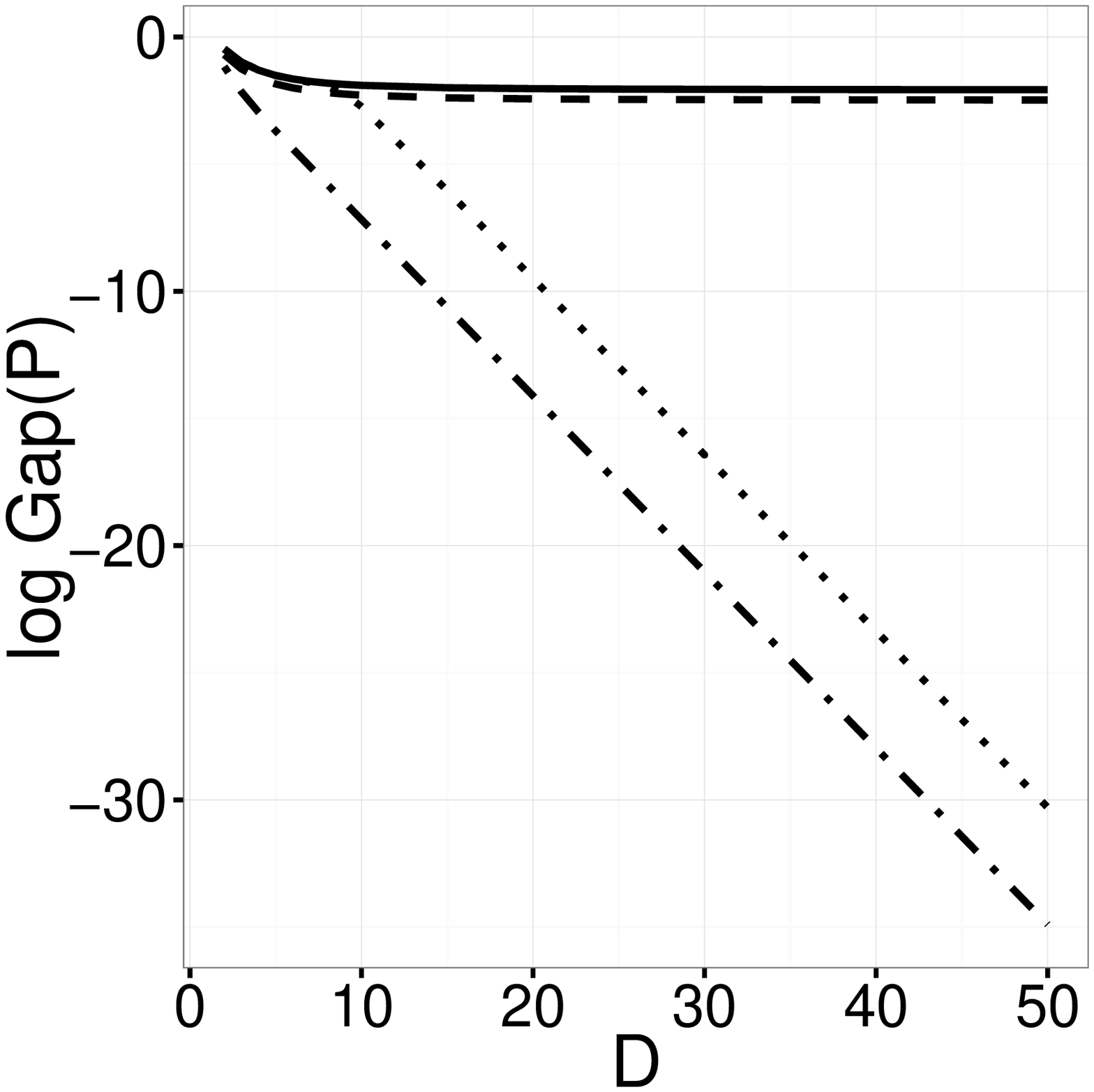}}\subfloat[$b=0.9$]{\includegraphics[scale=0.25]{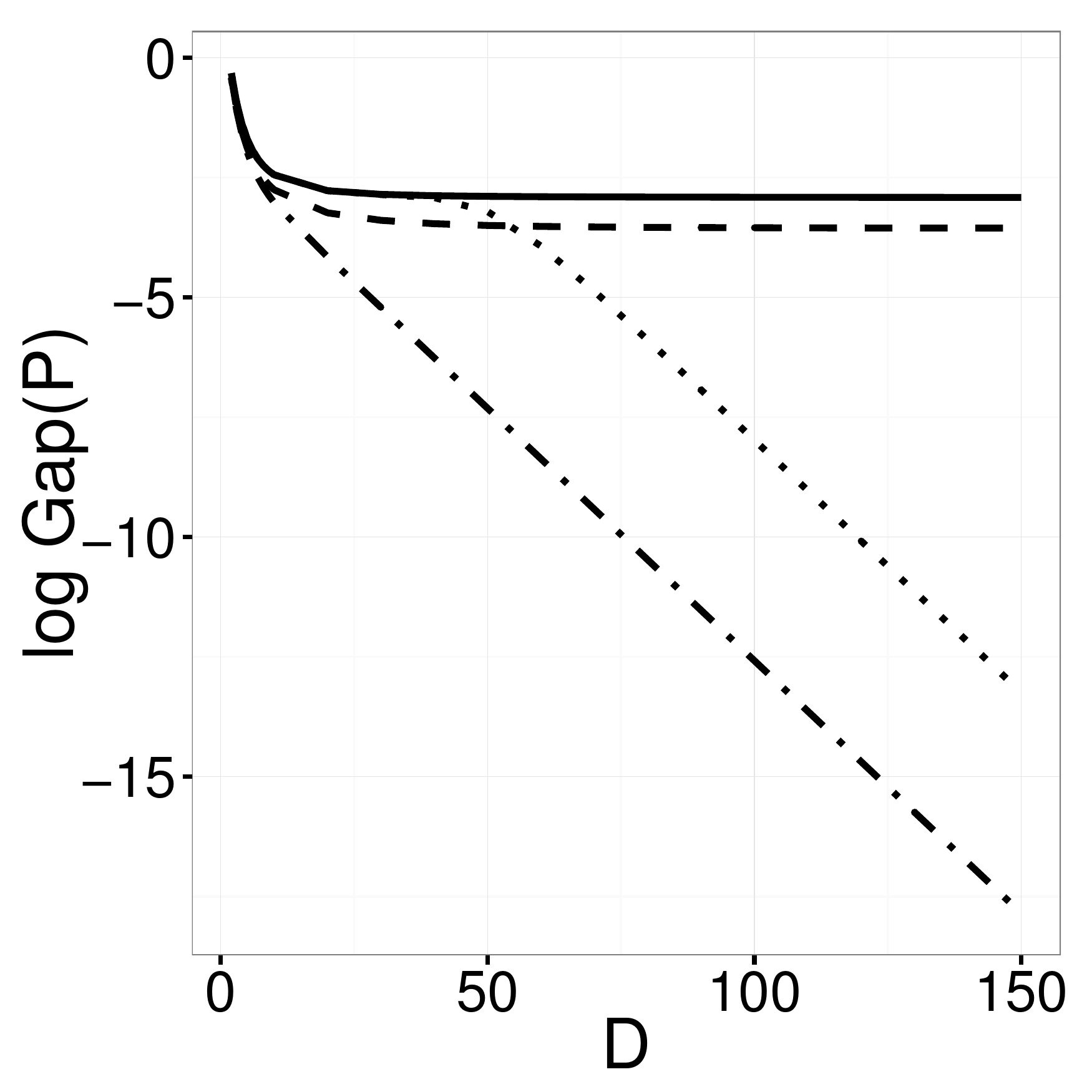}}
\par\end{centering}

\caption{Plot of the $\log$ spectral gap against $D$ for $P_{\ref{alg:pm_kernel_2},1}$
(dot-dashed), $P_{\ref{alg:pm_kernel_2},100}$ (dotted), $P_{\ref{alg:1_hit_kernel}}$
(dashed) and $P_{{\rm MH}}$ (solid), with $a=0.5$.\label{fig:log_spectral_gaps}}
\end{figure}

Figures~\ref{fig:log_asymp_variance} and~\ref{fig:log_asymp_variance_2}
show $\log{\rm var}(P,\varphi)$ against $D$ for $\varphi_{1}(\theta)=\theta$
and $\varphi_{2}(\theta)=(ab)^{-\theta/2.1}$, respectively, computed
using the expression of \citet[p. 84]{Kemeny1969}. The choice of
$\varphi_{2}$ is motivated by the fact when $p$ is not truncated,
$\varphi(\theta)=(ab)^{-\theta/(2+\delta)}$ is in $L^{2}(\pi)$ if
and only if $\delta>0$. While ${\rm var}(P,\varphi_{1})$ is stable
for all the kernels, ${\rm var}(P,\varphi_{2})$ increases rapidly
with $D$ for $P_{\ref{alg:pm_kernel_2},1}$ and $P_{\ref{alg:pm_kernel_2},100}$.
While ${\rm var}(P_{\ref{alg:pm_kernel_2},N},\varphi_{1})$ can be
lower than ${\rm var}(P_{\ref{alg:1_hit_kernel}},\varphi_{1})$, the
former requires many more simulations from the likelihood. Indeed,
while the results we have obtained pertain to qualitative properties
of the Markov kernels, this example illustrates that $P_{\ref{alg:1_hit_kernel}}$
can significantly outperform $P_{\ref{alg:pm_kernel_2},100}$ for
estimating even the more well-behaved $\pi(\varphi_{1})$, when cost
per iteration of each kernel is taken into account.

\begin{figure}
\begin{centering}
\subfloat[$b=0.1$]{\includegraphics[scale=0.25]{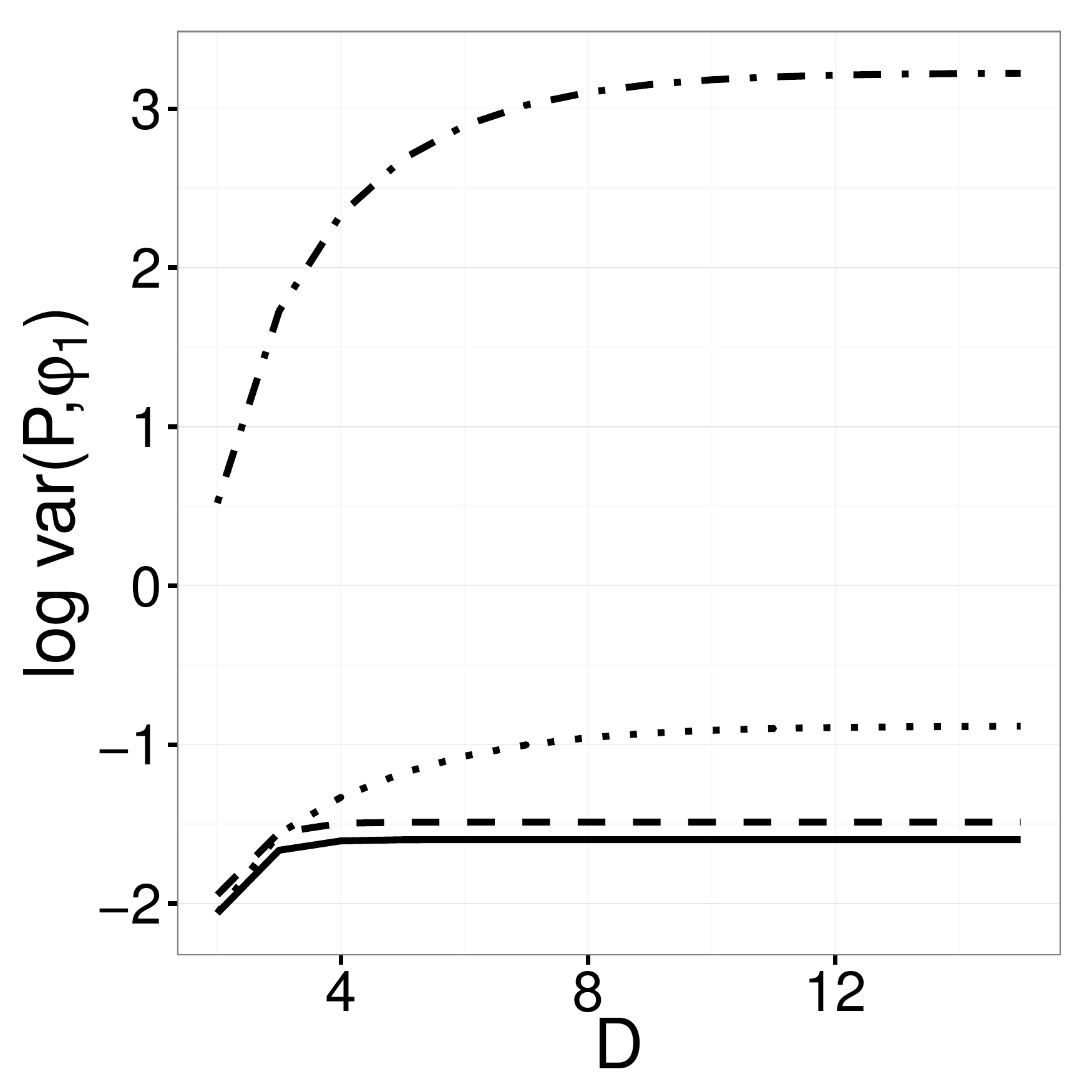}}\subfloat[$b=0.5$]{\includegraphics[scale=0.25]{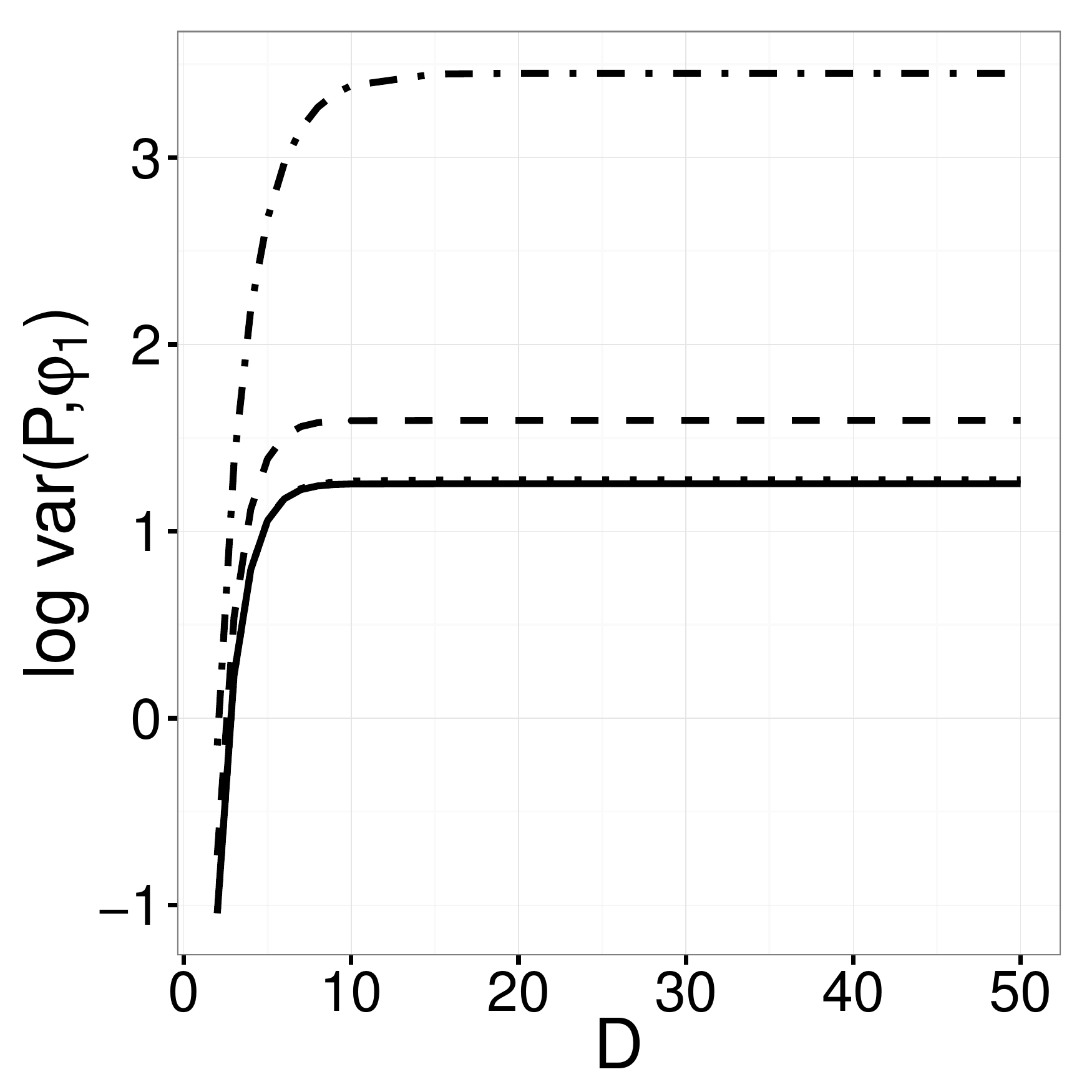}}\subfloat[$b=0.9$]{\includegraphics[scale=0.25]{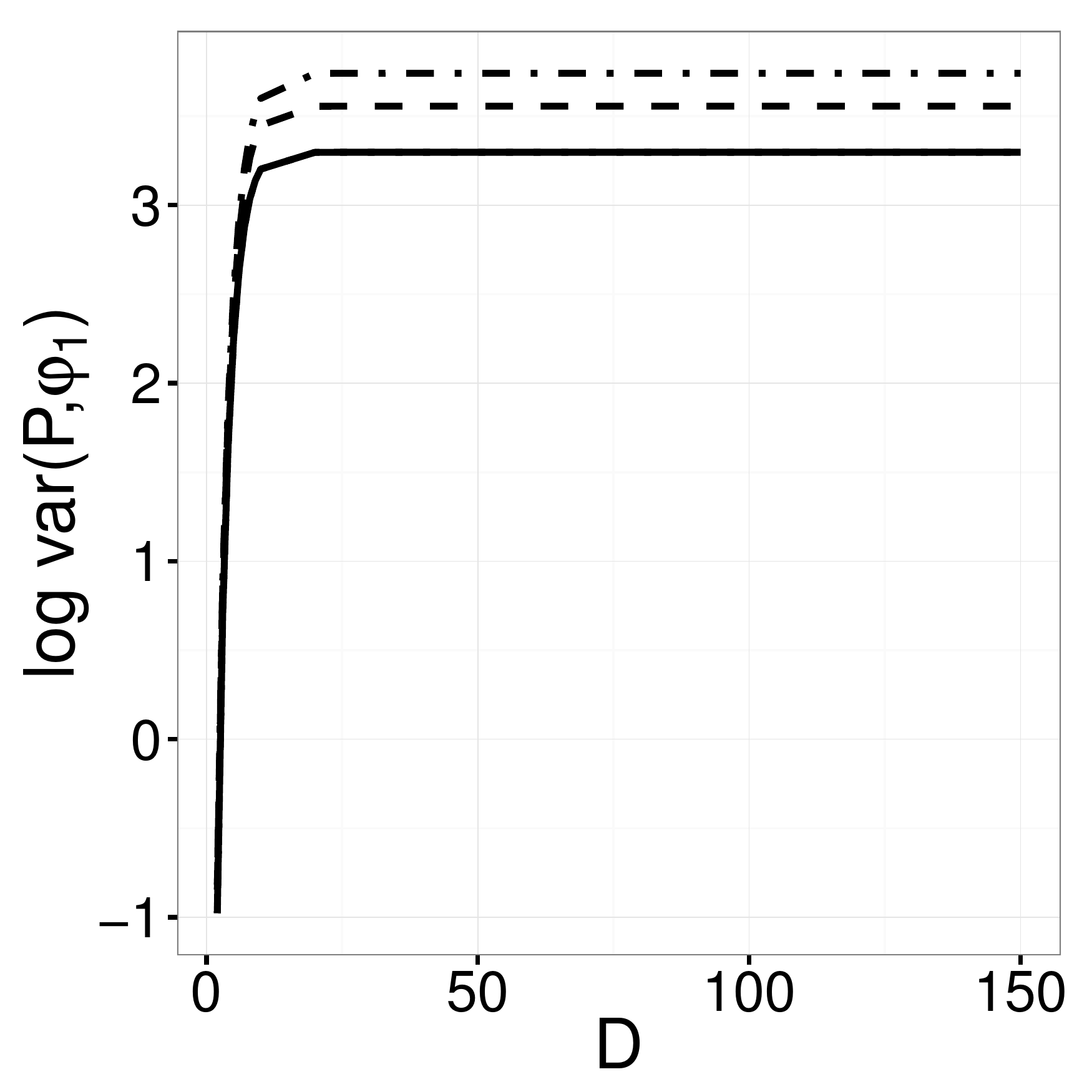}}
\par\end{centering}

\caption{Plot of $\log{\rm var}(P,\varphi_{1})$ against $D$ for $P=P_{\ref{alg:pm_kernel_2},1}$
(dot-dashed), $P=P_{\ref{alg:pm_kernel_2},100}$ (dotted), $P=P_{\ref{alg:1_hit_kernel}}$
(dashed) and $P=P_{{\rm MH}}$ (solid), with $a=0.5$.\label{fig:log_asymp_variance}}
\end{figure}

\begin{figure}
\begin{centering}
\subfloat[$b=0.1$]{\includegraphics[scale=0.25]{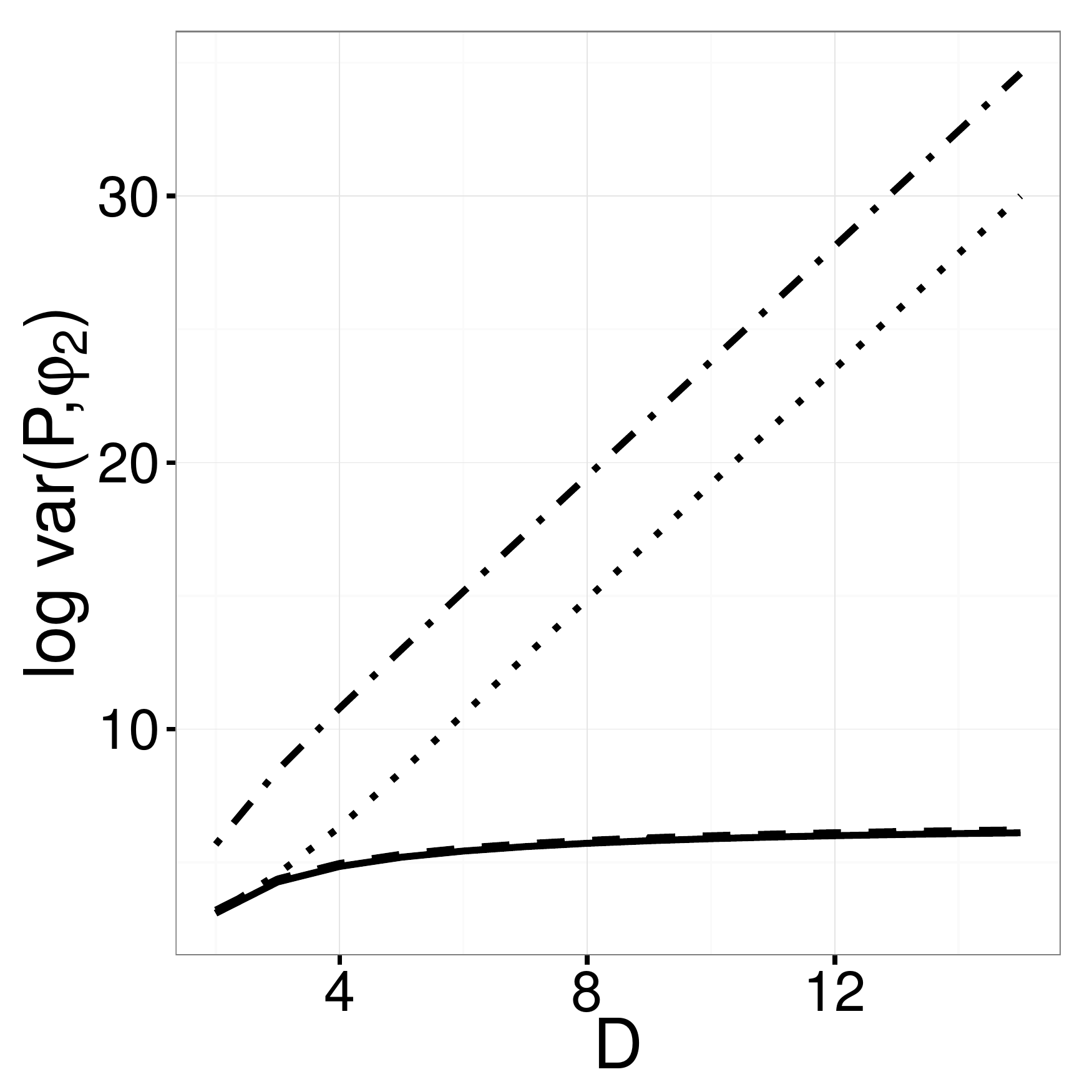}}\subfloat[$b=0.5$]{\includegraphics[scale=0.25]{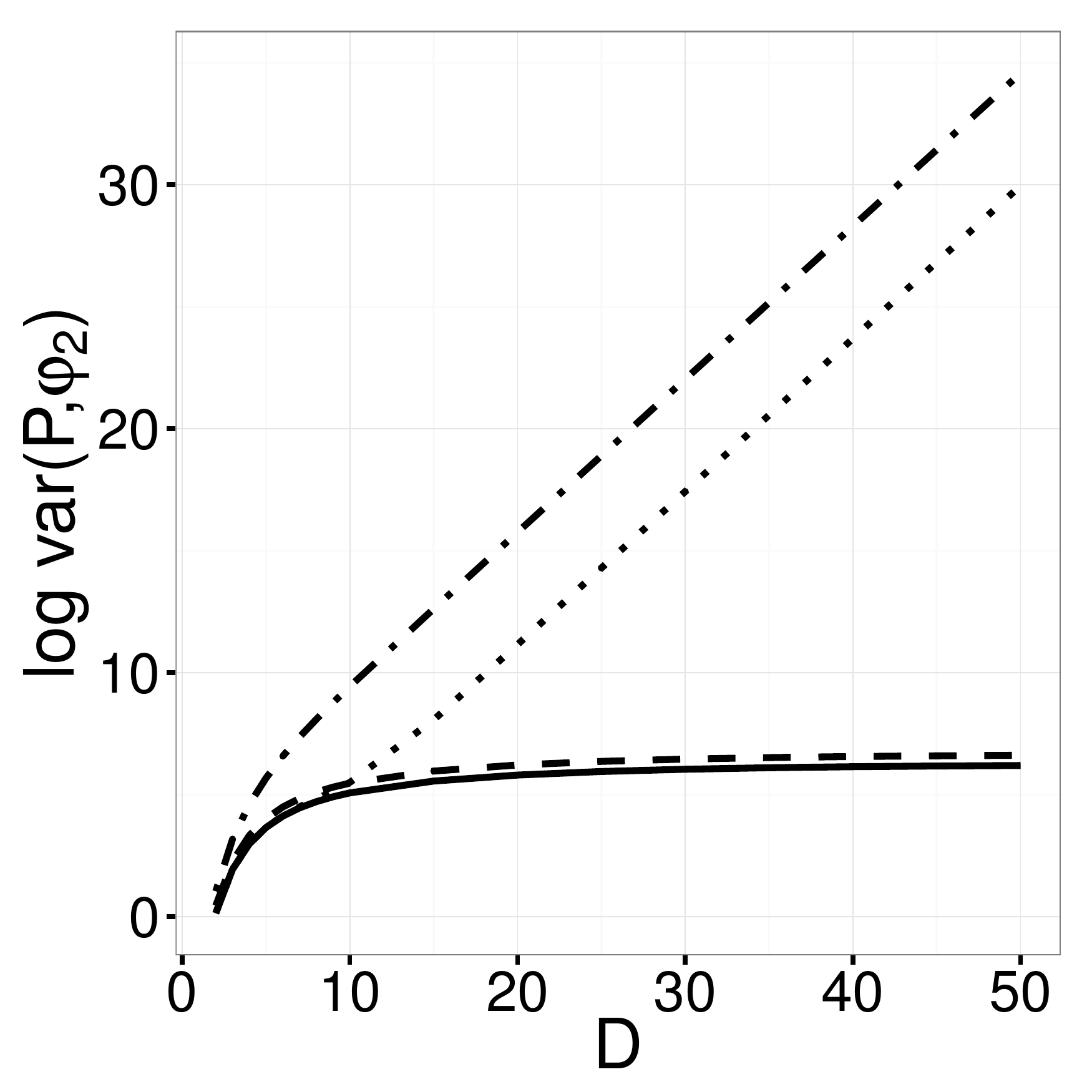}}\subfloat[$b=0.9$]{\includegraphics[scale=0.25]{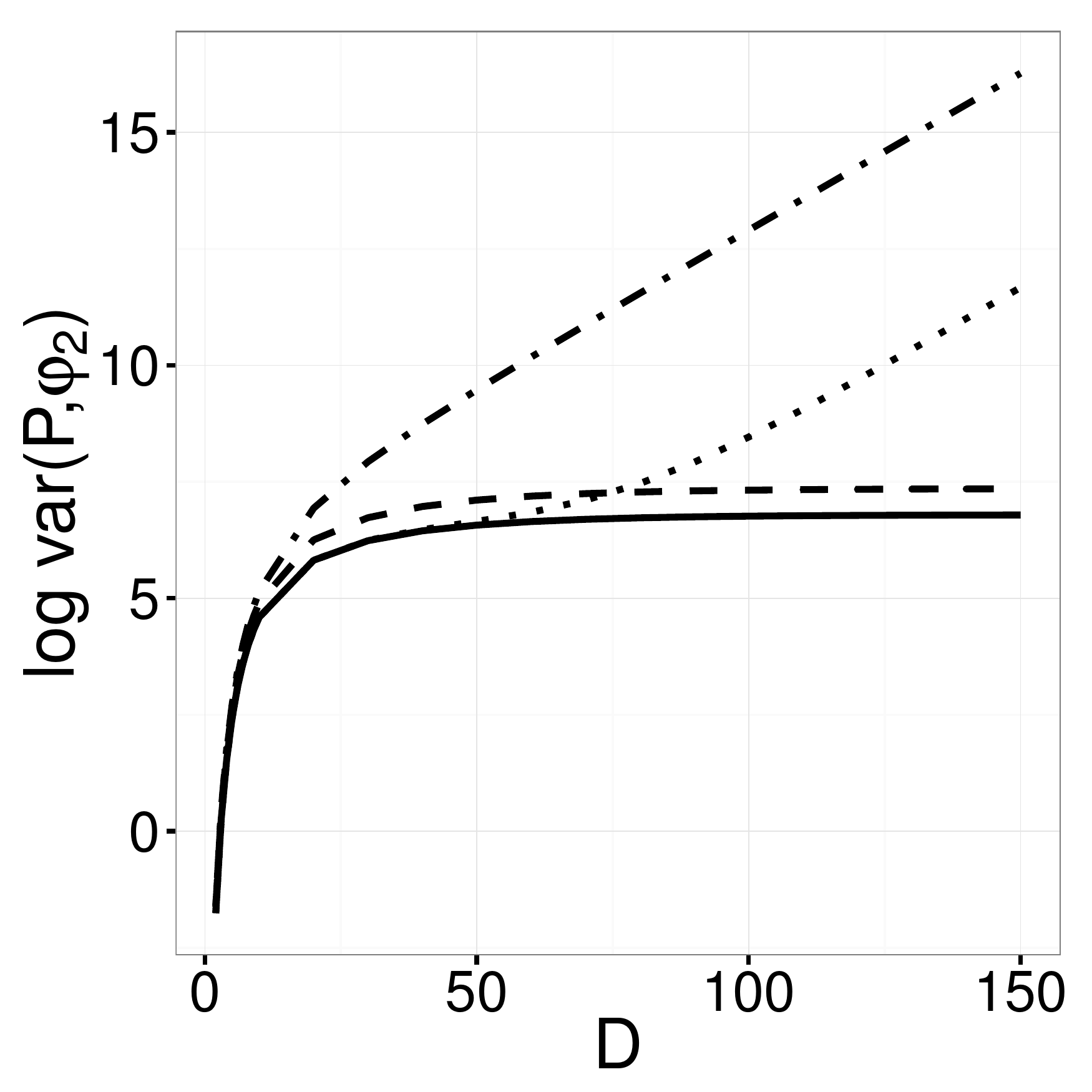}}
\par\end{centering}

\caption{Plot of $\log{\rm var}(P,\varphi_{2})$ against $D$ for $P=P_{\ref{alg:pm_kernel_2},1}$
(dot-dashed), $P=P_{\ref{alg:pm_kernel_2},100}$ (dotted), $P=P_{\ref{alg:1_hit_kernel}}$
(dashed) and $P=P_{{\rm MH}}$ (solid), with $a=0.5$.\label{fig:log_asymp_variance_2}}
\end{figure}

\begin{figure}
\begin{centering}
\subfloat[$b=0.1$]{\includegraphics[scale=0.25]{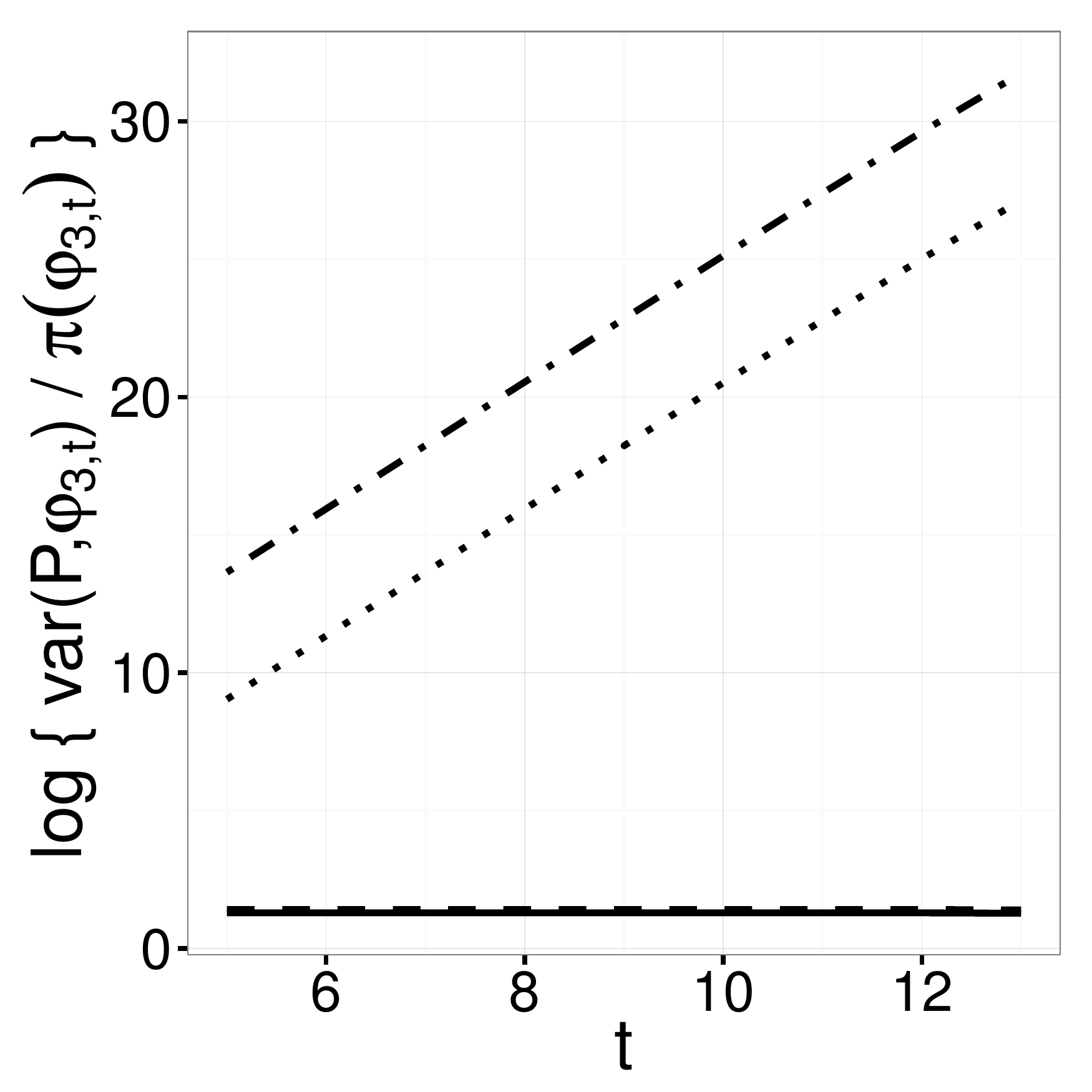}}\subfloat[$b=0.5$]{\includegraphics[scale=0.25]{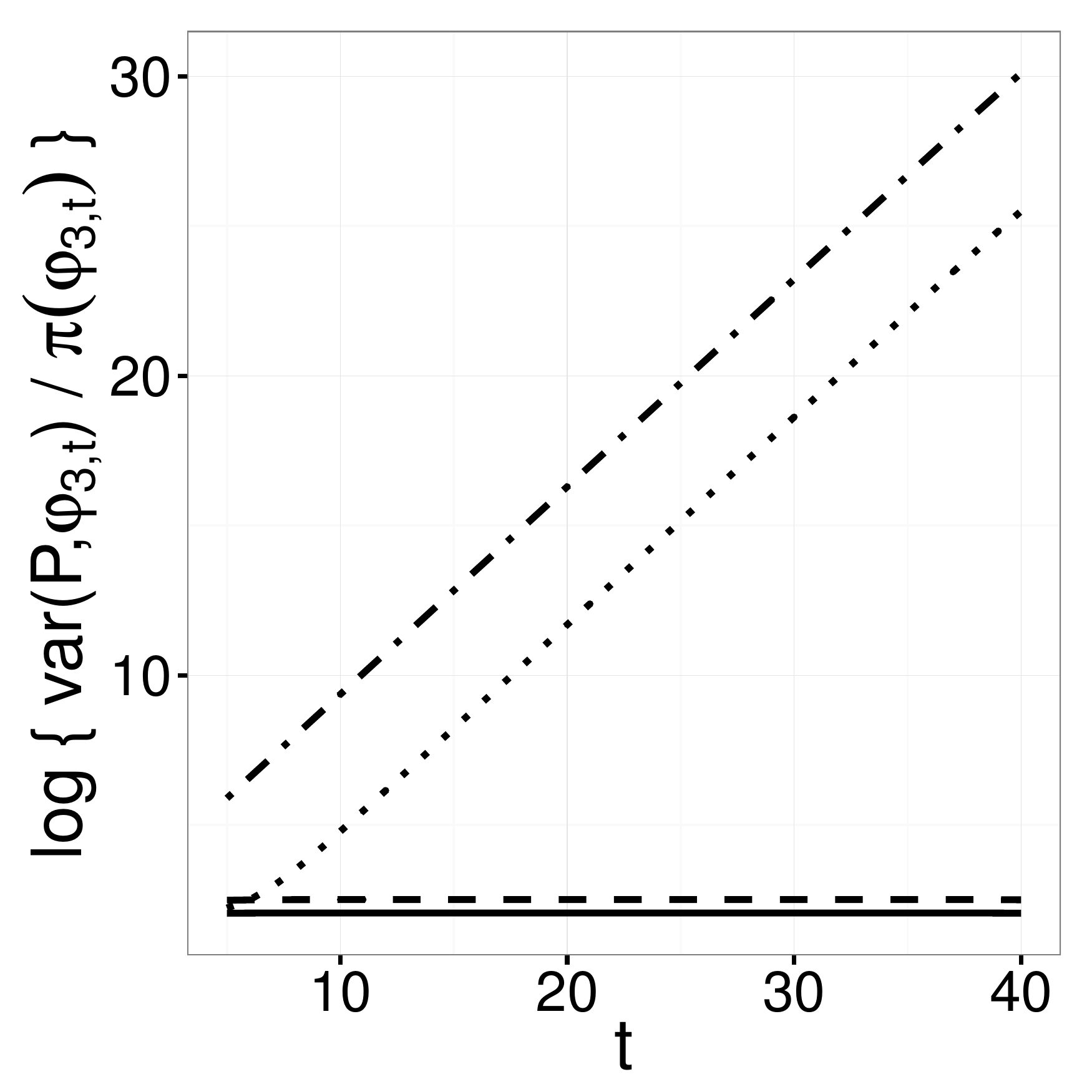}}\subfloat[$b=0.9$]{\includegraphics[scale=0.25]{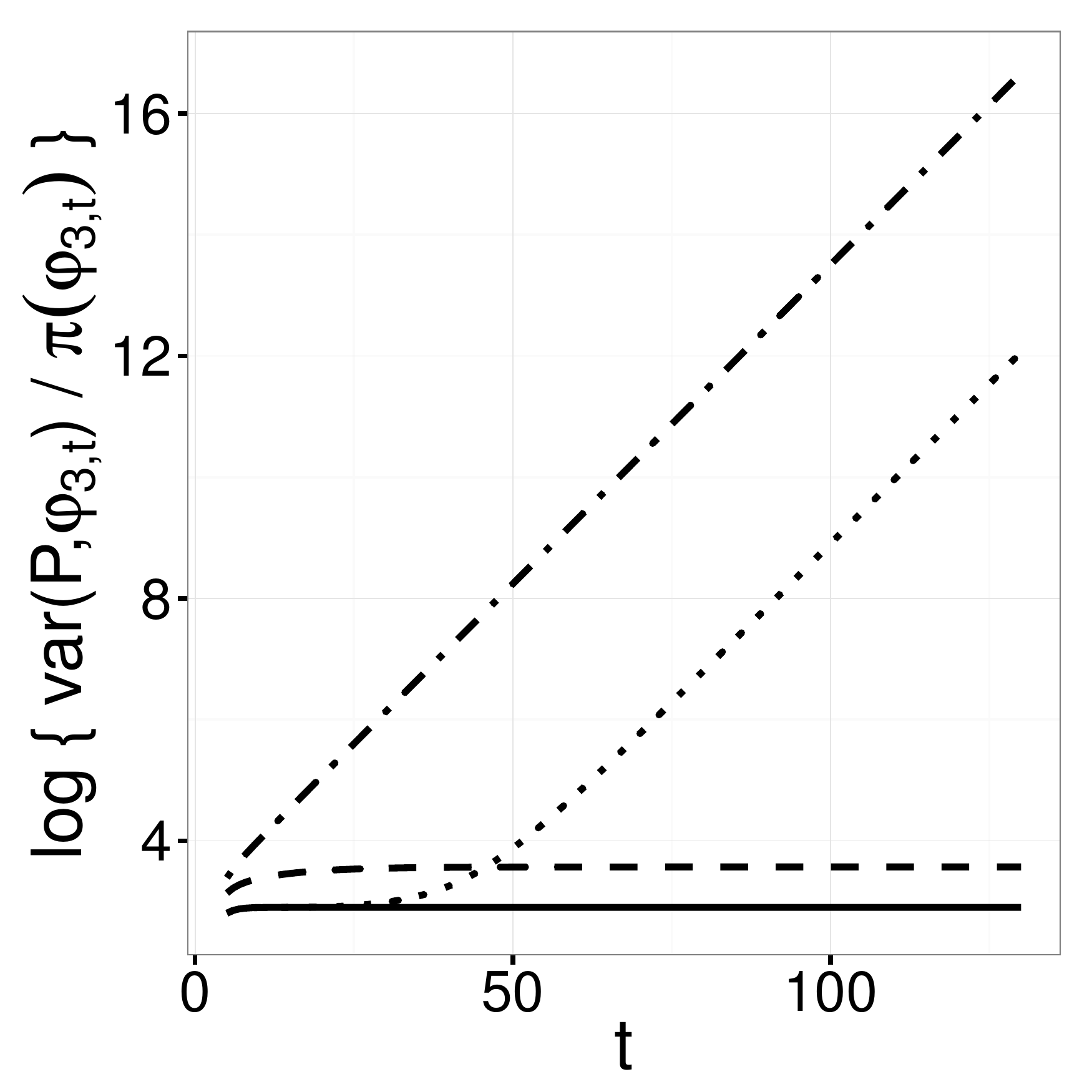}}
\par\end{centering}

\caption{Plot of $\log\left\{ {\rm var}(P,\varphi_{3,t})/\pi(\varphi_{3,t})\right\} $
against $t$ for $P=P_{\ref{alg:pm_kernel_2},1}$ (dot-dashed), $P=P_{\ref{alg:pm_kernel_2},100}$
(dotted), $P=P_{\ref{alg:1_hit_kernel}}$ (dashed) and $P=P_{{\rm MH}}$
(solid), with $a=0.5$.\label{fig:log_asymp_variance_tail}}
\end{figure}

Figure~\ref{fig:log_asymp_variance_tail} shows $\log\left\{ {\rm var}(P,\varphi_{3,t})/\pi(\varphi_{3,t})\right\} $
against $t$ for $\varphi_{3,t}(\theta)=\mathbf{1}_{\{t,t+1,\ldots\}}(\theta)$
so that $\pi(\varphi_{3,t})$ is the tail probability. The division
by $\pi(\varphi_{3,t})$ makes this an appropriately scaled relative
asymptotic variance since one needs $1/\pi(\varphi_{3,t})$ perfect
samples from $\pi$ in expectation to get a single sample in the region
$\{t,t+1,\ldots\}$. The figure shows that while $P_{{\rm MH}}$ and
$P_{\ref{alg:1_hit_kernel}}$ have constant $\log\left\{ {\rm var}(P,\varphi_{3,t})/\pi(\varphi_{3,t})\right\} $
as $t$ increases, $P_{\ref{alg:pm_kernel_2},1}$ and $P_{\ref{alg:pm_kernel_2},100}$
do not, as a result of their inability to estimate tail probabilities
accurately. In various applications, approximate Bayesian computation
might be used to infer such posterior tail probabilities and these
results indicate that $P_{\ref{alg:pm_kernel_1},N}$ and $P_{\ref{alg:pm_kernel_2},N}$
may not be appropriate when such inferences are desired.

\begin{figure}
\begin{centering}
\subfloat[$a=0.9$]{\includegraphics[scale=0.25]{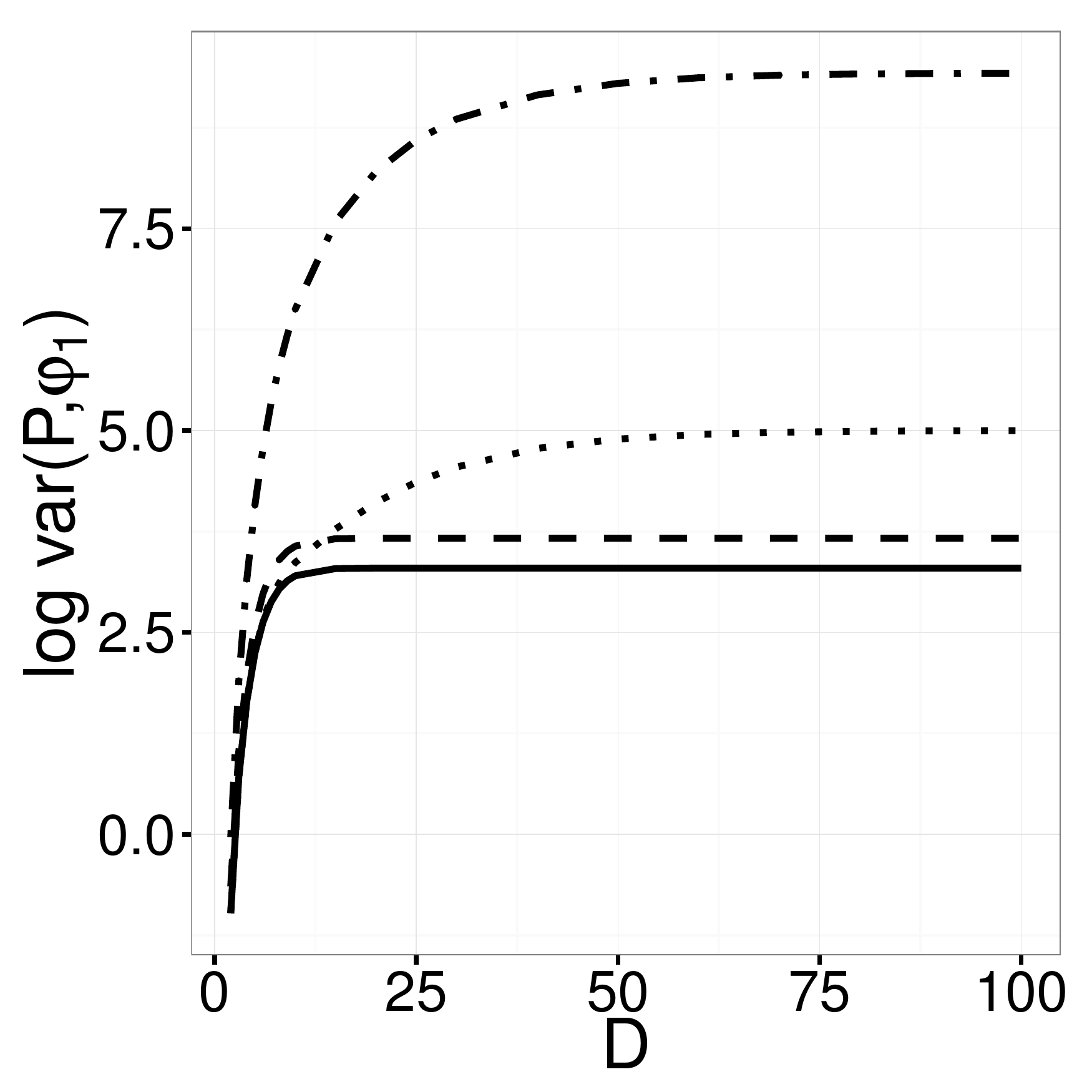}}\subfloat[$a=0.99$]{\includegraphics[scale=0.25]{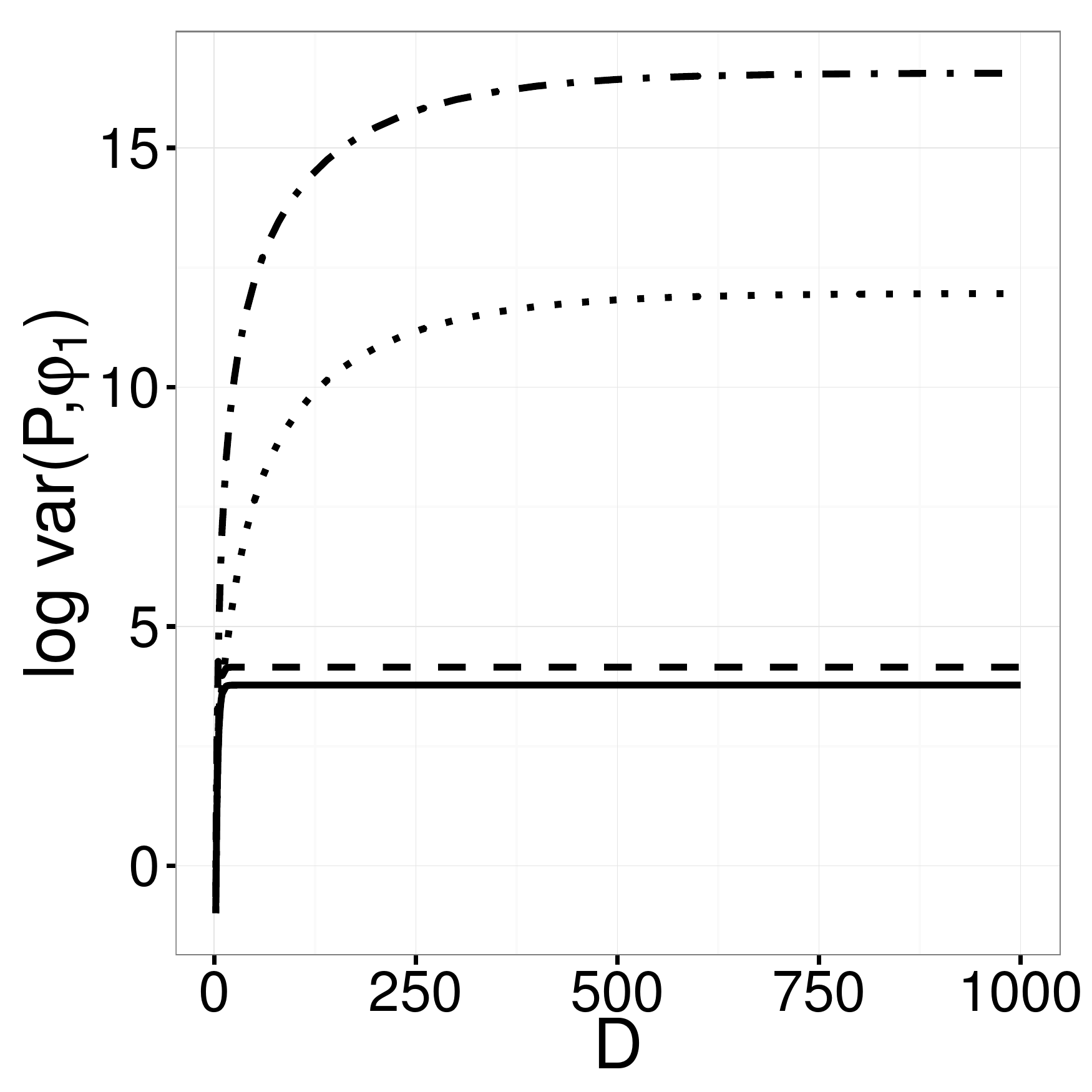}}\subfloat[$a=0.999$]{\includegraphics[scale=0.25]{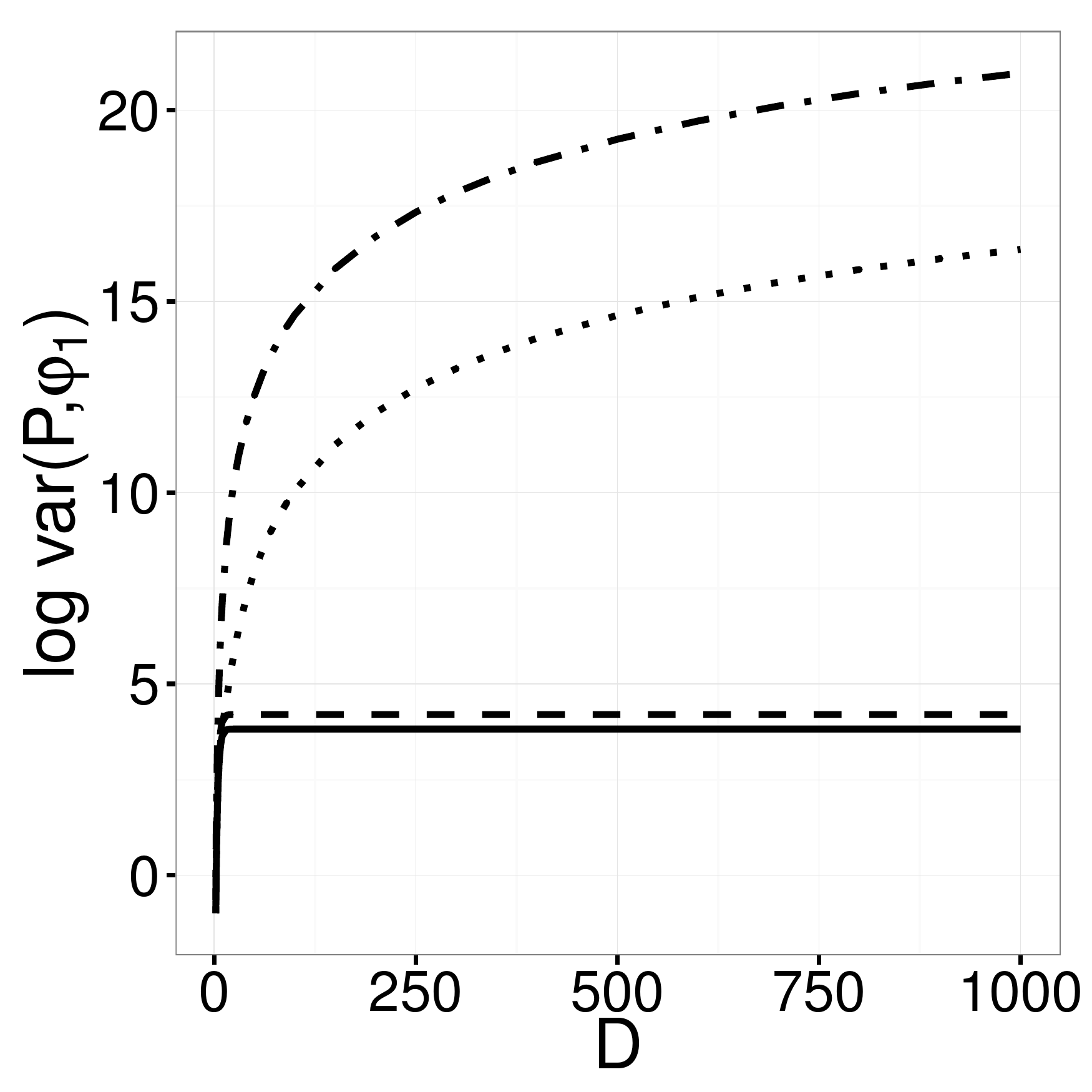}}
\par\end{centering}

\caption{Plot of $\log{\rm var}(P,\varphi_{1})$ against $D$ for $P=P_{\ref{alg:pm_kernel_2},1}$
(dot-dashed), $P=P_{\ref{alg:pm_kernel_2},100}$ (dotted), $P=P_{\ref{alg:1_hit_kernel}}$
(dashed) and $P=P_{{\rm MH}}$ (solid), with $b=0.5$.\label{fig:log_asymp_variance_varyinga}}
\end{figure}

Figure~\ref{fig:log_asymp_variance_varyinga} shows the behaviour
of the estimate of the posterior mean for $a\in\{0.9,0.99,0.999\}$
with corresponding values of $n$ for $P_{\ref{alg:1_hit_kernel}}$
being approximately $5$, $50$ and $500$. To take account of the
cost of the kernels, it is informative to consider $N{\rm var}(P_{\ref{alg:pm_kernel_2},N},\varphi_{1})$
and $n{\rm var}(P_{\ref{alg:1_hit_kernel}},\varphi_{1})$. For these
values of $a$, we have ${\rm var}(P_{\ref{alg:pm_kernel_2},1},\varphi_{1})$
roughly equal to $100{\rm var}(P_{\ref{alg:pm_kernel_2},100},\varphi_{1})$,
although $P_{\ref{alg:pm_kernel_2},100}$ is more feasibly implemented
in parallel on emerging many-core devices such as graphics processing
units \citep[see, e.g., ][]{Lee2010}. On the other hand ${\rm var}(P_{\ref{alg:pm_kernel_2},1},\varphi_{1})/\left\{ n{\rm var}(P_{\ref{alg:1_hit_kernel}},\varphi_{1})\right\} $
is about $75$, $5000$ and well over $60000$ for $a$ equal to $0.9$,
$0.99$ and $0.999$ respectively.

\subsection{Stochastic Lotka--Volterra model\label{eg:lotka_volterra}}

We turn to stochastic kinetic models for which the posterior is not
of a simple form, and exhibits strong correlations between components
of $\theta$. Such models are used, e.g.,\ in systems biology where
Bayesian inference has been investigated in \citet{boys2008bayesian}
and \citet{wilkinson2011stochastic}. We consider a simple member
of this class of models, the Lotka--Volterra predator-prey model \citep{Lotka1925,volterra1927variazioni},
which was also considered as an example for approximate Bayesian computation
in \citet{toni2009approximate} and \citet{Fearnhead2012}.

In this setting $X_{1:2}(t)$ is bivariate, integer-valued pure jump
Markov process with $X_{1:2}(0)=(50,100)$. For small $\Delta t$,
we have
\begin{multline*}
{\rm pr}\left\{ X_{1:2}(t+\Delta t)=z_{1:2}\mid X_{1:2}(t)=x_{1:2}\right\} \\
=\begin{cases}
\theta_{1}x_{1}\Delta t+o(\Delta t) & \text{if }z_{1:2}=(x_{1}+1,x_{2}),\\
\theta_{2}x_{1}x_{2}\Delta t+o(\Delta t) & \text{if }z_{1:2}=(x_{1}-1,x_{2}+1),\\
\theta_{3}x_{2}\Delta t+o(\Delta t) & \text{if }z_{1:2}=(x_{1},x_{2}-1),\\
1-\Delta t\left(\theta_{1}x_{1}+\theta_{2}x_{1}x_{2}+\theta_{3}x_{2}\right)+o(\Delta t) & \text{if }z_{1:2}=x_{1:2},\\
o(\Delta t) & {\rm otherwise},
\end{cases}
\end{multline*}
where the first three cases correspond in order to prey birth, prey
consumption and predator death. Theory and methodology related to
the simulation of this type of time-homogeneous, pure jump Markov
process and historical uses in statistics can be traced through \citet{feller1940integro},
\citet{doob1945markoff} and \citet{kendall1949stochastic,kendall1950artificial},
and the method was rediscovered in \citet{gillespie1977exact} in
the context of stochastic kinetic models. These articles develop a
straightforward way to simulate the full process $X_{1:2}(t),\; t\in[0,10]$,
as the inter-jump times are exponential random variables, although
more sophisticated approaches are possible \citep[see, e.g.,][Chapter 8]{wilkinson2011stochastic}.

The data was simulated with $\theta=(1,0.005,0.6)$, an example from
\citet[p. 152]{wilkinson2011stochastic}. Our observations are both
partial and discrete with $y=\{88,165,274,268,114,46,32,36,53,92\}$
the simulated values of $X_{1}$ at times $\{1,2,\ldots,10\}$, and
for approximate Bayesian computation we use a $\log$ transformation
of $X_{1}(t)$ and $y(t)$ with $\epsilon=1$, i.e., 
\[
B_{\epsilon}(y)=\left\{ X_{1}(t)\;:\;\log\left\{ X_{1}(i)\right\} -\log\left\{ y(i)\right\} \leq\epsilon,\text{ for each }i\in\{1,\ldots,10\}\right\} .
\]

\begin{figure}
\begin{centering}
\subfloat{\includegraphics[scale=0.25]{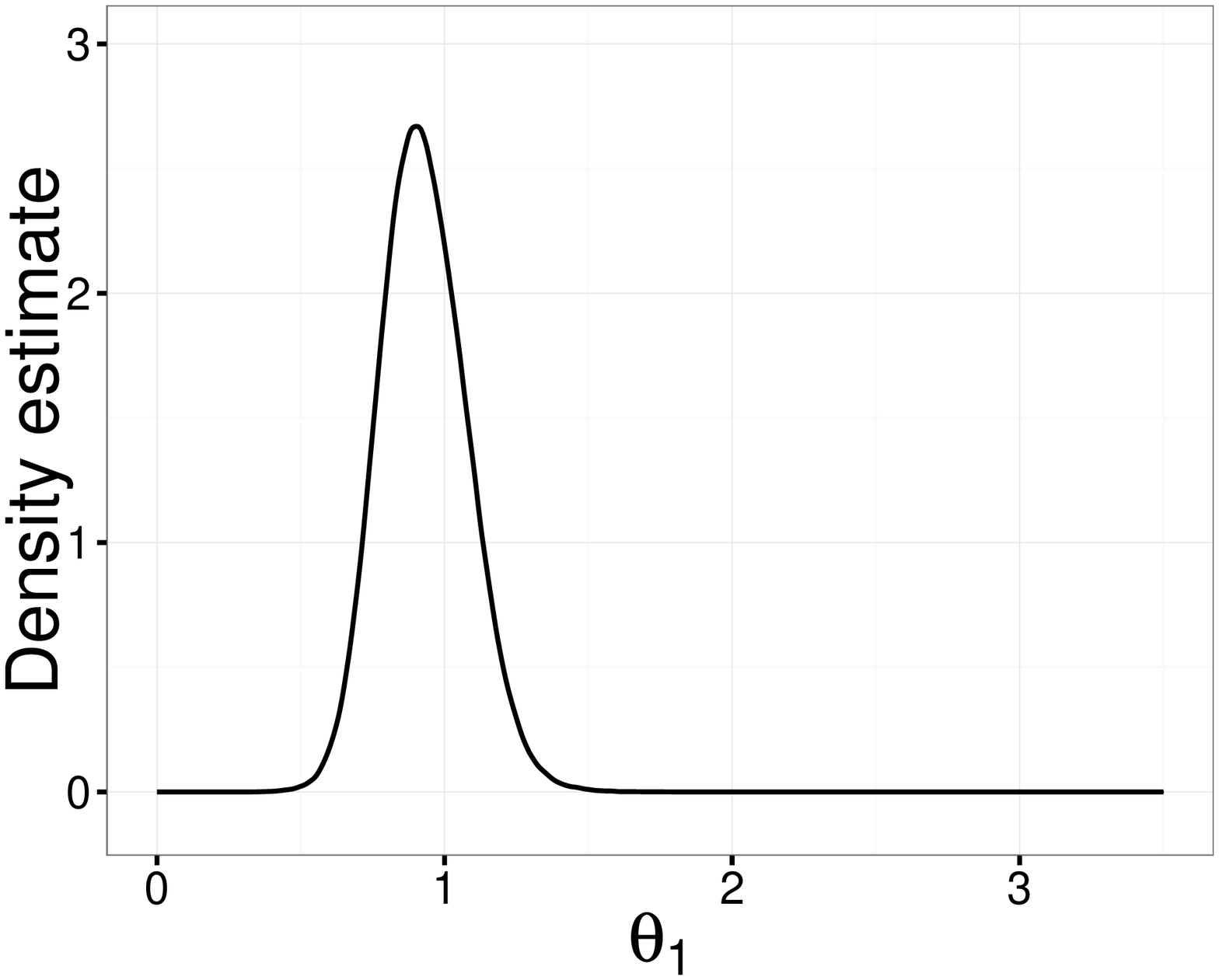}}\subfloat{\includegraphics[scale=0.25]{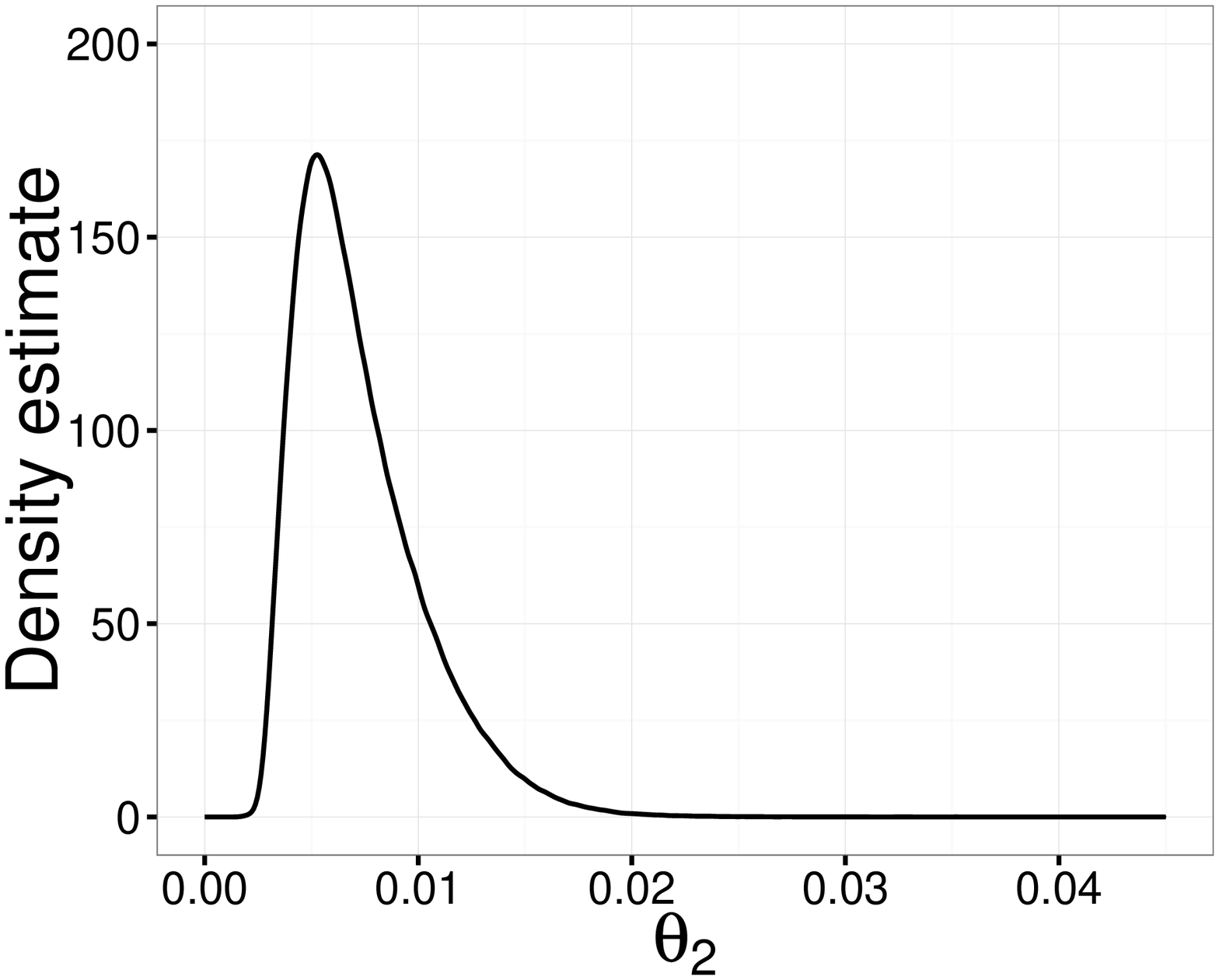}}\subfloat{\includegraphics[scale=0.25]{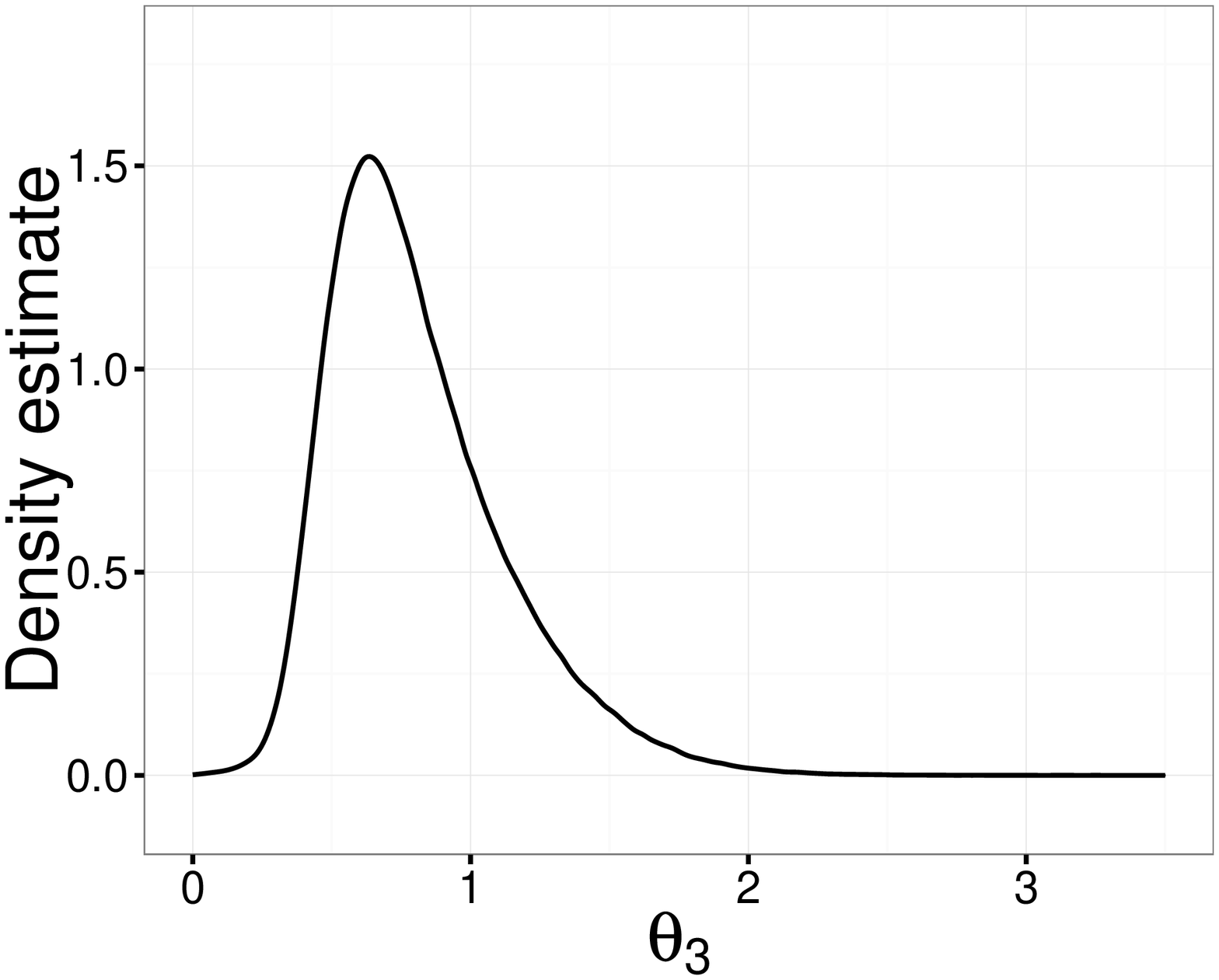}}
\par\end{centering}

\caption{Density estimates of the marginal posteriors for the Lotka--Volterra
model.\label{fig:lv_posteriors_normal}}
\end{figure}

\begin{figure}
\begin{centering}
\subfloat[$P_{\ref{alg:pm_kernel_1},1}$]{\includegraphics[scale=0.22]{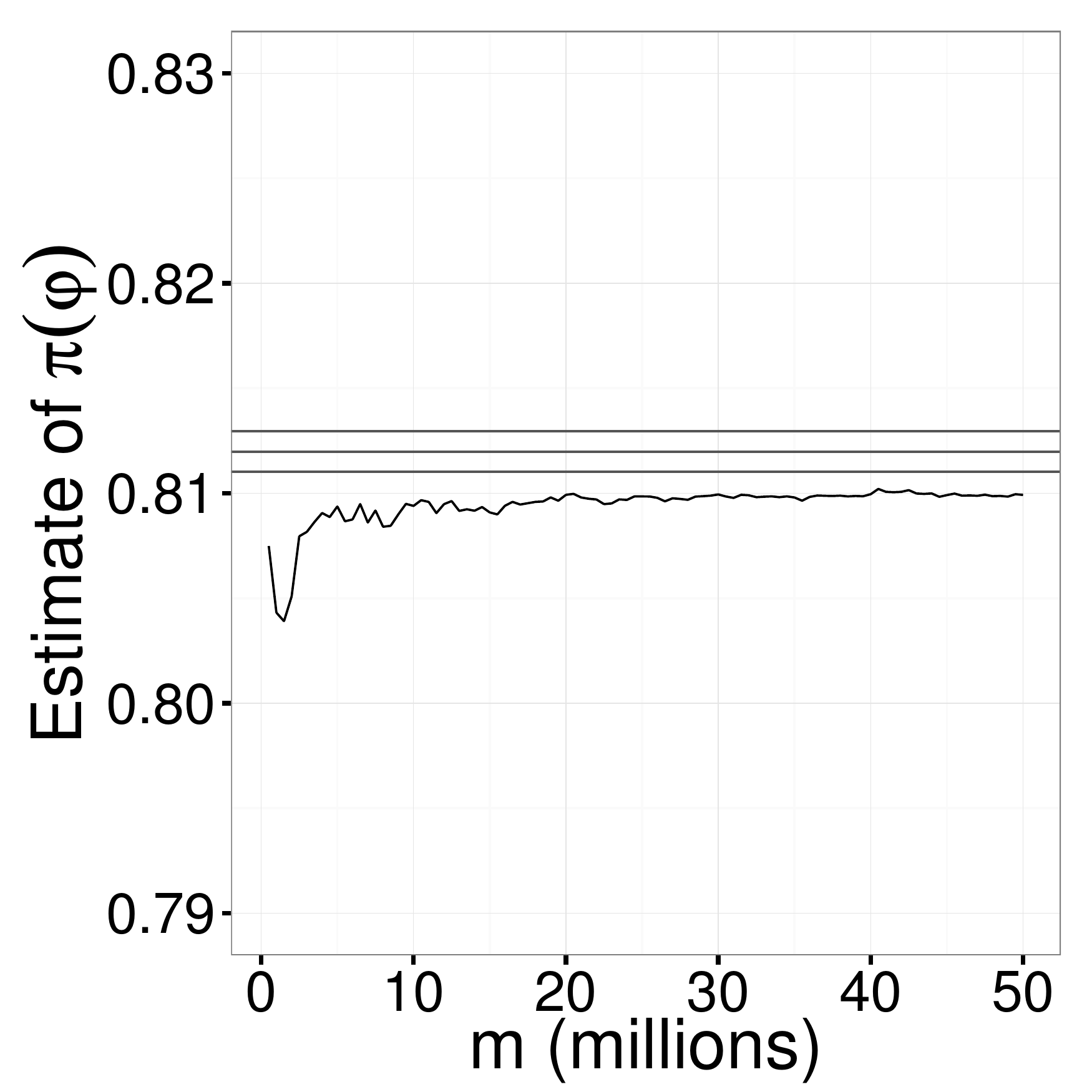}}\subfloat[$P_{\ref{alg:pm_kernel_1},15}$]{\includegraphics[scale=0.22]{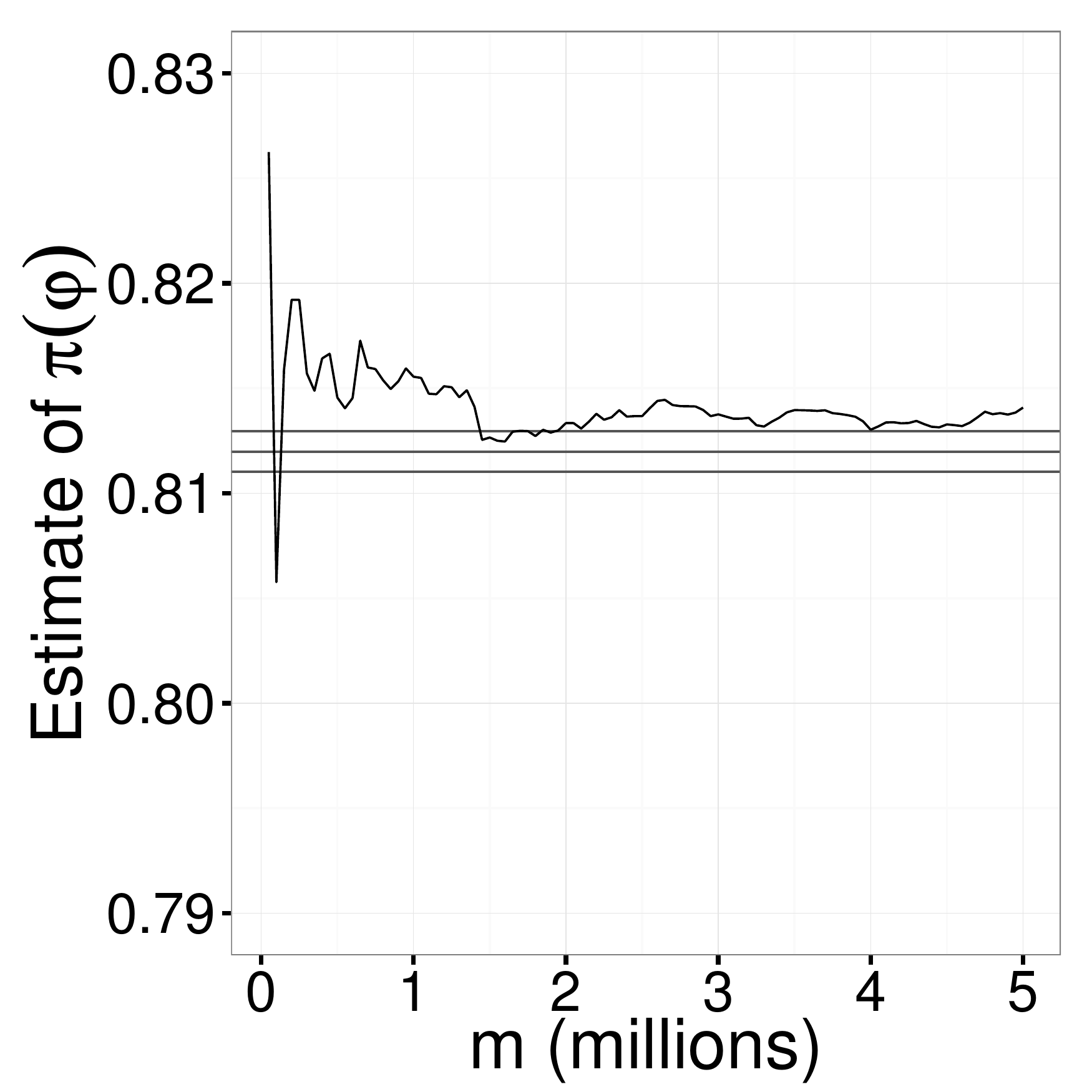}}\subfloat[$P_{\ref{alg:pm_kernel_2},15}$]{\includegraphics[scale=0.22]{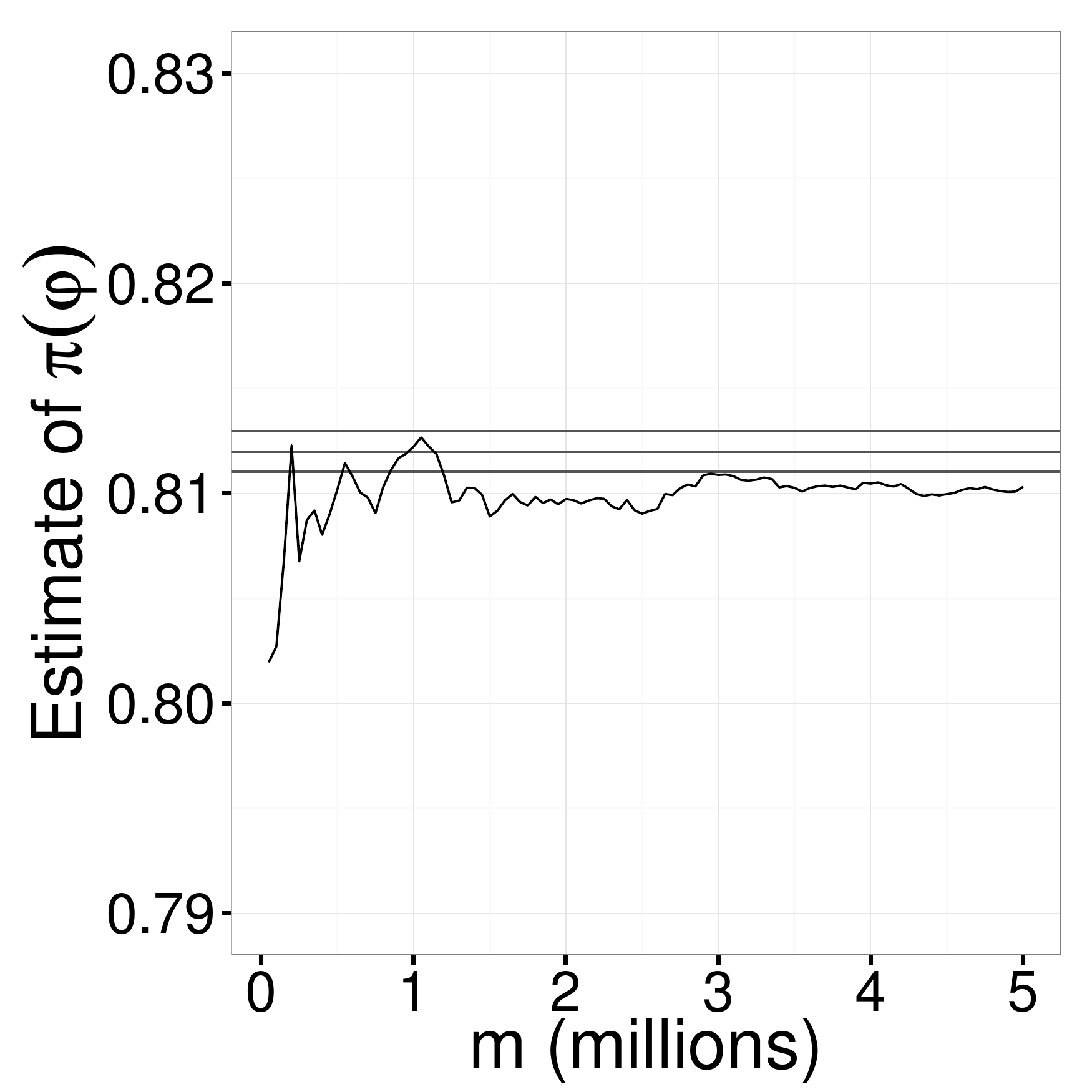}}\subfloat[$P_{\ref{alg:1_hit_kernel}}$]{\includegraphics[scale=0.22]{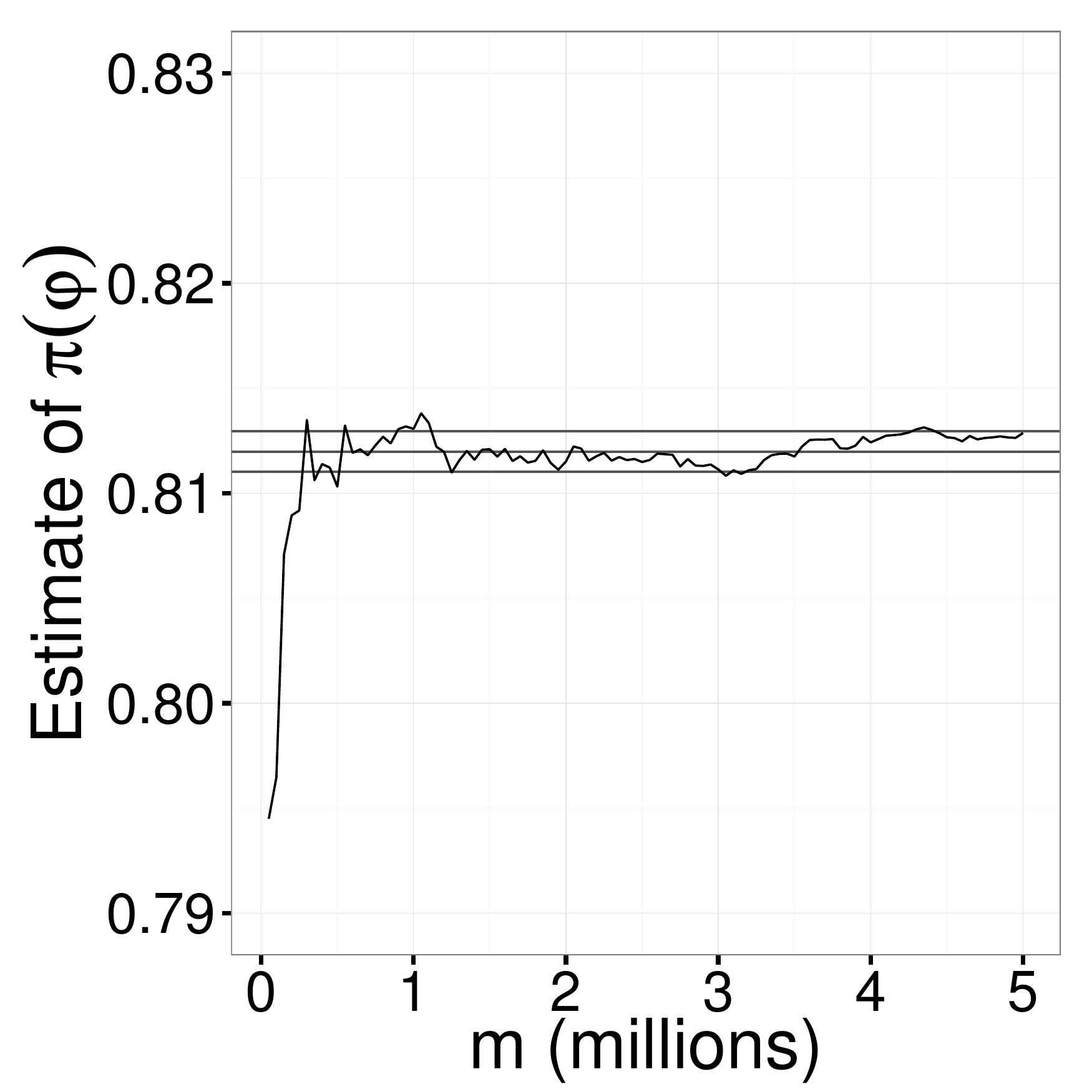}}
\par\end{centering}

\caption{Estimates of the posterior mean of $\theta_{3}$ by iteration using
each kernel. The three horizontal lines correspond to the estimate
obtained using the rejection sampler with two estimated standard deviations
added and subtracted.\label{fig:lv_partialsums_mean_normal}}
\end{figure}

\begin{figure}
\begin{centering}
\subfloat[$P_{\ref{alg:pm_kernel_1},1}$]{\includegraphics[scale=0.22]{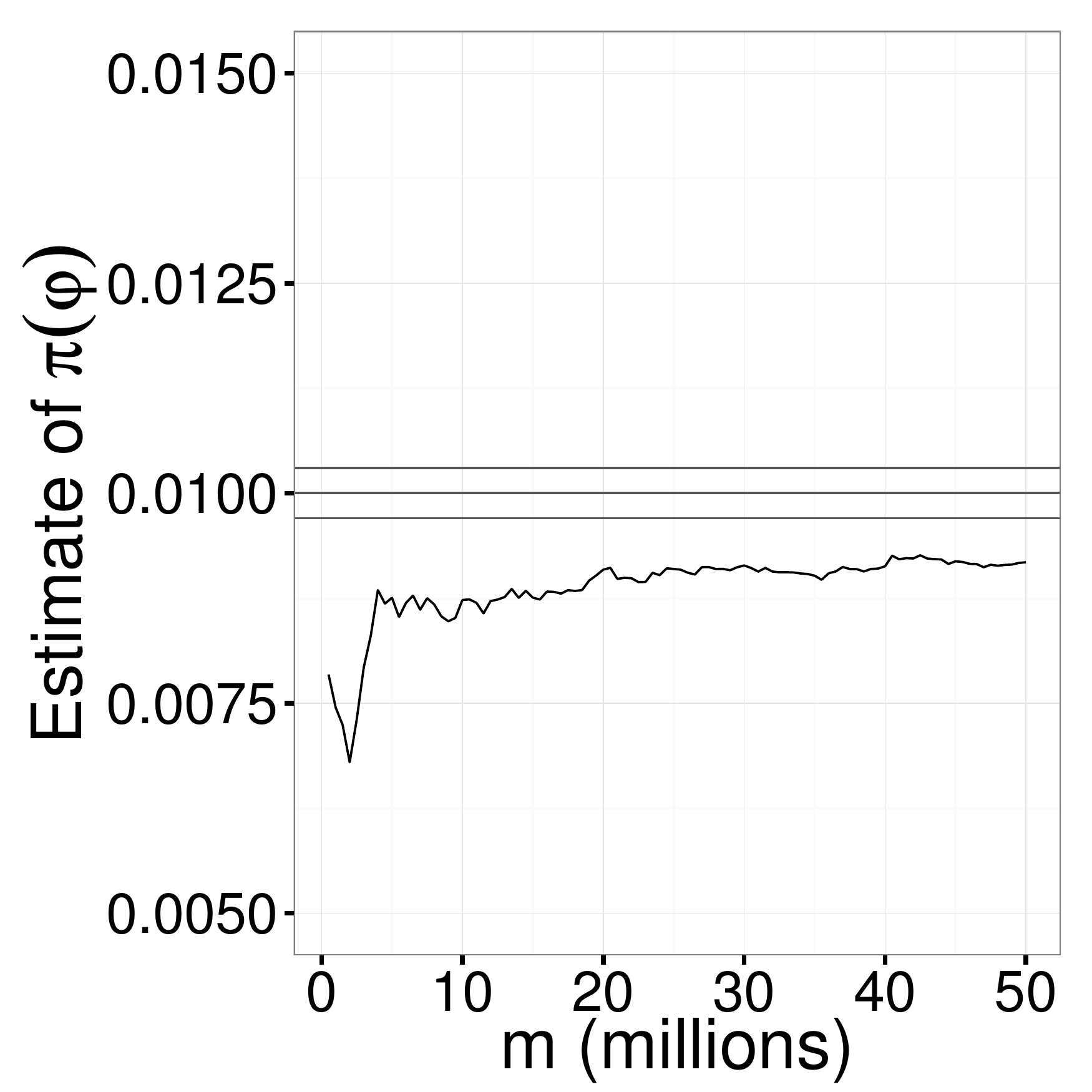}}\subfloat[$P_{\ref{alg:pm_kernel_1},15}$]{\includegraphics[scale=0.22]{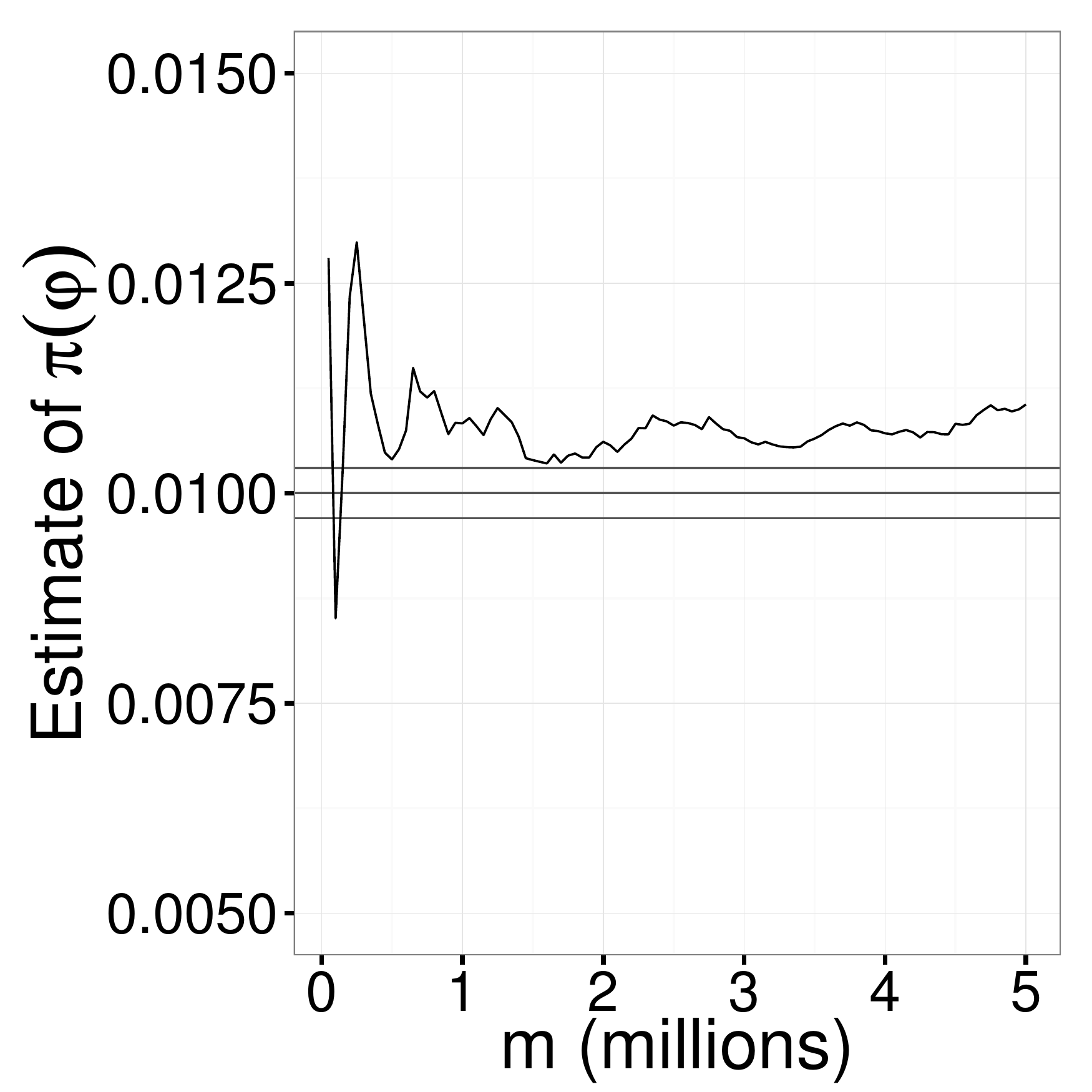}}\subfloat[$P_{\ref{alg:pm_kernel_2},15}$]{\includegraphics[scale=0.22]{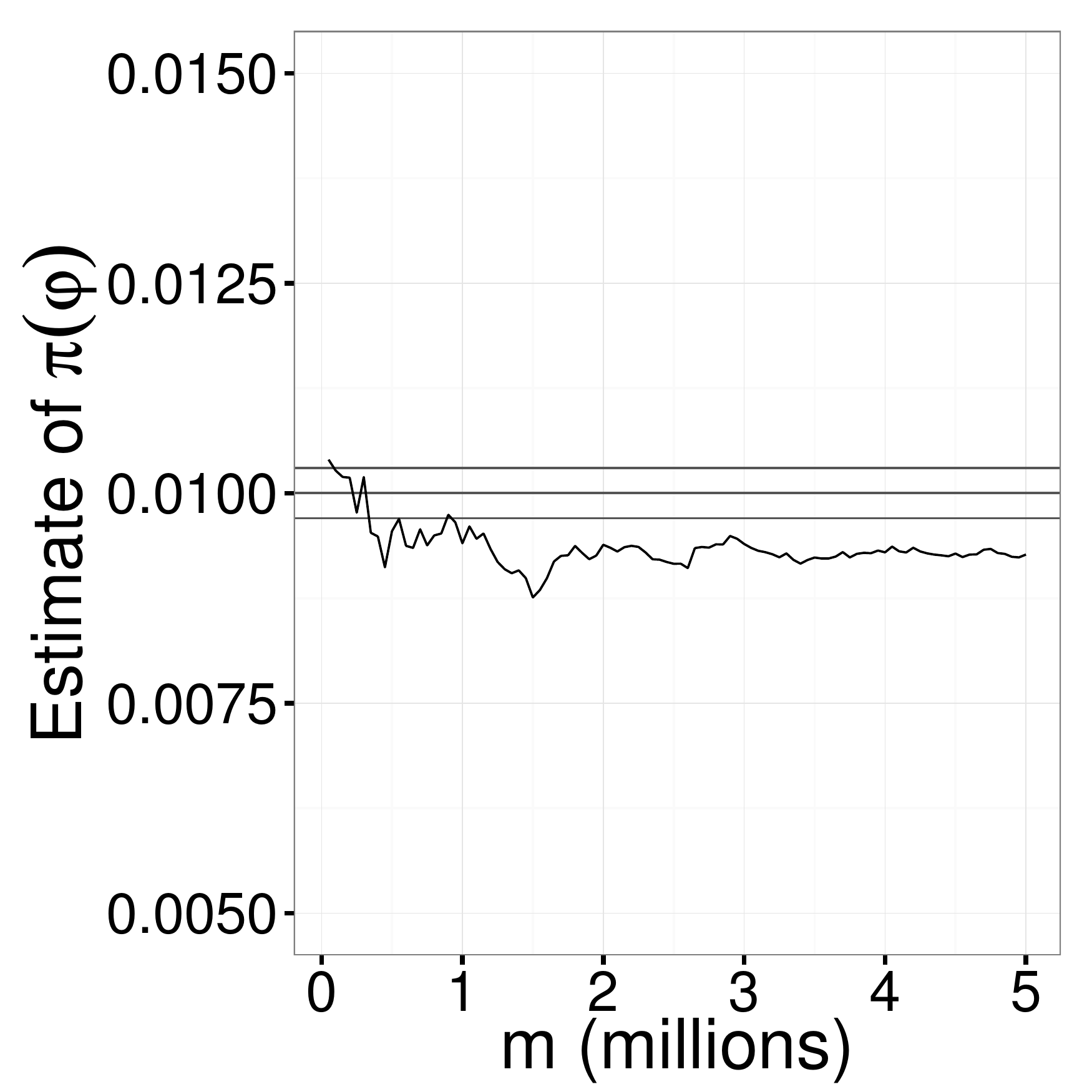}}\subfloat[$P_{\ref{alg:1_hit_kernel}}$]{\includegraphics[scale=0.22]{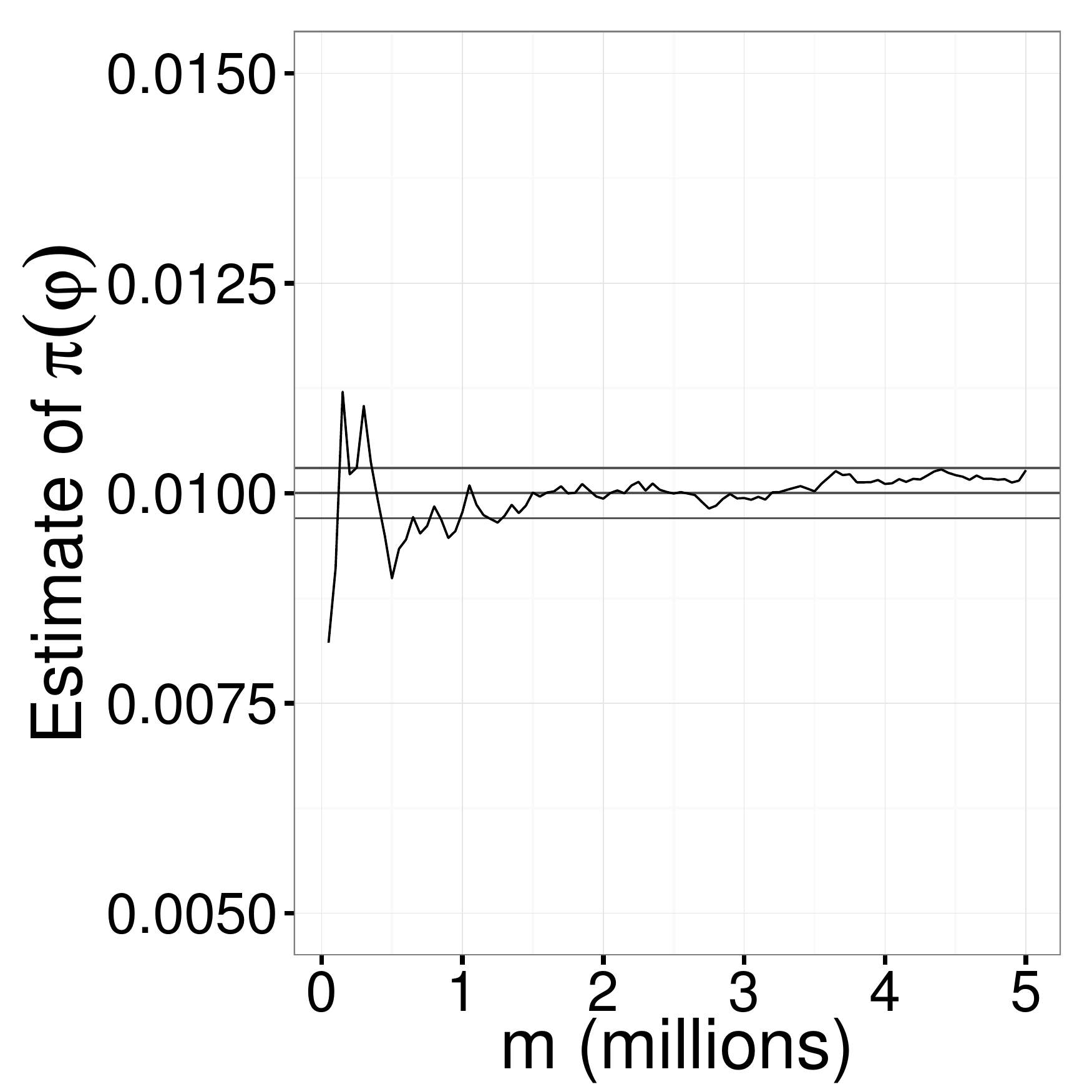}}
\par\end{centering}

\caption{Estimates of $\pi(\theta_{3}\geq1.79)$ by iteration using each kernel.
The three horizontal lines correspond to the estimate obtained using
the rejection sampler with two estimated standard deviations added
and subtracted.\label{fig:lv_partialsums_tail_normal}}
\end{figure}

We first model $\theta\in\Theta=[0,\infty)^{3}$ with $p(\theta)=100\exp(-\theta_{1}-100\theta_{2}-\theta_{3})$
and use $q(\theta,\vartheta)=\mathcal{N}(\vartheta;\theta,\Sigma)$
where $\Sigma={\rm diag}(.25,0.0025,.25)$. The choice of independent
exponential priors on $\theta$ is motivated by Condition~\ref{cond:c2}.
Density plots of the marginal posteriors for each component of $\theta$
are shown in Figure~\ref{fig:lv_posteriors_normal}, obtained using
$10^{6}$ samples from $\pi$ using a rejection sampler. $\theta_{1}$
has a tighter posterior than $\theta_{3}$ and while not shown here,
the samples indicate strong positive correlation between $\theta_{2}$
and $\theta_{3}$. In this setting, $P_{\ref{alg:1_hit_kernel}}$
for $5\times10^{6}$ iterations gave an average value of $n$ of $15$
and we also ran kernels $P_{\ref{alg:pm_kernel_1},1}=P_{\ref{alg:pm_kernel_2},1}$
for $5\times10^{7}$ iterations and $P_{\ref{alg:pm_kernel_1},15}$
and $P_{\ref{alg:pm_kernel_2},15}$ both for $5\times10^{6}$ iterations.
All kernels gave density estimates visibly indistinguishable from
those in Figure~\ref{fig:lv_posteriors_normal}, but inspection of
their partial sums by iteration reveals important differences. In
Figures~\ref{fig:lv_partialsums_mean_normal} and~\ref{fig:lv_partialsums_tail_normal}
we show estimates of the posterior mean of $\theta_{3}$ and the probability
that $\theta_{3}\geq1.79$ for each chain, accompanied by lines corresponding
to the estimate obtained using the samples from the rejection sampler.
The choice of $1.79$ corresponds to an estimate of the 90th percentile
using these latter samples. $P_{\ref{alg:1_hit_kernel}}$ seems to
accurately estimate both the same value as the estimate from the rejection
sampler and the uncertainty of the estimate seems to be correlated
with perturbations of the partial sum. However, the other kernels
seem to both miss the value of interest by some amount and, particularly
in the case of $P_{\ref{alg:pm_kernel_1},1}$, the perturbations of
the partial sum over time are small which may mislead practitioners
into believing the estimate has converged.

We performed a second analysis using a slightly different prior, with
$p(\theta)=0.01\exp(-\theta_{1}-0.01\theta_{2}-\theta_{3})$, where
differences in the kernels are accentuated. Here, the independent
prior for $\theta_{2}$ is all that has changed, and has been made
less informative. In this case, a rejection sampler cannot practically
be used to verify results as the expected number of proposals required
to obtain one sample by rejection is around $4.5\times10^{5}$. The
average value of $n$ for $P_{\ref{alg:1_hit_kernel}}$, however,
was $13$.

\begin{figure}
\begin{centering}
\subfloat[$P_{\ref{alg:pm_kernel_1},1}$]{\includegraphics[scale=0.22]{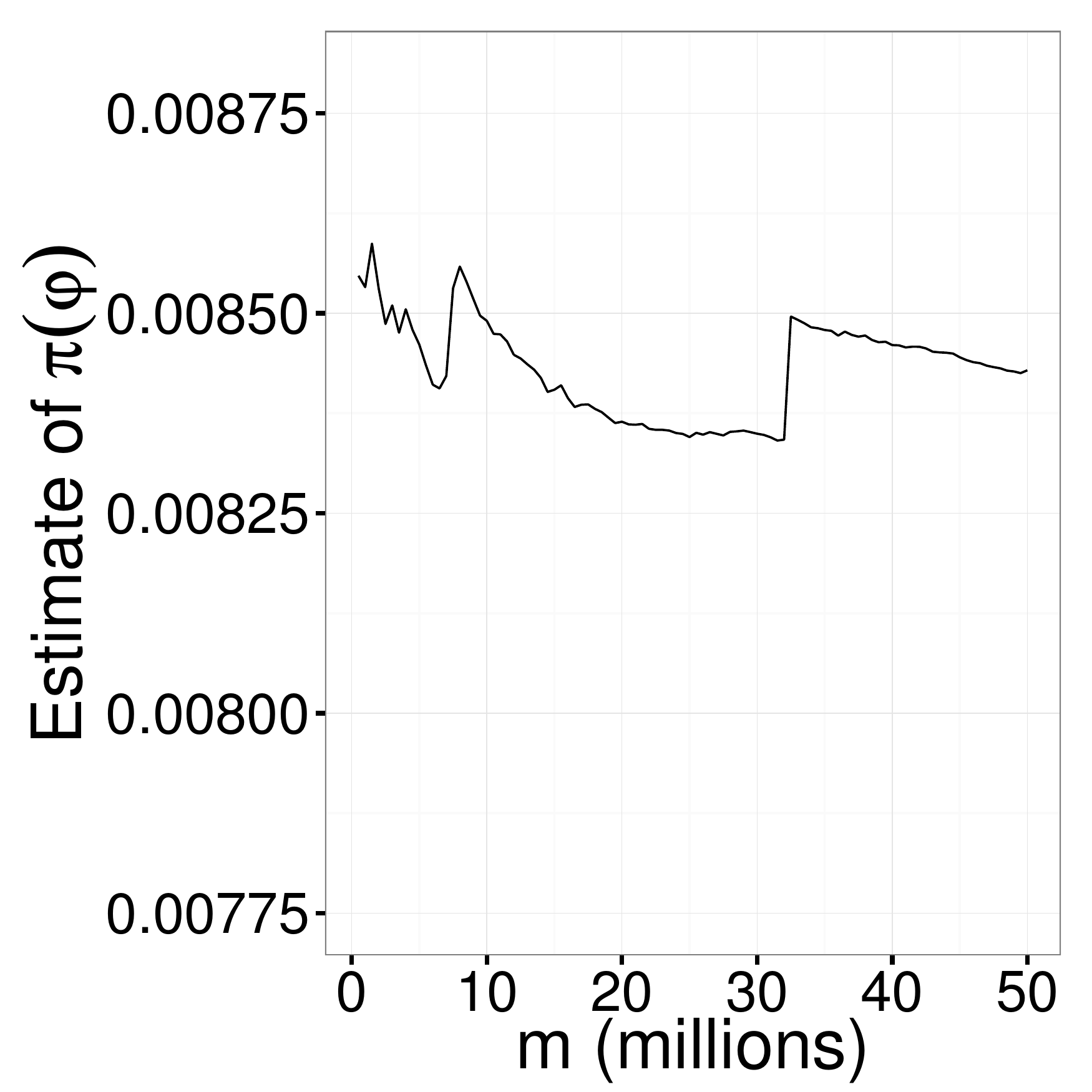}}\subfloat[$P_{\ref{alg:pm_kernel_1},15}$]{\includegraphics[scale=0.22]{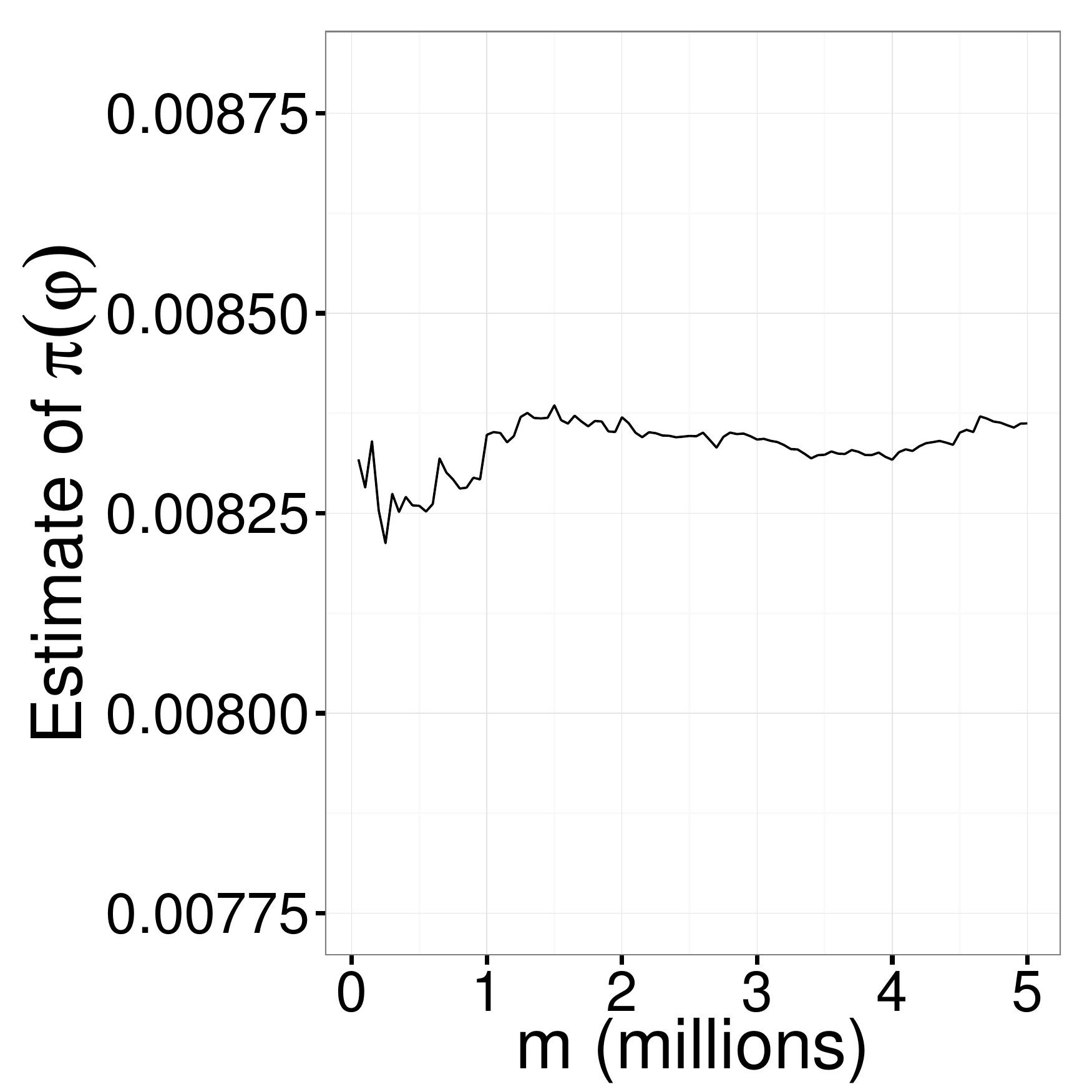}}\subfloat[$P_{\ref{alg:pm_kernel_2},15}$]{\includegraphics[scale=0.22]{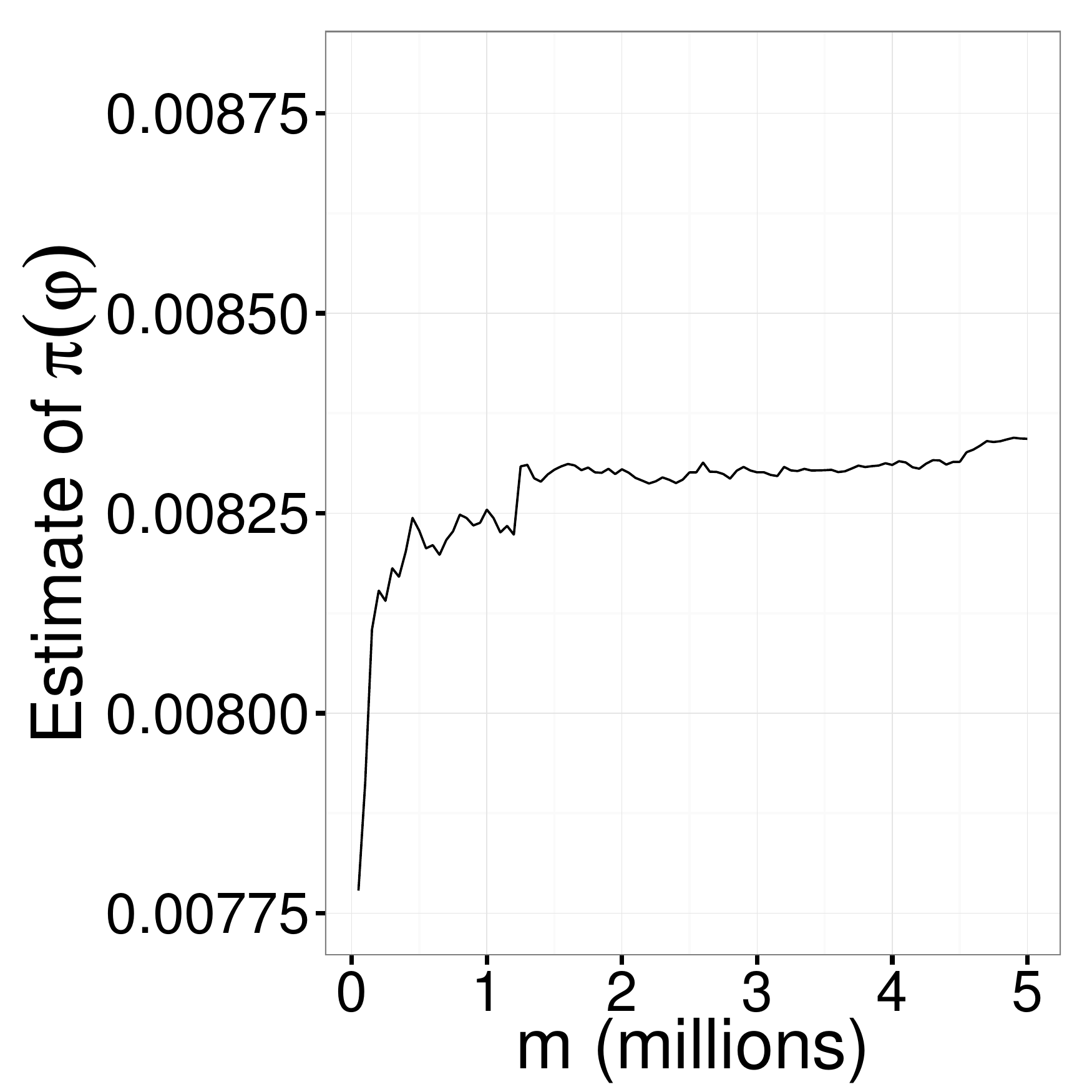}}\subfloat[$P_{\ref{alg:1_hit_kernel}}$]{\includegraphics[scale=0.22]{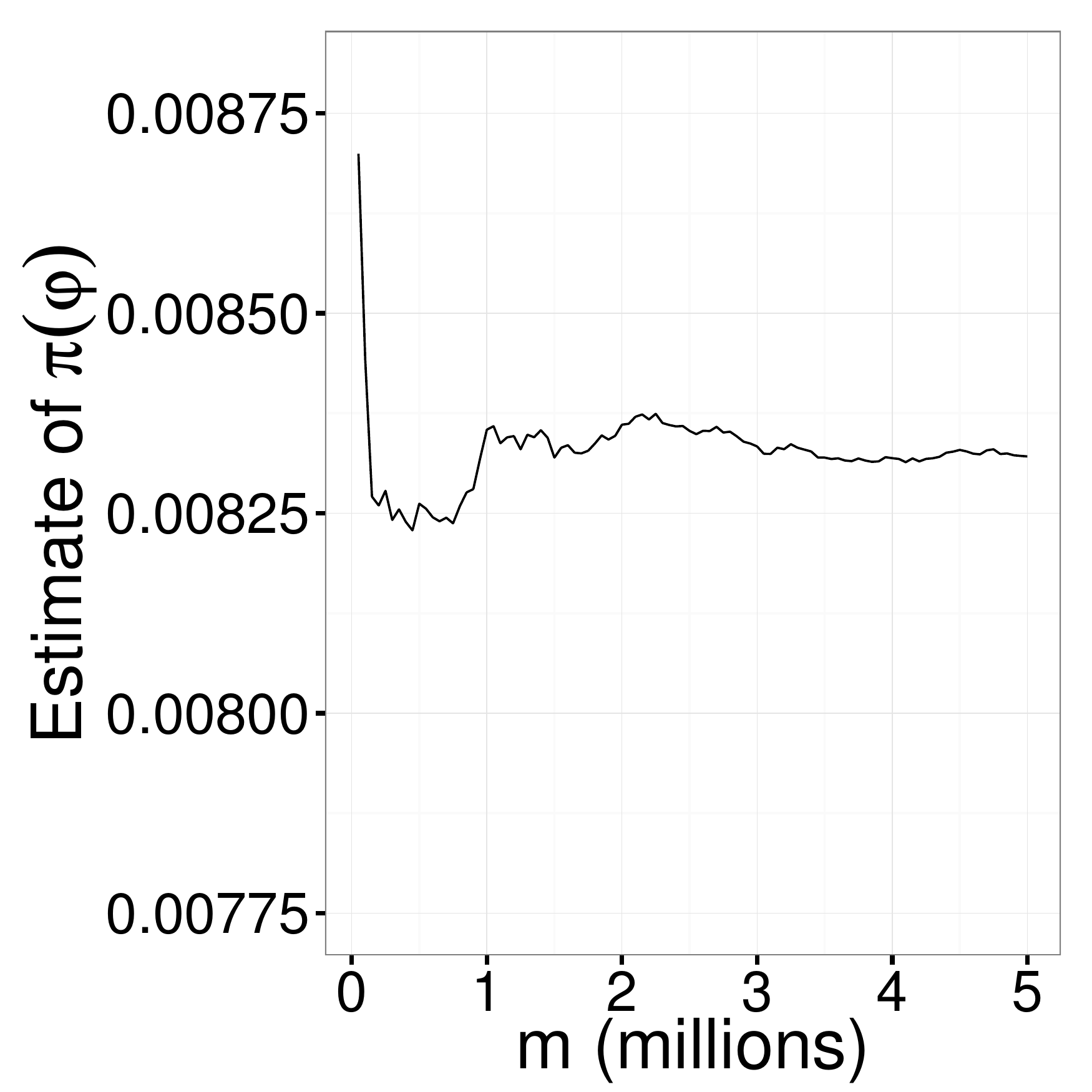}}
\par\end{centering}

\caption{Estimates of the posterior mean of $\theta_{2}$ by iteration using
each kernel.\label{fig:lv_partialsums_mean}}
\end{figure}

\begin{figure}
\begin{centering}
\subfloat[$P_{\ref{alg:pm_kernel_1},1}$]{\includegraphics[scale=0.22]{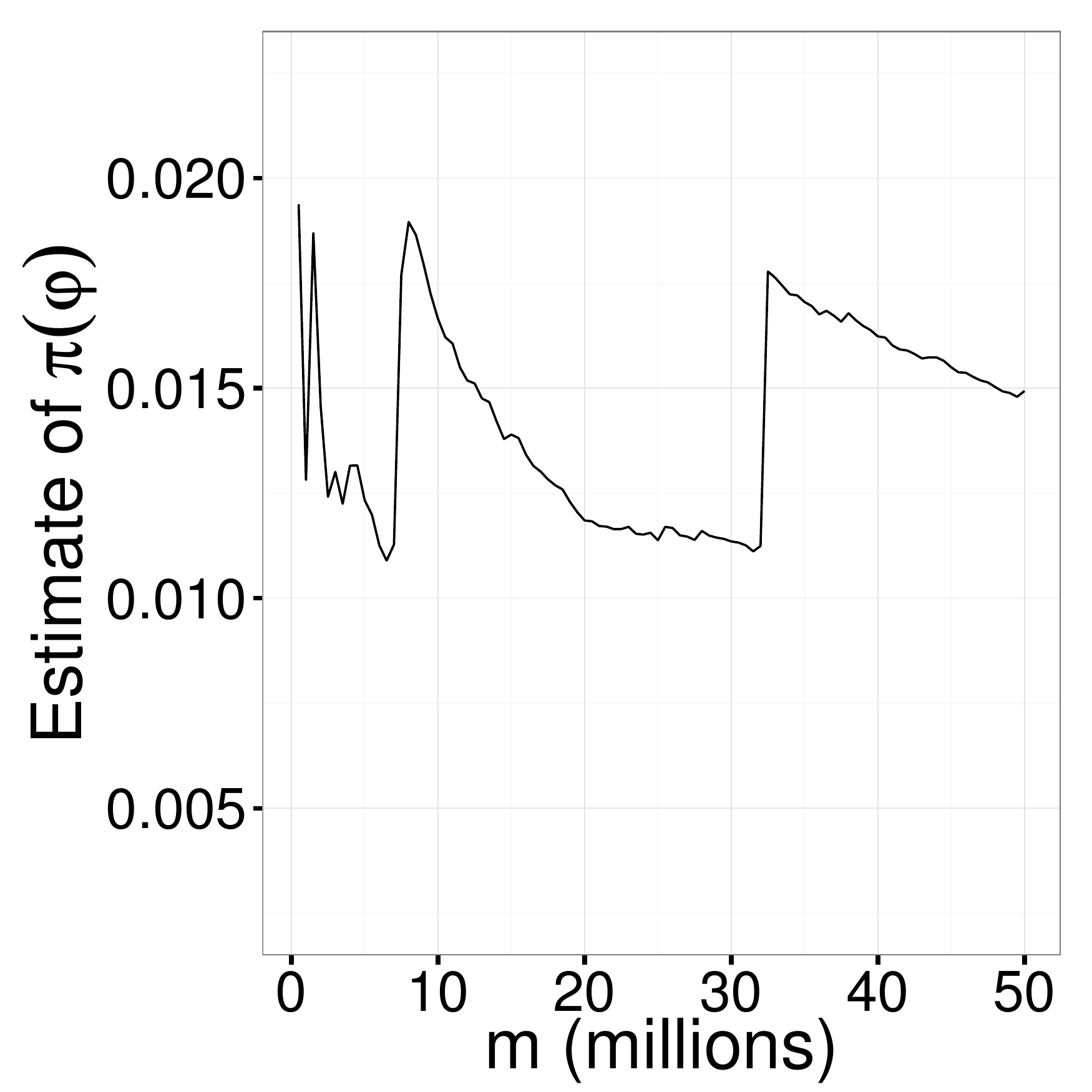}}\subfloat[$P_{\ref{alg:pm_kernel_1},15}$]{\includegraphics[scale=0.22]{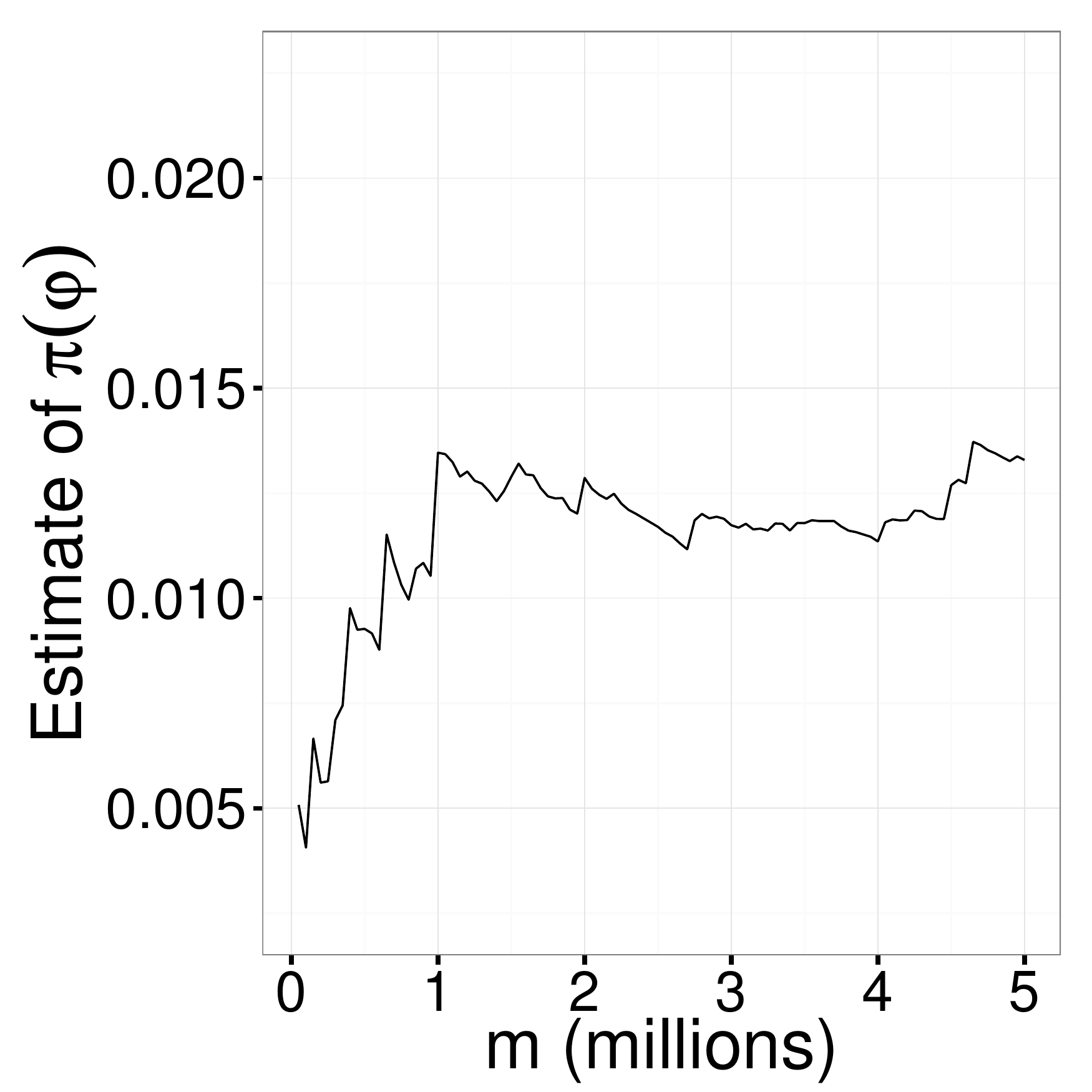}}\subfloat[$P_{\ref{alg:pm_kernel_2},15}$]{\includegraphics[scale=0.22]{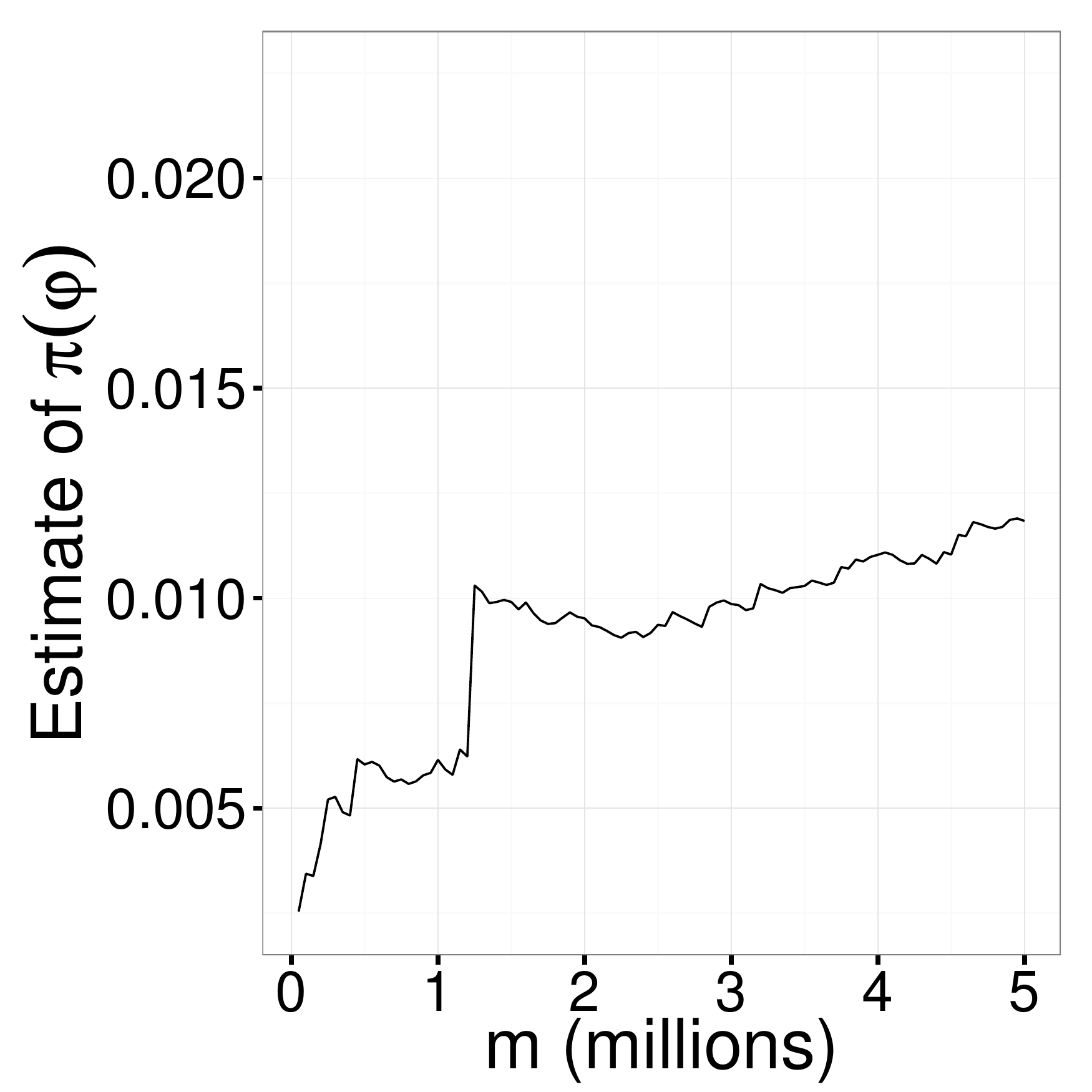}}\subfloat[$P_{\ref{alg:1_hit_kernel}}$]{\includegraphics[scale=0.22]{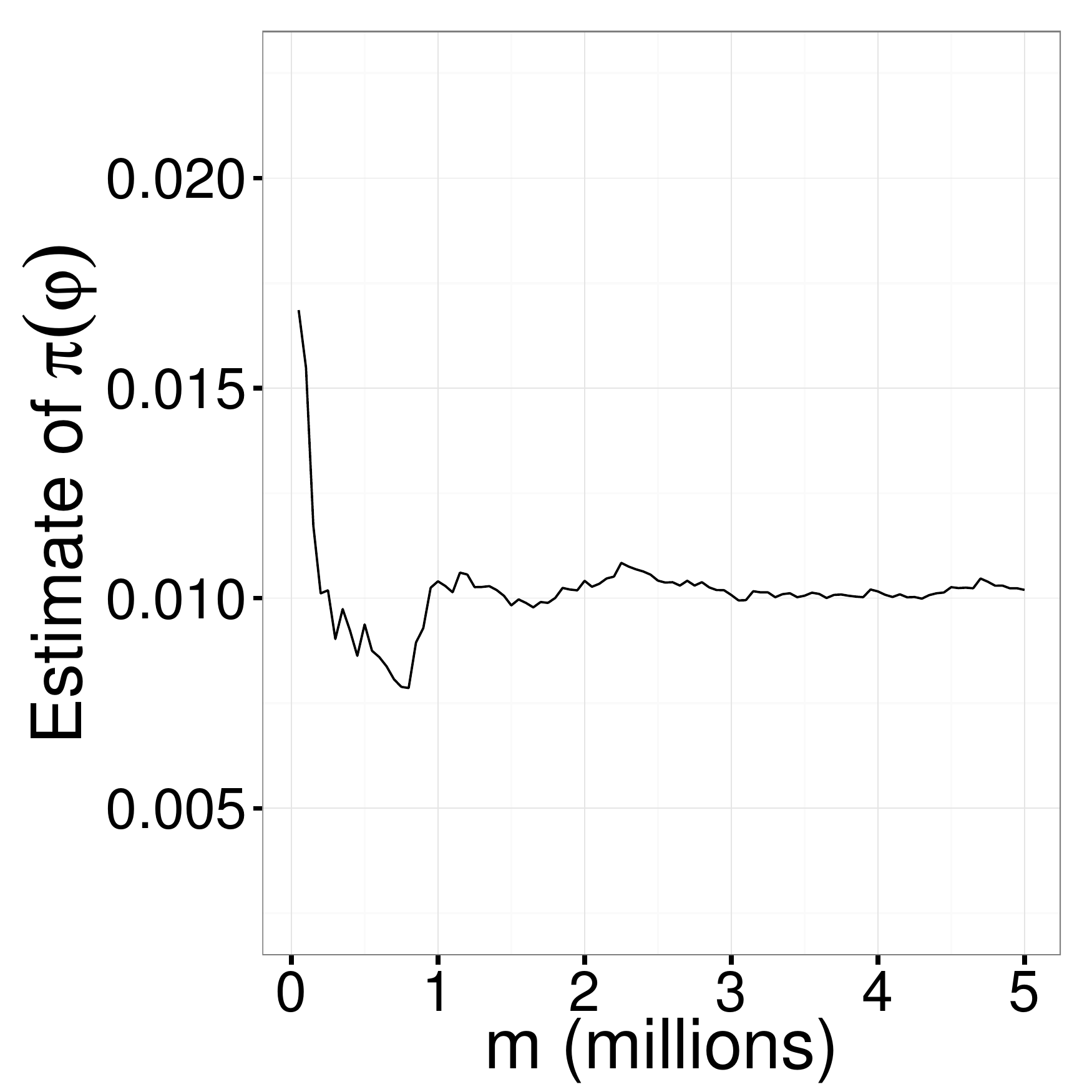}}
\par\end{centering}

\caption{Estimates of $\pi(\theta_{3}\geq2)$ by iteration using each kernel.\label{fig:lv_partialsums_tail}}
\end{figure}

While not shown here, marginal posterior density estimates using
each kernel for the parameters are reasonably close to those in Figure~\ref{fig:lv_posteriors_normal},
but those corresponding to $P_{\ref{alg:pm_kernel_1},1}$ exhibit
characteristic `bumps' in its tail. As above, we can inspect each
chain's corresponding partial sums by iteration to reveal important
differences. Figures~\ref{fig:lv_partialsums_mean} and~\ref{fig:lv_partialsums_tail}
show estimates of the posterior mean of $\theta_{2}$ and the posterior
probability that $\theta_{3}\geq2$ for each chain respectively, and
the latter is particularly illustrative of the inability of $P_{\ref{alg:pm_kernel_1}}$
and $P_{\ref{alg:pm_kernel_2}}$ to produce chains without long tail
excursions.

In practical applications such as this, it may not be possible to
determine easily if $P_{{\rm MH}}$ is variance bounding or geometrically
ergodic. However, Theorems~\ref{thm:vb}--\ref{thm:ge} do establish
that $P_{\ref{alg:1_hit_kernel}}$ will inherit either of these properties
from $P_{{\rm MH}}$ if it is. In practice, it is not unusual for
the conditions of Corollary~\ref{cor:jarner_hansen} to hold, and
one might expect them to do so here. Similarly, it is also quite common
for Condition~\ref{cond:c1} to hold, and so one might expect that
$P_{\ref{alg:pm_kernel_1}}$ and $P_{\ref{alg:pm_kernel_2}}$ are
not variance bounding here.

\section{Discussion}

Our analysis suggests that $P_{\ref{alg:1_hit_kernel}}$ may be geometrically
ergodic and/or variance bounding in a wide variety of situations where
kernels $P_{\ref{alg:pm_kernel_1},N}$ and $P_{\ref{alg:pm_kernel_2},N}$
are not. In practice, Condition~\ref{cond:c2} can be verified and
used to inform prior and proposal choice to ensure that $P_{\ref{alg:1_hit_kernel}}$
systematically inherits these properties from $P_{{\rm MH}}$. Of
course, variance bounding or geometric ergodicity of $P_{{\rm MH}}$
is often impossible to verify in the approximate Bayesian computation
setting due to the unknown nature of $f_{\theta}$. However, a prior
with regular contours as per \eqref{eq:regular_contours} will ensure
that $P_{{\rm MH}}$ is geometrically ergodic if $f_{\theta}$ decays
super-exponentially and also has regular contours. In addition, Condition~\ref{cond:c2}
is stronger than necessary but tighter conditions are likely to be
complicated and may require case-by-case treatment.

The combination of Theorems \ref{thm:nvb}--\ref{thm:vb} and Proposition
\ref{prop:n}, whose assumptions are not mutually exclusive, allow
us to conclude that the behaviour of $P_{\ref{alg:1_hit_kernel}}$
is characteristically different to $P_{\ref{alg:pm_kernel_1},N}$
and $P_{\ref{alg:pm_kernel_2},N}$ in some settings. In particular,
the use of a larger expected number of simulations from $f_{\theta}$
and $f_{\vartheta}$ in the tails of $\pi$ using $P_{\ref{alg:1_hit_kernel}}$
could be viewed as analogous to being ``stuck'' for many iterations
in the tails of $\pi$ using $P_{\ref{alg:1_hit_kernel}}$ or $P_{\ref{alg:pm_kernel_2},N}$.
However, while both the expected number of simulations and the asymptotic
variance of \eqref{eq:mcmc_estimate} for any $\varphi\in L^{2}(\pi)$
are finite under $P_{\ref{alg:1_hit_kernel}}$ under the conditions
of Theorem \ref{thm:vb}, there are $\varphi\in L^{2}(\pi)$ for which
a central limit theorem does not hold for \eqref{eq:mcmc_estimate}
when using $P_{\ref{alg:pm_kernel_1},N}$ or $P_{\ref{alg:pm_kernel_2},N}$
under the conditions of Theorem \ref{thm:nvb}.

Variance bounding and geometric ergodicity are likely to coincide
in most applications of interest, as variance bounding but non-geometrically
ergodic Metropolis--Hastings kernels exhibit periodic behaviour rarely
encountered in statistical inference. Bounds on the second largest
eigenvalue and/or spectral gap of $P_{\ref{alg:1_hit_kernel}}$ in
relation to properties of $P_{{\rm MH}}$ could be obtained through
Cheeger-like inequalities using conductance arguments as in the proofs
of Theorems~\ref{thm:vb} and~\ref{thm:ge}, although these may
be quite loose in some situations \citep[see, e.g., ][]{diaconis1991geometric}
and we have not pursued them here. Finally, \citet{robertsquantitative}
have demonstrated that some simple Markov chains that are not geometrically
ergodic can converge extremely slowly and that properties of such
algorithms can be very sensitive to even slight parameter changes.

The theoretical results obtained in Section \ref{sec:Theoretical-properties}
and the examples in Section \ref{sec:Examples} provide some understanding
of the relative qualitative merits of $P_{\ref{alg:1_hit_kernel}}$
over $P_{\ref{alg:pm_kernel_1},N}$ and $P_{\ref{alg:pm_kernel_2},N}$.
However, the results do not prove that $P_{\ref{alg:1_hit_kernel}}$
should necessarily be uniformly preferred over $P_{\ref{alg:pm_kernel_2},N}$,
although the examples do suggest that it may have better asymptotic
variance properties when taking cost of simulations into account in
a variety of scenarios. In addition, Theorem~\ref{thm:mix_geo_nongeo}
can be used to justify its mixture with alternative reversible kernels
such as $P_{\ref{alg:pm_kernel_2},N}$ if desired.

\section*{Acknowledgement}

Lee acknowledges the support of the Centre for Research in Statistical
Methodology. Łatuszyński would like to thank the Engineering and Physical
Sciences Research Council, U.K. We are grateful to Arnaud Doucet and
Gareth Roberts for helpful comments.

\section*{}

\appendix

\section{Proofs\label{sec:Proofs_main}}

Many of our proofs make use of the relationship between conductance,
the spectrum of a Markov kernel, and variance bounding for reversible
Markov kernels $P$. In particular, conductance $\kappa>0$ is equivalent
to $\sup S(P)<1$ \citep[Theorem 2.1]{lawler1988bounds}, which as
stated earlier is equivalent to variance bounding. Conductance $\kappa$
for a $\pi$-invariant, transition kernel $P$ on $\Theta$ is defined
as 
\[
\kappa=\inf_{A:0<\pi(A)\leq1/2}\kappa(A),\qquad\kappa(A)=\pi(A)^{-1}\int_{A}P(\theta,A^{\complement})\pi(\mathrm{d}\theta)=\int_{\Theta}P(\theta,A^{\complement})\pi_{A}(\mathrm{d}\theta),
\]
where $\pi_{A}(\mathrm{d}\theta)=\pi(\mathrm{d}\theta)\mathbf{1}_{A}(\theta)/\pi(A)$.

Finally, we make use of the fact that if $q\in\mathcal{Q}$ we can
define the function 
\[
r_{q}(\delta)=\inf\left\{ r:\text{ for all }\theta\in\Theta,\; q\left(\theta,B_{r,\theta}^{\complement}\right)<\delta\right\} .
\]

\begin{proof}[Proof of Theorem~\ref{thm:gen_nvb}]
If $\nu-{\rm ess}\sup_{\theta}P(\theta,\{\theta\})=1$ and $P(\theta,\{\theta\})$
is measurable, then the set $A_{\tau}=\{\theta\in\Theta\,:\, P(\theta,\{\theta\})\geq1-\tau\}$
is measurable and $\nu(A_{\tau})>0$ for every $\tau>0$. Moreover,
$a_{0}=\lim_{\tau\searrow0}\nu(A_{\tau})$ exists, since $A_{\tau_{2}}\subset A_{\tau_{1}}$
for $\tau_{2}<\tau_{1}$. Now, assume $a_{0}>0$, and define $A_{0}=\{\theta\in\Theta\::\: P(\theta,\{\theta\})=1\}=\bigcap_{n}A_{\tau_{n}}$
where $\tau_{n}\searrow0$. By continuity from above $\nu(A_{0})=a_{0}>0$
and since $\nu$ is not concentrated at a single point, $P$ is reducible,
which is a contradiction. Hence $a_{0}=0$. Consequently, by taking
$\tau_{n}\searrow0$ with $\tau_{1}$ small enough, we have $\nu(A_{\tau_{n}})<1/2$
for every $n$, and can upper bound the conductance of $P$ by 
\[
\kappa\leq\lim_{n}\kappa(A_{\tau_{n}})=\lim_{n}\int_{A_{\tau_{n}}}P(\theta,A_{\tau_{n}}^{\complement})\nu_{A_{\tau_{n}}}(\mathrm{d}\theta)\leq\lim_{n}\int_{A_{\tau_{n}}}P(\theta,\{\theta\}^{\complement})\nu_{A_{\tau_{n}}}(\mathrm{d}\theta)=\lim_{n}\tau_{n}=0.
\]
Therefore $P\notin\mathcal{V}$.
\end{proof}

\begin{proof}[Proof of Theorem~\ref{thm:nvb}]
We prove the result for $P_{\ref{alg:pm_kernel_2},N}$. The proof
for $P_{\ref{alg:pm_kernel_1},N}$ is essentially identical, with
minor adjustments for the extended state space, and is omitted. By
Theorem~\ref{thm:gen_nvb}, it suffices to show that $\pi-{\rm ess}\sup_{\theta}P_{\ref{alg:pm_kernel_2},N}(\theta,\{\theta\})=1$,
i.e., for all $\tau>0$, there exists $A\subseteq\Theta$ with $\pi(A)>0$
such that for all $\theta\in A$, $P_{\ref{alg:pm_kernel_2},N}(\theta,\{\theta\}^{\complement})\leq\tau$. 

From Condition~\ref{cond:c1}, $q\in\mathcal{Q}$. Given $\tau>0$,
let $r=r_{q}(\tau/2)$, $v=\inf\left\{ v:\sup_{\theta\in B_{v}^{c}(0)}h(\theta)<1-\left(1-\tau/2\right){}^{1/N}\right\} $
and $A=B_{v+r,0}^{\complement}$. From Condition~\ref{cond:c1},
$\pi(A)>0$ and using \eqref{eq:alpha_kernel} and \eqref{eq:alpha_pm2},
for all $\theta\in A$,
\begin{align*}
P_{\ref{alg:pm_kernel_2},N}(\theta,\{\theta\}^{\complement}) & =\int_{\{\theta\}^{\complement}}\int_{\mathsf{Y}^{N}}\int_{\mathsf{Y}^{N-1}}\left[1\wedge\frac{c(\vartheta,\theta)\sum_{j=1}^{N}w(z_{j})}{c(\theta,\vartheta)\left\{ 1+\sum_{j=1}^{N-1}w(x_{j})\right\} }\right]f_{\theta}^{\otimes N-1}({\rm d}x_{1:N-1})f_{\vartheta}^{\otimes N}({\rm d}z_{1:N})q(\theta,{\rm d}\vartheta)\\
 & \leq\sup_{\theta\in\Theta}q\left(\theta,B_{r,\theta}^{\complement}\right)+\int_{B_{r,\theta}}\int_{\mathsf{Y}^{N}}I\left\{ \sum_{i=1}^{N}w(z_{i})\geq1\right\} f_{\vartheta}^{\otimes N}({\rm d}z_{1:N})q(\theta,{\rm d}\vartheta)\\
 & \leq\frac{\tau}{2}+\int_{B_{r,\theta}}\left[1-\left\{ 1-\sup_{\vartheta\in B_{r,\theta}}h(\vartheta)\right\} ^{N}\right]q(\theta,{\rm d}\vartheta)\leq\tau.
\end{align*}

\end{proof}
The following two Lemmas are pivotal in the proofs of Proposition
\ref{prop:vb_necessity} and Theorems~\ref{thm:vb} and~\ref{thm:ge},
and make extensive use of \eqref{eq:alpha_kernel}, \eqref{eq:alpha_1hit}
and \eqref{eq:alpha_mh}. Their proofs can be found in Appendix~\ref{sec:proofs}.
\begin{lem}
\label{lem:diagonal_dominance}\textup{$P_{\ref{alg:1_hit_kernel}}(\theta,\{\theta\})\geq P_{{\rm MH}}(\theta,\{\theta\})$.}
\end{lem}

\begin{lem}
\label{lem:kernel_bound}Assume Condition~\ref{cond:c2}\textup{.
For $\pi$-almost all $\theta$ and any $A\subseteq\Theta$ such that
$\theta\in A$ and $r>0$, 
\[
P_{{\rm MH}}(\theta,A^{\complement})\leq\sup_{\theta}q(\theta,B_{r,\theta}^{\complement})+(1+M_{r})P_{\ref{alg:1_hit_kernel}}(\theta,A^{\complement}),
\]
}where $M_{r}$ is as defined in Condition~\ref{cond:c2}.\end{lem}
\begin{proof}[Proof of Theorem~\ref{thm:vb}]
We prove the result under Condition~\ref{cond:c2}. Let $\kappa_{{\rm MH}}$
and $\kappa_{\ref{alg:1_hit_kernel}}$ be the conductance of $P_{{\rm MH}}$
and $P_{\ref{alg:1_hit_kernel}}$ respectively, and $A$ be a measurable
set with $\pi(A)>0$. Since $q\in\mathcal{Q}$ we let $R=r_{q}(\kappa_{{\rm MH}}/2)$
and $M_{R}$ be as in Condition~\ref{cond:c2}. Then by Lemma~\ref{lem:kernel_bound}
we have
\begin{align*}
\kappa_{{\rm MH}}(A) & =\int_{\Theta}P_{{\rm MH}}(\theta,A^{\complement})\pi_{A}(\mathrm{d}\theta)\leq\frac{\kappa_{{\rm MH}}}{2}+(1+M_{R})\int_{\Theta}P_{\ref{alg:1_hit_kernel}}(\theta,A^{\complement})\pi_{A}(\mathrm{d}\theta)\\
 & =\frac{\kappa_{{\rm MH}}}{2}+(1+M_{R})\kappa_{\ref{alg:1_hit_kernel}}(A).
\end{align*}
Since $A$ is arbitrary, we conclude that $\kappa_{{\rm MH}}\leq2(1+M_{R})\kappa_{\ref{alg:1_hit_kernel}}$
so $\kappa_{{\rm MH}}>0\Rightarrow\kappa_{\ref{alg:1_hit_kernel}}>0$.
\end{proof}

\section{Supplementary proofs\label{sec:proofs}}
\begin{proof}[Proof of Proposition~\ref{prop:vb_necessity}]
Lemma~\ref{lem:diagonal_dominance} gives $P_{\ref{alg:1_hit_kernel}}\preceq P_{{\rm MH}}$
in the sense of \citep{Peskun1973,Tierney1998} and so ${\rm var}(P_{\ref{alg:1_hit_kernel}},\varphi)\geq{\rm var}(P_{{\rm MH}},\varphi)$.
By \citet[Theorem 8]{roberts2008variance}, $P_{\ref{alg:1_hit_kernel}}\preceq P_{{\rm MH}}\Longrightarrow(P_{\ref{alg:1_hit_kernel}}\in\mathcal{V}\Rightarrow P_{{\rm MH}}\in\mathcal{V})$.
\end{proof}

\begin{proof}[Proof of Lemma~\ref{lem:diagonal_dominance}]
We show that for any $(\theta,\vartheta)$, $\alpha_{\ref{alg:1_hit_kernel}}(\theta,\vartheta)\leq\alpha_{{\rm MH}}(\theta,\vartheta)$.
Consider the case $c(\vartheta,\theta)\leq c(\theta,\vartheta)$.
Then since $h(\theta)\leq1$,
\begin{align*}
\alpha_{\ref{alg:1_hit_kernel}}(\theta,\vartheta) & =\frac{c(\vartheta,\theta)}{c(\theta,\vartheta)}\frac{h(\vartheta)}{h(\vartheta)+h(\theta)-h(\vartheta)h(\theta)}\leq1\wedge\frac{c(\vartheta,\theta)h(\vartheta)}{c(\theta,\vartheta)h(\theta)}=\alpha_{{\rm MH}}(\theta,\vartheta).
\end{align*}
Similarly, if $c(\vartheta,\theta)>c(\theta,\vartheta)$, we have
\[
\alpha_{\ref{alg:1_hit_kernel}}(\theta,\vartheta)=\frac{h(\vartheta)}{h(\vartheta)+h(\theta)-h(\vartheta)h(\theta)}\leq1\wedge\frac{c(\vartheta,\theta)h(\vartheta)}{c(\theta,\vartheta)h(\theta)}=\alpha_{{\rm MH}}(\theta,\vartheta).
\]
 This immediately implies $P_{\ref{alg:1_hit_kernel}}(\theta,\{\theta\})\geq P_{{\rm MH}}(\theta,\{\theta\})$
since $P(\theta,\{\theta\})=1-\int_{\Theta\setminus\{\theta\}}q(\theta,\vartheta)\alpha(\theta,\vartheta)d\vartheta$.
\end{proof}

\begin{proof}[Proof of Lemma~\ref{lem:kernel_bound}]
We begin by showing that for $\vartheta\in B_{r}(\theta)$ and $\vartheta\neq\theta$,
\begin{align}
\alpha_{{\rm MH}}(\theta,\vartheta) & \leq(1+M_{r})\alpha_{\ref{alg:1_hit_kernel}}(\theta,\vartheta).\label{eq:alpha_bound}
\end{align}
First we deal with the case $h(\vartheta)p(\vartheta)q(\vartheta,\theta)=0$.
Then the inequality is trivially satisfied as $\alpha_{{\rm MH}}(\theta,\vartheta)=\alpha_{\ref{alg:1_hit_kernel}}(\theta,\vartheta)=0$.
Conversely, if $\pi(\theta)q(\theta,\vartheta)>0$ and $\pi(\vartheta)q(\vartheta,\theta)>0$
and additionally $\vartheta\in B_{r,\theta}$, then under Condition~\ref{cond:c2},
\begin{align*}
\frac{(1+M_{r})c(\vartheta,\theta)h(\vartheta)}{\alpha_{{\rm MH}}(\theta,\vartheta)} & =(1+M_{r})\left\{ c(\theta,\vartheta)h(\theta)\vee c(\vartheta,\theta)h(\vartheta)\right\} \\
 & \geq\left\{ c(\theta,\vartheta)h(\theta)\vee c(\vartheta,\theta)h(\vartheta)\right\} +\left\{ c(\theta,\vartheta)h(\vartheta)\vee c(\vartheta,\theta)h(\theta)\right\} \\
 & \geq\{(c(\theta,\vartheta)h(\theta)+c(\theta,\vartheta)h(\vartheta))\vee(c(\vartheta,\theta)h(\vartheta)+c(\vartheta,\theta)h(\theta))\}\\
 & =\frac{c(\vartheta,\theta)h(\vartheta)}{\{\frac{c(\vartheta,\theta)h(\vartheta)}{c(\theta,\vartheta)h(\theta)+c(\theta,\vartheta)h(\vartheta)}\wedge\frac{c(\vartheta,\theta)h(\vartheta)}{c(\vartheta,\theta)h(\vartheta)+c(\vartheta,\theta)h(\theta)}\}}=\frac{c(\vartheta,\theta)h(\vartheta)}{\frac{h(\vartheta)}{h(\vartheta)+h(\theta)}\{\frac{c(\vartheta,\theta)}{c(\theta,\vartheta)}\wedge1\}}\\
 & \geq\frac{c(\vartheta,\theta)h(\vartheta)}{\alpha_{\ref{alg:1_hit_kernel}}(\theta,\vartheta)},
\end{align*}
i.e., $\alpha_{{\rm MH}}(\theta,\vartheta)\leq(1+M_{r})\alpha_{\ref{alg:1_hit_kernel}}(\theta,\vartheta)$.
The first inequality is obtained by recalling that under Condition~\ref{cond:c2},
when $\pi(\theta)q(\theta,\vartheta)\wedge\pi(\vartheta)q(\vartheta,\theta)>0$
we have $M_{r}^{-1}\leq h(\vartheta)/h(\theta)\leq M_{r}$ or $M_{r}^{-1}\leq c(\vartheta,\theta)/c(\theta,\vartheta)\leq M_{r}$
and in either case $M_{r}\left\{ c(\theta,\vartheta)h(\theta)\vee c(\vartheta,\theta)h(\vartheta)\right\} \geq\left\{ c(\theta,\vartheta)h(\vartheta)\right\} \vee\left\{ c(\vartheta,\theta)h(\theta)\right\} $.

Hence, we have
\begin{align*}
P_{{\rm MH}}(\theta,A^{\complement}) & =\int_{A^{\complement}}\alpha_{{\rm MH}}(\theta,\vartheta)q(\theta,{\rm d}\vartheta)\leq q(\theta,B_{r,\theta}^{\complement})+\int_{A^{\complement}\cap B_{r,\theta}}(1+M_{r})\alpha_{\ref{alg:1_hit_kernel}}(\theta,\vartheta)q(\theta,{\rm d}\vartheta)\\
 & \leq\sup_{\theta}q(\theta,B_{r,\theta}^{\complement})+(1+M_{r})P_{\ref{alg:1_hit_kernel}}(\theta,A^{\complement}).
\end{align*}

\end{proof}

\begin{proof}[Proof of Theorem~\ref{thm:ge}]
Recall that geometric ergodicity is equivalent to $\sup|\sigma_{0}(P)|<1$.
From the spectral mapping theorem \citep{conway1990course} this is
equivalent to $\sup\sigma_{0}(P^{2})<1$, where $\sigma_{0}(P^{2})$
is the spectrum of $P^{2}$, the two-fold iterate of $P$. We denote
by $\kappa_{\ref{alg:1_hit_kernel}}^{(2)}$ and $\kappa_{{\rm MH}}^{(2)}$
the conductance of $P_{\ref{alg:1_hit_kernel}}^{2}$ and $P_{{\rm MH}}^{2}$
respectively. Since $q\in\mathcal{Q}$ we let $R=r_{q}(\kappa_{{\rm MH}}^{(2)}/4)$
and $M_{R}$ be as in Condition~\ref{cond:c2}. By Lemmas~\ref{lem:diagonal_dominance}
and~\ref{lem:kernel_bound}, we have for any measurable $A\subseteq\Theta$
\begin{align*}
P_{{\rm MH}}(\theta,A) & =P_{{\rm MH}}(\theta,A\setminus\{\theta\})+I(\theta\in A)P_{{\rm MH}}(\theta,\{\theta\})\\
 & \leq\kappa_{{\rm MH}}^{(2)}/4+(1+M_{R})P_{\ref{alg:1_hit_kernel}}(\theta,A\setminus\{\theta\})+P_{\ref{alg:1_hit_kernel}}(\theta,\{\theta\})\\
 & \leq\kappa_{{\rm MH}}^{(2)}/4+(1+M_{R})P_{\ref{alg:1_hit_kernel}}(\theta,A).
\end{align*}
We can also upper bound, for any $\theta\in\Theta$, the Radon--Nikodym
derivative of $P_{{\rm MH}}(\theta,\cdot)$ with respect to $P_{\ref{alg:1_hit_kernel}}(\theta,\cdot)$
for any $\vartheta\in B_{R,\theta}$ as
\begin{align*}
\frac{{\rm d}P_{{\rm MH}}(\theta,\cdot)}{{\rm d}P_{\ref{alg:1_hit_kernel}}(\theta,\cdot)}(\vartheta) & =I(\vartheta\in B_{R,\theta}\setminus\{\theta\})\frac{{\rm d}q(\theta,\cdot)}{{\rm d}q(\theta,\cdot)}(\vartheta)\frac{\alpha_{{\rm MH}}(\theta,\vartheta)}{\alpha_{\ref{alg:1_hit_kernel}}(\theta,\vartheta)}+I(\vartheta=\theta)\frac{P_{{\rm MH}}(\theta,\{\theta\})}{P_{\ref{alg:1_hit_kernel}}(\theta,\{\theta\})}\\
 & \leq I(\vartheta\in B_{R,\theta}\setminus\{\theta\})(1+M_{R})+I(\vartheta=\theta)\leq1+M_{R},
\end{align*}
where we have used~\eqref{eq:alpha_bound} and Lemma~\ref{lem:diagonal_dominance}
in the first inequality.

Let $A$ be a measurable set with $\pi(A)>0$. We have
\begin{align*}
\kappa_{{\rm MH}}^{(2)}(A) & =\int_{A}\left\{ \int_{\Theta}P_{{\rm MH}}(\vartheta,A^{\complement})P_{{\rm MH}}(\theta,{\rm d}\vartheta)\right\} \pi_{A}(\mathrm{d}\theta)\\
 & =\int_{A}\left\{ \int_{B_{R,\theta}^{\complement}}P_{{\rm MH}}(\vartheta,A^{\complement})P_{{\rm MH}}(\theta,{\rm d}\vartheta)+\int_{B_{R,\theta}}P_{{\rm MH}}(\vartheta,A^{\complement})P_{{\rm MH}}(\theta,{\rm d}\vartheta)\right\} \pi_{A}(\mathrm{d}\theta)\\
 & \leq\int_{A}\left\{ q(\theta,B_{R,\theta}^{\complement})+\int_{B_{R,\theta}}P_{{\rm MH}}(\vartheta,A^{\complement})P_{{\rm MH}}(\theta,{\rm d}\vartheta)\right\} \pi_{A}(\mathrm{d}\theta)\\
 & \leq\kappa_{{\rm MH}}^{(2)}/4+\int_{A}\int_{B_{R,\theta}}P_{{\rm MH}}(\vartheta,A^{\complement})P_{{\rm MH}}(\theta,{\rm d}\vartheta)\pi_{A}(\mathrm{d}\theta)\\
 & \leq\kappa_{{\rm MH}}^{(2)}/4+\int_{A}\int_{B_{R,\theta}}\left\{ \kappa_{{\rm MH}}^{(2)}/4+(1+M_{R})P_{\ref{alg:1_hit_kernel}}(\vartheta,A^{\complement})\right\} P_{{\rm MH}}(\theta,{\rm d}\vartheta)\pi_{A}(\mathrm{d}\theta)\\
 & \leq\kappa_{{\rm MH}}^{(2)}/2+(1+M_{R})\int_{A}\int_{B_{R,\theta}}P_{\ref{alg:1_hit_kernel}}(\vartheta,A^{\complement})P_{{\rm MH}}(\theta,{\rm d}\vartheta)\pi_{A}(\mathrm{d}\theta)\\
 & =\kappa_{{\rm MH}}^{(2)}/2+(1+M_{R})\int_{A}\int_{B_{R,\theta}}P_{\ref{alg:1_hit_kernel}}(\vartheta,A^{\complement})\frac{{\rm d}P_{{\rm MH}}(\theta,\cdot)}{{\rm d}P_{\ref{alg:1_hit_kernel}}(\theta,\cdot)}(\vartheta)P_{\ref{alg:1_hit_kernel}}(\theta,{\rm d}\vartheta)\pi_{A}(\mathrm{d}\theta)\\
 & \leq\kappa_{{\rm MH}}^{(2)}/2+(1+M_{R})^{2}\int_{A}\int_{B_{R,\theta}}P_{\ref{alg:1_hit_kernel}}(\vartheta,A^{\complement})P_{\ref{alg:1_hit_kernel}}(\theta,{\rm d}\vartheta)\pi_{A}(\mathrm{d}\theta)\\
 & \leq\kappa_{{\rm MH}}^{(2)}/2+(1+M_{R})^{2}\int_{A}\int_{\Theta}P_{\ref{alg:1_hit_kernel}}(\vartheta,A^{\complement})P_{\ref{alg:1_hit_kernel}}(\theta,{\rm d}\vartheta)\pi_{A}(\mathrm{d}\theta)\\
 & =\kappa_{{\rm MH}}^{(2)}/2+(1+M_{R})^{2}\kappa_{\ref{alg:1_hit_kernel}}^{(2)}(A).
\end{align*}
Since $A$ is arbitrary, we conclude that $\kappa_{{\rm MH}}^{(2)}\leq2(1+M_{R})^{2}\kappa_{\ref{alg:1_hit_kernel}}^{(2)}$
so $\kappa_{{\rm MH}}^{(2)}>0\Rightarrow\kappa_{\ref{alg:1_hit_kernel}}^{(2)}>0$.\end{proof}
\begin{lem}
\label{lem:mix_vb_geo}Let $K_{1}$ be a reversible Markov kernel
with unique invariant distribution $\pi$ and let $K_{2}$ be reversible
with invariant distribution $\pi$. Let $\tilde{K}=aK_{1}+(1-a)K_{2}$
be a mixture of $K_{1}$ and $K_{2}$ for $a\in(0,1]$. Then $K_{1}\in\mathcal{V}\Rightarrow\tilde{K}\in\mathcal{V}$
and $K_{1}\in\mathcal{G}\Rightarrow\tilde{K}\in\mathcal{G}$.\end{lem}
\begin{proof}
For the first part, assume $K_{1}\in\mathcal{V}$. Then since $K_{1}$
is reversible with unique invariant distribution $\pi$, its conductance
$\kappa_{1}$ satisfies $\kappa_{1}>0$. Since $K_{2}$ is also reversible,
the mixture $\tilde{K}$ is reversible with unique invariant distribution
$\pi$ and its conductance is
\begin{align*}
\tilde{\kappa} & =\inf_{A:0<\pi(A)\leq1/2}\int_{A}\tilde{K}(\theta,A^{\complement})\pi_{A}(\mathrm{d}\theta)\\
 & \geq\inf_{A:0<\pi(A)\leq1/2}\int_{A}aK_{1}(\theta,A^{\complement})\pi_{A}(\mathrm{d}\theta)\\
 & =a\kappa_{1}>0.
\end{align*}
Hence $\tilde{K}\in\mathcal{V}$.

Similarly, for the second part, assume $K_{1}\in\mathcal{G}$. Then
the conductance of $K_{1}^{2}$, $\kappa_{1}^{(2)}$, satisfies $\kappa_{1}^{(2)}>0$
by the spectral mapping theorem \citep{conway1990course}. Let $\tilde{\kappa}^{(2)}$
be the conductance of $\tilde{K}^{2}$, and it suffices to show that
$\tilde{\kappa}^{(2)}>0$. We have 
\begin{align*}
\tilde{\kappa}^{(2)} & =\inf_{A:0<\pi(A)\leq1/2}\int_{A}\tilde{K}^{2}(\theta,A^{\complement})\pi_{A}(\mathrm{d}\theta)\\
 & \geq\inf_{A:0<\pi(A)\leq1/2}\int_{A}a^{2}K_{1}^{2}(\theta,A^{\complement})\pi_{A}(\mathrm{d}\theta)\\
 & =a^{2}\kappa_{1}^{(2)}>0.
\end{align*}
Hence $\tilde{K}\in\mathcal{G}$.
\end{proof}

\begin{proof}[Proof of Theorem~\ref{thm:mix_geo_nongeo}]
The result is immediate upon defining $L_{1}=K_{1}$, $L_{2}=\left(1-a_{1}\right)^{-1}\sum_{i=2}^{\infty}a_{i}K_{i}$
and $\tilde{L}=a_{1}L_{1}+(1-a_{1})L_{2}$ and applying Lemma~\ref{lem:mix_vb_geo}.
\end{proof}

\begin{proof}[Proof of Proposition~\ref{prop:n}]
If the current state of the Markov chain is $\theta$, the expected
value of $N$ is
\[
n(\theta)=\int_{\Theta}\frac{\left[1\wedge\left\{ c(\vartheta,\theta)/c(\theta,\vartheta)\right\} \right]}{h(\theta)+h(\vartheta)-h(\theta)h(\vartheta)}q(\theta,{\rm d}\vartheta),
\]
since upon drawing $\vartheta\sim q(\theta,\cdot)$, $N=0$ with probability
$1-\left\{ 1\wedge c(\vartheta,\theta)\right\} $ and with probability
$\left\{ 1\wedge c(\vartheta,\theta)\right\} $ it is the minimum
of two geometric random variables with success probabilities $h(\theta)$
and $h(\vartheta)$, i.e.\ it is a geometric random variable with
success probability $h(\theta)+h(\vartheta)-h(\theta)h(\vartheta)$.

Since $P_{\ref{alg:1_hit_kernel}}$ is $\pi$-invariant and irreducible,
the strong law of large numbers for Markov chains implies
\[
n=\int_{\Theta}n(\theta)\pi(\mathrm{d}\theta)=H^{-1}\int_{\Theta^{2}}\frac{\left[1\wedge\left\{ c(\vartheta,\theta)/c(\theta,\vartheta)\right\} \right]h(\theta)}{h(\theta)+h(\vartheta)-h(\theta)h(\vartheta)}q(\theta,{\rm d}\vartheta)p({\rm d}\theta)\leq H^{-1}<\infty,
\]
where we have used $\int_{\Theta}p(\theta)\mathrm{d}\theta=1$ in
the first inequality.
\end{proof}

\section{Negative results in other settings\label{sec:more_neg_results}}

This appendix extends Theorem~\ref{thm:nvb} to a number of related
approximate Bayesian computation settings. These results indicate
that the conclusions of Theorem~\ref{thm:nvb} about lack of geometric
ergodicity and variance bounding property hold much more universally.
We first consider the case where one utilizes a proposal that falls
just outside the definition of $\mathcal{Q}$. Of particular interest
could be those proposals that are biased towards the centre of $\Theta$
but are not global. To this end, we can define
\[
\mathcal{Q}_{0}=\{q\,:\,\text{for all }\delta>0\:\text{and }r>0,\:\text{there exists }R>0\:\text{such that for all }\theta\in B_{R,0}^{\complement},\: q(\theta,B_{r,0})<\delta\},
\]
which includes, for example, the autoregressive proposal $q(\theta,\vartheta)=\mathcal{N}(\vartheta;\rho\theta,\sigma^{2})$
for some $\rho\in(0,1)$. The following result indicates that such
proposals are similarly associated with lack of variance bounding
for $P_{\ref{alg:pm_kernel_2},N}$.
\begin{prop}
Let $q\in\mathcal{Q}_{0}$ and assume that \textup{for all $r>0$,
$\pi(B_{r,0}^{\complement})>0$, and for all $\delta>0$ there exists
$v>0$ such that $\sup_{\theta\in B_{v,0}^{\complement}}h(\theta)<\delta$.
Then }$P_{\ref{alg:pm_kernel_2},N}\notin\mathcal{V}$ for any $N\in\mathbb{N}$.\end{prop}
\begin{proof}
By Theorem~\ref{thm:gen_nvb}, it suffices to show that $\pi-{\rm ess}\sup_{\theta}P_{\ref{alg:pm_kernel_2},N}(\theta,\{\theta\})=1$.
Let $q\in\mathcal{Q}_{0}$, $\tau>0$ and take $r=\inf\left\{ r:\sup_{\theta\in B_{r,0}^{\complement}}h(\theta)<1-\left(1-\tau/2\right)^{1/N}\right\} $
and $R=\inf\left\{ R:\sup_{\theta\in B_{R,0}^{\complement}}q(\theta,B_{r,0})<\tau/2\right\} $,
which both exist by assumption. Furthermore, $\pi(B_{R,0}^{\complement})>0$.
Let $A_{\tau}=B_{R,0}^{\complement}$. We have for all $\theta\in A$,
\begin{align*}
P_{\ref{alg:pm_kernel_2},N}(\theta,\{\theta\}^{\complement}) & =\int_{\{\theta\}^{\complement}}\int_{\mathsf{Y}^{N}}\int_{\mathsf{Y}^{N-1}}\left[1\wedge\frac{c(\vartheta,\theta)\sum_{j=1}^{N}w(z_{j})}{c(\theta,\vartheta)\left\{ 1+\sum_{j=1}^{N-1}w(x_{j})\right\} }\right]f_{\theta}^{\otimes N-1}({\rm d}x_{1:N-1})f_{\vartheta}^{\otimes N}({\rm d}z_{1:N})q(\theta,{\rm d}\vartheta)\\
 & \leq q\left(\theta,B_{r,\theta}\right)+\int_{B_{r,\theta}^{\complement}}\int_{\mathsf{Y}^{N}}I\left\{ \sum_{i=1}^{N}w(z_{i})\geq1\right\} f_{\vartheta}^{\otimes N}({\rm d}z_{1:N})q(\theta,{\rm d}\vartheta)\\
 & \leq\frac{\tau}{2}+\int_{B_{r,\theta}^{\complement}}\left[1-\left\{ 1-\sup_{\vartheta\in B_{r,\theta}^{\complement}}h(\vartheta)\right\} ^{N}\right]q(\theta,{\rm d}\vartheta)\leq\tau,
\end{align*}
so $\pi-{\rm ess}\sup_{\theta}P_{\ref{alg:pm_kernel_2},N}(\theta,\{\theta\})=1$.
\end{proof}
We now consider a more general specification of~\eqref{eq:abc_approx_likelihood},
and consider the artificial likelihood
\[
\tilde{f}_{\theta}^{\epsilon}(y)=\int_{\mathsf{Y}}K_{\epsilon}(x,y)f_{\theta}(x)\mathrm{d}x,
\]
where $K_{\epsilon}$ is a Markov kernel. Note that with $K_{\epsilon}(x,y)=V(\epsilon)^{-1}I(y\in B_{\epsilon,x})$
we recover~\eqref{eq:abc_approx_likelihood}. We further consider
a target augmented with $\epsilon$, i.e. 
\[
\pi(\theta,\epsilon)\propto p(\theta,\epsilon)\tilde{f}_{\theta}^{\epsilon}(y),
\]
as such targets have been suggested in an attempt to improve performance
of associated Markov kernels \citep[see, e.g., ][]{bortot2007inference,sisson2010likelihood}.
Note that one could allow $p(\epsilon)$ to be concentrated at a single
point to define a target with a fixed value of $\epsilon$.

We consider the Markov kernel
\begin{align*}
P_{4}(\theta,\epsilon,x;{\rm d}\vartheta,{\rm d}\varepsilon,{\rm d}z) & =q_{4}(\theta,\epsilon;{\rm d}\vartheta,{\rm d}\varepsilon)f_{\vartheta}({\rm d}z)\alpha_{4}(\theta,\epsilon,x;\vartheta,\varepsilon,z)\\
 & \quad+\left\{ 1-\int_{\Theta}q_{4}(\theta,\epsilon;{\rm d}\theta',{\rm d}\epsilon')f_{\vartheta}({\rm d}x')\alpha_{4}(\theta,\theta')\right\} \delta_{(\theta,\epsilon,x)}({\rm d}\vartheta,{\rm d}\varepsilon,{\rm d}z),
\end{align*}
where
\[
\alpha_{4}(\theta,\epsilon,x;\vartheta,\varepsilon,z)=1\wedge\frac{p(\vartheta,\varepsilon)q((\vartheta,\varepsilon),(\theta,\epsilon))K_{\varepsilon}(x,y)}{p(\theta,\epsilon)q((\theta,\epsilon),(\vartheta,\varepsilon))K_{\epsilon}(z,y)},
\]
which can be seen as an analogue of $P_{\ref{alg:pm_kernel_1},1}$.
Extensions to $N>1$ are possible using the methodology of \citet{Beaumont2003,Andrieu2009},
and the following result also holds for $N>1$. Furthermore, if $P_{4}$
is irreducible and aperiodic it admits $\tilde{\pi}(\theta,\epsilon,x)\propto p(\theta,\epsilon)f_{\theta}(x)K_{\epsilon}(x,y)$
as its unique invariant distribution which after integrating out $x$
results in the $(\theta,\epsilon)$-marginal $\tilde{\pi}(\theta,\epsilon)\propto p(\theta,\epsilon)\tilde{f}_{\theta}^{\epsilon}(y)$.
The following result indicates that $P_{4}$ is not variance bounding
under some mild general conditions.

We first introduce mild general assumptions for Propositions~\ref{prop:weakest_neg}
and~\ref{prop:light_tails}.
\begin{enumerate}
\item [(G1)]The prior can be factorized as $p(\theta,\epsilon)=p_{\theta}(\theta)p_{\epsilon}(\epsilon)$,
\item [(G2)]The proposal can be factorized as $q_{4}(\theta,\epsilon;\vartheta,\varepsilon)=q(\theta,\vartheta)g(\theta,\vartheta,\epsilon;\varepsilon)$
with $q\in\mathcal{Q}$,
\item [(G3)]For every $\varepsilon>\varepsilon_{0}>0$, $\sup{}_{\theta,\vartheta,\epsilon}g(\theta,\vartheta,\epsilon;\varepsilon)<M_{1}(\varepsilon_{0})$,
\item [(G4)]The proposal satisfies $\sup_{\theta,\vartheta}q(\theta,\vartheta)<M_{2}<\infty$,
\item [(G5)]For every $\epsilon_{0}>0$, there exists $k=k(\epsilon_{0})>0$
such that for every $\epsilon<\epsilon_{0}$ we have 
\[
\int_{\mathsf{Y}}I\left\{ K_{\epsilon}(x,y)>k\right\} f_{\theta}({\rm d}x)>0,
\]

\item [(G6)]For every $(\epsilon,x)$ such that $K_{\epsilon}(x,y)>0$,
the conditional distribution of $\theta$ under $\tilde{\pi}$ is
not compactly supported, i.e.,\ $\int_{B_{R,0}^{\complement}}\tilde{\pi}(\theta,\epsilon,x){\rm d}\theta>0$
for all $R>0$.\end{enumerate}
\begin{prop}
\label{prop:weakest_neg}Assume in addition to (G1)--(G6), the following
additional conditions:\end{prop}
\begin{enumerate}
\item [{\rm(G7)}]The artificial likelihood satisfies $\lim{}_{R\to\infty}\sup_{\theta\in B_{R,0}^{\complement}}\tilde{f}_{\theta}(y)=0$,
where $\tilde{f}_{\theta}(y)=\int_{0}^{\infty}p(\epsilon)\tilde{f}_{\theta}^{\epsilon}(y){\rm d}\epsilon$,
\item [{\rm(G8)}]The prior $p_{\theta}(\theta)$ has at most exponentially
decaying tails, i.e.,\ for every $r>0$ there exist $M_{3}(r)>0$
and $M_{4}(r)>0$ such that 
\[
\sup_{\theta\in B_{M_{3}(r),0}^{\complement},\vartheta\in B_{r,\theta}}\frac{p_{\theta}(\vartheta)}{p_{\theta}(\theta)}<M_{4}(r)<\infty.
\]

\end{enumerate}
Then $P_{4}\notin\mathcal{V}$, and consequently also $P_{4}\notin\mathcal{G}$.
\begin{proof}
By Theorem~\ref{thm:gen_nvb}, it suffices to show that $\tilde{\pi}-{\rm ess}\sup_{\theta,\epsilon,x}P_{4}(\theta,\epsilon,x;\{(\theta,\epsilon,x)\})=1$.
First choose fixed $0<\epsilon_{l}<\epsilon_{r}<\infty$ and $\delta_{1}>0$
so that $\int_{(\epsilon_{l},\epsilon_{r})}p_{\epsilon}(\epsilon)I\{p_{\epsilon}(\epsilon)>\delta_{1}\}{\rm d}\epsilon>0$,
and then by assumption (G5) choose $\delta_{2}$ so that $\int_{\mathsf{Y}}I\{K_{\epsilon}(x,y)>\delta_{2}\}f_{\theta}({\rm d}x)>0$
for every $\epsilon<\epsilon_{r}$. Define $\Re_{\epsilon}=\left\{ \epsilon\in(\epsilon_{l},\epsilon_{r})\,:\, p_{\epsilon}(\epsilon)>\delta_{1}\right\} $
and $\Re_{z}=\left\{ z\in\mathsf{Y}\,:\, K_{\epsilon}(x,y)>\delta_{2}\right\} $.
The sets $\Re_{\epsilon}$ and $\Re_{z}$ will be fixed throughout
the proof. Now for every $R>0$ define the set $A(R)$ as 
\[
\Theta\times\mathbb{R}_{+}\times\mathsf{Y}\;\supset\; A(R)\;=\; B_{R,0}^{\complement}\times\Re_{\epsilon}\times\Re_{z}.
\]
The set $A(R)$ has positive $\tilde{\pi}$ mass for every $R$ by
(G5) and (G6). We will investigate the behaviour of $P_{4}$ in $A(R)$
as $R\to\infty$. Let $\tau>0$ and take $r=\inf\left\{ r\,:\, q\left(\theta,B_{r,\theta}^{\complement}\right)<\tau/2\right\} $.
For every $(\theta,\epsilon,x)\in A(R)$ we can compute
\begin{align*}
 & P_{4}((\theta,\epsilon,x),\{(\theta,\epsilon,x)\}^{c})\\
 & =\int_{\Theta\times\mathbb{R}_{+}\times\mathsf{Y}}\left\{ 1\wedge\frac{p_{\theta}(\vartheta)p_{\epsilon}(\varepsilon)q(\vartheta,\theta)g(\vartheta,\theta,\varepsilon;\epsilon)K_{\varepsilon}(z,y)}{p_{\theta}(\theta)p_{\epsilon}(\epsilon)q(\theta,\vartheta)g(\theta,\vartheta,\epsilon;\varepsilon)K_{\epsilon}(x,y)}\right\} q(\theta,\vartheta)g(\theta,\vartheta,\epsilon;\varepsilon)f_{\vartheta}(z){\rm d}z{\rm d}\varepsilon{\rm d}\vartheta\\
 & \leq\frac{\tau}{2}+\int_{B_{r,\theta}\times\mathbb{R}_{+}}\left\{ 1\wedge\frac{p_{\theta}(\vartheta)p_{\epsilon}(\varepsilon)q(\vartheta,\theta)g(\vartheta,\theta,\varepsilon;\epsilon)K_{\varepsilon}(z,y)}{p_{\theta}(\theta)p_{\epsilon}(\epsilon)q(\theta,\vartheta)g(\theta,\vartheta,\epsilon;\varepsilon)K_{\epsilon}(x,y)}\right\} q(\theta,\vartheta)g(\theta,\vartheta,\epsilon;\varepsilon)f_{\vartheta}(z){\rm d}z{\rm d}\varepsilon{\rm d}\vartheta\\
 & \leq\frac{\tau}{2}+\int_{B_{r,\theta}\times\mathbb{R}_{+}}\frac{p_{\theta}(\vartheta)p_{\epsilon}(\varepsilon)q(\vartheta,\theta)g(\vartheta,\theta,\varepsilon;\epsilon)K_{\varepsilon}(z,y)}{p_{\theta}(\theta)p_{\epsilon}(\epsilon)q(\theta,\vartheta)g(\theta,\vartheta,\epsilon;\varepsilon)K_{\epsilon}(x,y)}q(\theta,\vartheta)g(\theta,\vartheta,\epsilon;\varepsilon)f_{\vartheta}(z){\rm d}z{\rm d}\varepsilon{\rm d}\vartheta\\
 & \leq\frac{\tau}{2}+\int_{B_{r,\theta}\times\mathbb{R}_{+}}\frac{p_{\theta}(\vartheta)p_{\epsilon}(\varepsilon)q(\vartheta,\theta)g(\vartheta,\theta,\varepsilon;\epsilon)}{p_{\theta}(\theta)p_{\epsilon}(\epsilon)K_{\epsilon}(x,y)}\tilde{f}_{\vartheta}^{\varepsilon}(y){\rm d}\varepsilon{\rm d}\vartheta\\
 & \leq\frac{\tau}{2}+M_{1}(\epsilon_{l})\int_{B_{r,\theta}\times\mathbb{R}_{+}}\frac{p_{\theta}(\vartheta)p_{\epsilon}(\varepsilon)q(\vartheta,\theta)}{p_{\theta}(\theta)p_{\epsilon}(\epsilon)K_{\epsilon}(x,y)}\tilde{f}_{\vartheta}^{\varepsilon}(y){\rm d}\varepsilon{\rm d}\vartheta\\
 & \leq\frac{\tau}{2}+\frac{M_{1}(\epsilon_{l})M_{2}}{\delta_{2}\delta_{1}}\int_{B_{r,\theta}\times\mathbb{R}_{+}}\frac{p_{\theta}(\vartheta)}{p_{\theta}(\theta)}p_{\epsilon}(\varepsilon)\tilde{f}_{\vartheta}^{\varepsilon}(y){\rm d}\varepsilon{\rm d}\vartheta
\end{align*}
Then by assumption (G8) for $R>M_{3}(r)$ we have
\begin{align*}
P_{4}(\theta,\epsilon,x;\{(\theta,\epsilon,x)\}^{\complement}) & \leq\frac{\tau}{2}+\frac{M_{1}(\epsilon_{l})M_{2}M_{4}(r)}{\delta_{2}\delta_{1}}\int_{B_{r,\theta}\times\mathbb{R}_{+}}p_{\epsilon}(\varepsilon)\tilde{f}_{\vartheta}^{\varepsilon}(y){\rm d}\varepsilon{\rm d}\vartheta\\
 & =\frac{\tau}{2}+\frac{M_{1}(\epsilon_{l})M_{2}M_{4}(r)}{\delta_{2}\delta_{1}}\int_{B_{r,\theta}}\tilde{f}_{\vartheta}(y){\rm d}\vartheta\\
 & \leq\frac{\tau}{2}+\frac{M_{1}(\epsilon_{l})M_{2}M_{4}(r)V(r)}{\delta_{2}\delta_{1}}\sup_{\vartheta\in B_{r,\theta}}\tilde{f}_{\vartheta}(y).
\end{align*}
Now by (G7), $\sup_{\theta\in B_{R,0}^{\complement},\vartheta\in B_{r,\theta}}\tilde{f}_{\vartheta}(y)\to0$
as $R\to\infty$. Consequently, for fixed $\tau$ we obtain $\tilde{\pi}-{\rm ess}\sup_{\theta,\epsilon,x}P_{4}(\theta,\epsilon,x;\{(\theta,\epsilon,x)\})\geq1-\tau$
by taking an increasing sequence $R_{i}$. Since $\tau$ can be taken
arbitrarily small, this implies that $\tilde{\pi}-{\rm ess}\sup_{\theta,\epsilon,x}P_{4}(\theta,\epsilon,x;\{(\theta,\epsilon,x)\})=1$
and we conclude.\end{proof}
\begin{rem}
Of the conditions under which Proposition \ref{prop:weakest_neg}
holds, (G8) is perhaps the strongest. We relax this assumption in
the statement of Proposition \ref{prop:light_tails}, replacing it
with assumptions on $g$.\end{rem}
\begin{prop}
\label{prop:light_tails}Assume in addition to (G1)--(G6), the following
additional conditions:\end{prop}
\begin{enumerate}
\item [{\rm(G9)}]For any fixed $\bar{\varepsilon}>0$ and $r>0$,
\[
\lim_{R\rightarrow\infty}\sup_{\theta\in B_{R,0}^{\complement},\vartheta\in B_{r,\theta},\varepsilon\in[0,\bar{\varepsilon}]}\frac{p_{\theta}(\vartheta)}{p_{\theta}(\theta)}p_{\epsilon}(\varepsilon)\tilde{f}_{\vartheta}^{\varepsilon}(y)=0,
\]

\item [{\rm(G10)}]The proposal for $\varepsilon$ is independent of $(\theta,\vartheta)$,
i.e., $g(\theta,\vartheta,\epsilon;\varepsilon)=g(\epsilon,\varepsilon)$,
\item [{\rm(G11)}]There exist $0<\epsilon_{L}<\epsilon_{R}<\infty$ with
$\int_{\epsilon_{L}}^{\epsilon_{R}}p_{\epsilon}(\epsilon){\rm d}\epsilon>0$
such that the family of distributions $\{g(\epsilon,\cdot)\}_{\epsilon\in[\epsilon_{L},\epsilon_{R}]}$
is tight. In particular, if $G_{\epsilon}$ is the cumulative distribution
function associated with $g(\epsilon,\cdot)$ then there exists a
function $\phi$ such that for all $u\in(0,1)$, $\sup_{\epsilon\in[\epsilon_{L},\epsilon_{R}]}G_{\epsilon}^{-1}(u)\leq\phi(u)<\infty$.
\end{enumerate}
Then $P_{4}\notin\mathcal{V}$, and consequently also $P_{4}\notin\mathcal{G}$.
\begin{proof}
By Theorem~\ref{thm:gen_nvb}, it suffices to show that $\tilde{\pi}-{\rm ess}\sup_{\theta,\epsilon,x}P_{4}(\theta,\epsilon,x;\{(\theta,\epsilon,x)\})=1$.
From (G11) choose fixed $\epsilon_{L}\leq\epsilon_{l}<\epsilon_{r}\leq\epsilon_{R}$
and $\delta_{1}>0$ so that $\int_{(\epsilon_{l},\epsilon_{r})}p_{\epsilon}(\epsilon)\mathbb{I}(p_{\epsilon}(\epsilon)>\delta_{1}){\rm d}\epsilon>0$,
and then by (G5) choose $\delta_{2}$ so that $\int_{\mathsf{Y}}I\{K_{\epsilon}(x,y)>\delta_{2}\}f_{\theta}({\rm d}x)>0$
for every $\epsilon<\epsilon_{r}$. Define $\Re_{\epsilon}=\left\{ \epsilon\in(\epsilon_{l},\epsilon_{r})\,:\, p_{\epsilon}(\epsilon)>\delta_{1}\right\} $
and $\Re_{z}=\left\{ z\in\mathsf{Y}\,:\, K_{\epsilon}(x,y)>\delta_{2}\right\} $.
The sets $\Re_{\epsilon}$ and $\Re_{z}$ will be fixed throughout
the proof. Now for every $R>0$ define the set $A(R)$ as 
\[
\Theta\times\mathbb{R}_{+}\times\mathsf{Y}\;\supset\; A(R)\;=\; B_{R,0}^{\complement}\times\Re_{\epsilon}\times\Re_{z}.
\]
The set $A(R)$ has positive $\tilde{\pi}$ mass for every $R$ by
(G5) and (G6). We will investigate the behaviour of $P_{4}$ in $A(R)$
as $R\to\infty$. Let $\tau>0$ and take $r=\inf\left\{ r\,:\, q(\theta,B_{r,\theta}^{\complement})<\tau/2\right\} $.
We take $\bar{\varepsilon}(\tau)=\phi(1-\tau/4)$. For every $(\theta,\epsilon,x)\in A(R)$
we can compute
\begin{align*}
 & P_{4}((\theta,\epsilon,x),\{(\theta,\epsilon,x)\}^{c})\\
 & =\int_{\Theta\times\mathbb{R}_{+}\times\mathsf{Y}}\left\{ 1\wedge\frac{p_{\theta}(\vartheta)p_{\epsilon}(\varepsilon)q(\vartheta,\theta)g(\varepsilon,\epsilon)K_{\varepsilon}(z,y)}{p_{\theta}(\theta)p_{\epsilon}(\epsilon)q(\theta,\vartheta)g(\epsilon,\varepsilon)K_{\epsilon}(x,y)}\right\} q(\theta,\vartheta)g(\epsilon,\varepsilon)f_{\vartheta}(z){\rm d}z{\rm d}\varepsilon{\rm d}\vartheta\\
 & \leq\frac{\tau}{2}+\int_{B_{r,\theta}\times\mathbb{R}_{+}}\left\{ 1\wedge\frac{p_{\theta}(\vartheta)p_{\epsilon}(\varepsilon)q(\vartheta,\theta)g(\varepsilon,\epsilon)K_{\varepsilon}(z,y)}{p_{\theta}(\theta)p_{\epsilon}(\epsilon)q(\theta,\vartheta)g(\epsilon,\varepsilon)K_{\epsilon}(x,y)}\right\} q(\theta,\vartheta)g(\epsilon,\varepsilon)f_{\vartheta}(z){\rm d}z{\rm d}\varepsilon{\rm d}\vartheta\\
 & \leq\frac{\tau}{2}+\frac{\tau}{4}+\int_{B_{r,\theta}\times[0,\bar{\varepsilon}(\tau)]}\left\{ 1\wedge\frac{p_{\theta}(\vartheta)p_{\epsilon}(\varepsilon)q(\vartheta,\theta)g(\varepsilon,\epsilon)K_{\varepsilon}(z,y)}{p_{\theta}(\theta)p_{\epsilon}(\epsilon)q(\theta,\vartheta)g(\epsilon,\varepsilon)K_{\epsilon}(x,y)}\right\} q(\theta,\vartheta)g(\epsilon,\varepsilon)f_{\vartheta}(z){\rm d}z{\rm d}\varepsilon{\rm d}\vartheta\\
 & \leq\frac{3\tau}{4}+\int_{B_{r,\theta}\times[0,\bar{\varepsilon}(\tau)]}\frac{p_{\theta}(\vartheta)p_{\epsilon}(\varepsilon)q(\vartheta,\theta)g(\varepsilon,\epsilon)K_{\varepsilon}(z,y)}{p_{\theta}(\theta)p_{\epsilon}(\epsilon)q(\theta,\vartheta)g(\epsilon,\varepsilon)K_{\epsilon}(x,y)}q(\theta,\vartheta)g(\epsilon,\varepsilon)f_{\vartheta}(z){\rm d}z{\rm d}\varepsilon{\rm d}\vartheta\\
 & \leq\frac{3\tau}{4}+\int_{B_{r,\theta}\times[0,\bar{\varepsilon}(\tau)]}\frac{p_{\theta}(\vartheta)p_{\epsilon}(\varepsilon)q(\vartheta,\theta)g(\varepsilon,\epsilon)}{p_{\theta}(\theta)p_{\epsilon}(\epsilon)K_{\epsilon}(x,y)}\tilde{f}_{\vartheta}^{\varepsilon}(y){\rm d}\varepsilon{\rm d}\vartheta\\
 & \leq\frac{3\tau}{4}+M_{1}(\epsilon_{l})M_{2}\int_{B_{r,\theta}\times[0,\bar{\varepsilon}(\tau)]}\frac{p_{\theta}(\vartheta)p_{\epsilon}(\varepsilon)}{p_{\theta}(\theta)p_{\epsilon}(\epsilon)K_{\epsilon}(x,y)}\tilde{f}_{\vartheta}^{\varepsilon}(y){\rm d}\varepsilon{\rm d}\vartheta\\
 & \leq\frac{3\tau}{4}+\frac{M_{1}(\epsilon_{l})M_{2}}{\delta_{2}\delta_{1}}\int_{B_{r,\theta}\times[0,\bar{\varepsilon}(\tau)]}\frac{p_{\theta}(\vartheta)}{p_{\theta}(\theta)}p_{\epsilon}(\varepsilon)\tilde{f}_{\vartheta}^{\varepsilon}(y){\rm d}\varepsilon{\rm d}\vartheta.\\
 & \leq\frac{3\tau}{4}+\frac{M_{1}(\epsilon_{l})M_{2}V(r)\bar{\varepsilon}(\tau)}{\delta_{2}\delta_{1}}\sup_{\vartheta\in B_{r,\theta},\varepsilon\in[0,\bar{\varepsilon}(\tau)]}\frac{p_{\theta}(\vartheta)}{p_{\theta}(\theta)}p_{\epsilon}(\varepsilon)\tilde{f}_{\vartheta}^{\varepsilon}(y).
\end{align*}
Now by (G9),
\[
\sup_{\theta\in B_{R,0}^{\complement},\vartheta\in B_{r,\theta},\varepsilon\in[0,\bar{\varepsilon}(\tau)]}\frac{p_{\theta}(\vartheta)}{p_{\theta}(\theta)}p_{\epsilon}(\varepsilon)\tilde{f}_{\vartheta}^{\varepsilon}(y)\rightarrow0
\]
as $R\to\infty$ for any $r>0$. Consequently, for fixed $\tau$ we
obtain $\tilde{\pi}-{\rm ess}\sup_{\theta,\epsilon,x}P_{4}(\theta,\epsilon,x;\{(\theta,\epsilon,x)\})\geq1-\tau$
by taking an increasing sequence $R_{i}$. Since $\tau$ can be taken
arbitrarily small, this implies that $\tilde{\pi}-{\rm ess}\sup_{\theta,\epsilon,x}P_{4}(\theta,\epsilon,x;\{(\theta,\epsilon,x)\})=1$
and we conclude.
\end{proof}
We provide two examples to show how Propositions~\ref{prop:weakest_neg}
and~\ref{prop:light_tails} can be applied.
\begin{example}
If $f_{\theta}(\cdot)=\mathcal{N}(\cdot;\theta,\sigma^{2})$, $K_{\epsilon}(x,y)=\mathcal{N}(y;x,\epsilon)$
and $p(\theta,\epsilon)=\lambda_{1}\lambda_{2}/2\exp(-\lambda_{1}|\theta|-\lambda_{2}\epsilon)$,
the conditions of Proposition~\ref{prop:weakest_neg} are met for
any $(\sigma^{2},\lambda_{1},\lambda_{2})\in(0,\infty)^{3}$ when
$q$ and $g$ satisfy (G2)--(G4).
\end{example}

\begin{example}
Let $f_{\theta}(\cdot)=\mathcal{N}(\cdot;\theta,\sigma^{2})$, $K_{\epsilon}(x,y)=\mathcal{N}(y;x,\epsilon)$
and $p(\theta,\epsilon)=\mathcal{N}(\theta;0,\delta^{2})\lambda\exp(-\lambda\epsilon)$,
with $q$ and $g$ satisfying (G2)--(G4) and (G10)--(G11). (G1) and
(G5)--(G6) hold in this case and it remains to show that (G9) is satisfied
so we can apply Proposition~\ref{prop:light_tails}. Without loss
of generality assume that $y\geq0$ and note that
\[
\sup_{\vartheta\in B_{r,\theta}}\frac{p_{\theta}(\vartheta)}{p_{\theta}(\theta)}p_{\epsilon}(\varepsilon)\tilde{f}_{\vartheta}^{\varepsilon}(y)\leq\left\{ 2\pi(\sigma^{2}+\varepsilon)\right\} ^{-\frac{1}{2}}\lambda\exp\left[\frac{r\theta}{\delta^{2}}-\lambda\varepsilon-\frac{\left\{ \theta-(y+r)\right\} {}^{2}}{2(\sigma^{2}+\varepsilon)}\right].
\]
With $\theta\in(R,\infty)$ and large enough $R$, we have 
\[
\sup_{\vartheta\in B_{r,\theta}}\sup_{\varepsilon\in[0,\varepsilon_{0}]}\frac{p_{\theta}(\vartheta)}{p_{\theta}(\theta)}p_{\epsilon}(\varepsilon)\tilde{f}_{\vartheta}^{\varepsilon}(y)\leq\left\{ 2\pi(\sigma^{2}+\varepsilon_{0})\right\} ^{-\frac{1}{2}}\lambda\exp\left[\frac{r\theta}{\delta^{2}}-\lambda\varepsilon_{0}-\frac{\left\{ \theta-(y+r)\right\} {}^{2}}{2(\sigma^{2}+\varepsilon_{0})}\right].
\]
Therefore, 
\[
\lim_{R\rightarrow\infty}\sup_{\theta\in B_{R,0}^{\complement},\vartheta\in B_{r,\theta},\varepsilon\in[0,\varepsilon_{0}]}\frac{p_{\theta}(\vartheta)}{p_{\theta}(\theta)}p_{\epsilon}(\varepsilon)\tilde{f}_{\vartheta}^{\varepsilon}(y)=0,
\]
so (G9) is satisfied for any $(\delta^{2},\lambda,\sigma^{2})\in(0,\infty)^{3}$.
\end{example}

\section{Calculations for the example in Section~\ref{eg:geometric}\label{sec:eg1_calculations}}

To obtain $n_{R}=H^{-1}$ calculate 
\[
H=(1-a)\sum_{\theta=1}^{\infty}a^{\theta-1}b^{\theta}=b(1-a)\sum_{\theta=0}^{\infty}(ab)^{\theta}=\frac{b(1-a)}{(1-ab)},
\]
 so $n_{R}=(1-ab)/(b(1-a))$. To bound $n$, we have
\begin{align*}
n & =(1-ab)\left\{ \sum_{\theta=1}^{\infty}(ab)^{\theta-1}\frac{1}{2}\left(\frac{1}{b^{\theta}+b^{\theta-1}-b^{2\theta-1}}+\frac{a}{b^{\theta}+b^{\theta+1}-b^{2\theta+1}}\right)\right\} -(1-ab)\frac{1}{2}\\
 & =\frac{1-ab}{2}\left\{ -1+\sum_{\theta=1}^{\infty}a^{\theta-1}\left(\frac{1}{b+1-b^{\theta}}+\frac{a/b}{1+b-b^{\theta+1}}\right)\right\} ,
\end{align*}
and so both
\[
n\leq\frac{1-ab}{2}\left\{ -1+\sum_{\theta=1}^{\infty}a^{\theta-1}\left(1+\frac{a}{b}\right)\right\} =\frac{1-ab}{2}\left\{ \frac{a+b}{b(1-a)}-1\right\} ,
\]
and
\[
n\geq\frac{1-ab}{2}\left\{ -1+\sum_{\theta=1}^{\infty}a^{\theta-1}\left(\frac{1+a/b}{1+b}\right)\right\} =\frac{1-ab}{2}\left\{ \frac{a+b}{b(1-a)(1+b)}-1\right\} .
\]

\bibliographystyle{biometrika}
\bibliography{abc_ge}

\end{document}